\documentclass[12pt]{article}
\usepackage{amsmath,amssymb,amsthm,bm}
\usepackage{graphicx}
\usepackage{subfigure}
\usepackage[sort]{natbib}
\usepackage{url,xcolor}
\usepackage[colorlinks]{hyperref}
\hypersetup{
    colorlinks,
    linkcolor={red!50!black},
    citecolor={blue!80!black},
    urlcolor={blue!80!black}
}

\newcommand{\be}{\begin{equation} \begin{aligned}}
\newcommand{\ee}{\end{aligned} \end{equation}}
\newcommand{\bq}{\begin{quote}}
\newcommand{\eq}{\end{quote}}
\newcommand{\bp}{\begin{parts}}
\newcommand{\ep}{\end{parts}}
\newcommand{\bqp}{\begin{quote}\begin{parts}}
\newcommand{\epq}{\end{parts}\end{quote}}

\newcommand{\real}{\mathbb{R}}

\newcommand{\rnn}{\real^{n\times n}}

\DeclareMathOperator*{\argmin}{argmin}

\newcommand\ma{\mathbf{A}}
\newcommand\mb{\mathbf{B}}
\newcommand\mc{\mathbf{C}}
\newcommand\md{\mathbf{D}}
\newcommand\me{\mathbf{E}}
\newcommand\mh{\mathbf{H}}
\newcommand\mi{\mathbf{I}}

\newcommand\mr{\mathbf{R}}
\newcommand\ms{\mathbf{S}}

\newcommand\mv{\mathbf{V}}
\newcommand\mx{\mathbf{X}}

\newcommand\mw{\mathbf{W}}

\newcommand\mt{\mathbf{T}}

\newcommand\mpp{\mathbf{P}}
\newcommand\mo{\mathbf{O}}
\newcommand\muu{\mathbf{U}}

\newtheorem{theorem}{Theorem}
\newtheorem{proposition}{Proposition}
\theoremstyle{definition}
\newtheorem{definition}{Definition}

\begin{document}

\bibliographystyle{chicago}

\def\spacingset#1{\renewcommand{\baselinestretch}%
{#1}\small\normalsize} \spacingset{1}



  \title{\bf Vertex nomination between graphs via spectral embedding and quadratic programming}
  \author{Runbing Zheng \\
    Department of Statistics, North Carolina State University\\
    Vince Lyzinski\thanks{VL research is sponsored by the Air Force Research Laboratory and DARPA under agreement number FA8750-20-2-1001. The U.S. Government is authorized to reproduce and distribute reprints for Governmental purposes notwithstanding any copyright notation thereon. The views and conclusions contained herein are those of the authors and should not be interpreted as necessarily representing the official policies or endorsements, either expressed or implied, of the Air Force Research Laboratory and DARPA or the U.S. Government.}
 \\
    Department of Mathematics, University of Maryland \\
    Carey E. Priebe \\
    Department of Applied Mathematics and Statistics, Johns Hopkins University \\
    Minh Tang \\
    Department of Statistics, North Carolina State University}
  \maketitle


\begin{abstract} Given a network and a subset of interesting vertices
whose identities are {\em only partially} known, the vertex nomination
problem seeks to rank the remaining vertices in such a way that the
interesting vertices are ranked at the top of the list.  An important
variant of this problem is vertex nomination in the multiple graphs
setting. Given two graphs $G_1, G_2$ with common vertices and a vertex
of interest $x \in G_1$, we wish to rank the vertices of $G_2$ such
that the vertices most similar to $x$ are ranked at the top of the
list.  The current paper addresses this problem and proposes a method
that first applies adjacency spectral graph embedding to embed the
graphs into a common Euclidean space, and then solves a penalized
linear assignment problem to obtain the nomination lists. Since the
spectral embedding of the graphs are only unique up to orthogonal
transformations, we present two approaches to eliminate this potential
non-identifiability. One approach is based on orthogonal Procrustes
and is applicable when there are enough vertices with known
correspondence between the two graphs. Another approach uses adaptive
point set registration and is applicable when there are few or no
vertices with known correspondence. We show that our nomination scheme
leads to accurate nomination under a generative model for pairs of
random graphs that are approximately low-rank and possibly with
pairwise edge correlations. We illustrate our algorithm's performance
through simulation studies on synthetic data as well as analysis of a
high-school friendship network and analysis of transition rates
between web pages on the Bing search engine.
\end{abstract}

\noindent%

\spacingset{1.5} 




\section{Introduction}
\label{sec:introduction}
Graphs are widely used to model data in various
fields wherein vertices represent entities or objects of interest and
the edges represent pairwise relationships between the vertices. For
example, in social network graphs the vertices represent individuals
with the edges showing the communication between these individuals. Another example is
citation network where the vertices represent articles and the (directed) edges
represent citations between the articles. Finally, in many
neuroscience applications, the vertices represent brain regions of
interest and the edges summarize the inter-connectivity between these
regions.  Due to the prevalence of network data, there has been a
great deal of research done recently in statistical inference on
graphs, including, but not limited to, estimation of graph parameters,
\citep{xu2017rates,lloyd2012random}, one-sample and multi-sample
hypothesis testing \citep{moreno2013network,tang2014semiparametric},
graph clustering and classification
\citep{kudo2005application,zhang2018end,schaeffer2007graph,yin2017local},
and vertex nomination \citep{fishkind2015vertex,yoder2020vertex}.

Vertex nomination is the graph analog of recommender systems for
general tabular data. The simplest and most widely studied variant of
this problem is in the single graph setting wherein, given a network
and a subset of interesting vertices whose identities are partially
known, the task is to identify, using the known interesting vertices,
the remaining vertices of interest. The number of interesting vertices
is, in general, much smaller than the total number of vertices in the
graphs, and vertex nomination algorithms usually seek to output a list
of candidate vertices (that are deemed interesting) with the aim that
the remaining true but unknown vertices of interest are concentrated
near the top of the list.


The vertex nomination problem in the single graph setting appears, at
first blush, to be similar to the more widely studied community
detection problem \citep{duch2005community,newman2006finding,fortunato2010community};
however, there are important conceptual and practical differences
between the two. More specifically, community detection is concerned
with clustering or partitioning {\em all} the vertices of a network
into communities or clusters; a cluster is, roughly speaking, a group of
vertices exhibiting a different connectivity pattern within the
cluster as compared to the connectivity between clusters. 
In contrast,
as we alluded to earlier, vertex nomination is only concerned with
identifying a small subset of vertices of interest, and furthermore,
these vertices of interest might not form a cluster in the usual
sense, e.g., their intraconnectivity does not need to be
qualitatively different from their connectivity to other
"non-interesting" vertices. 
Due to this reason,
vertex nomination is also not the same as local graph clustering
\citep{spielman2013local,yin2017local} as local graph clustering is
concerned with clustering the vertices around the neighbourhood of a
given seed vertex; the unknown vertices of interest might or might not
be included in that neighbourhood.

Before continuing with our exposition we emphasize that, 
similar to recommender systems for tabular data or community detection
in networks, vertex nomination is intrinsically an
unsupervised learning problem. It is universally accepted that these type
of problems generally do not have clearly defined solutions. For example,
given a collection of data points in $\mathbb{R}^{d}$, there can be numerous different
ways to group these data points into clusters and furthermore all of these
clusterings can be valid, and the preference for one clustering over
another clustering depends on the setting and/or objective of the
data analysis at hand. Analogously, given a graph $G$, what
characterizes a vertex $v$ or a collection of vertices $S$ in $G$ as
being interesting is in general not well-defined and could vary
between applications and/or between users. 
Therefore, to present a relevant notion of ``interestingness'', it is usually
assumed that there exists a generative process underlying the observed
graph(s); the process itself could be latent or partially
observed. This is similar to the use of stochastic block models or its
variants in the context of community detection, the use of the
Bradley-Terry model in ranking problems, and the assumption of
low-rank factor models in recommender
systems. Furthermore, even when we assume that the latent generative
model belongs to a collection of models denoted by $\mathcal{M}$, if
$\mathcal{M}$ is not sufficiently restricted then there does not
exists a vertex nomination algorithm that is well-behaved for all
models in $\mathcal{M}$, i.e., for any algorithm $\mathcal{A}$ there
exists a model $M_0 \in \mathcal{M}$ such that the nomination accuracy
of $\mathcal{A}$ when given data generated from $M_0$ is no better
than random guessing; see \citet{lyzinski2019consistent} for a more precise
formulation and statement of this result. 
In light of the above discussion,  in this paper we shall present our
methodology in the context of a generative model for which there are
latent but unobserved features associated with the vertices and, from
these latent features, we can define an appropriate notion of
``interestingness''. Similarly, our real data analysis
examples are motivated by datasets where there are one-to-one
correspondence between (a subset of) the vertices.

Typical applications of vertex nomination include predicting group
membership in social networks \citep{coppersmith2012vertex}, searching
and indexing in databases \citep{levin2017graph}, and identifying
specific type of neurons (e.g., motor neurons) in neuroscience
\citep{fishkind2015vertex}. A sizable number of techniques have been
developed for vertex nomination in the single graph setting, including
methods based on likelihood maximization, Bayesian MCMC, and spectral
decomposition of the adjacency matrices, see
\citet{fishkind2015vertex,yoder2020vertex,lyzinski2016consistency,lee2012bayesian,coppersmith2012vertex,
sun2012comparison} and the references therein. Among these diverse
techniques, the spectral decomposition approach is one of the most
practical as it is computationally efficient and can be scaled to
handle reasonably large and sparse networks.


The current paper focuses on another important, albeit much less
studied, variant of the vertex nomination problem, namely vertex
nomination across graphs. More specifically, given two networks $G_1$
and $G_2$ and a vertex of interest $x$ in the network $G_1$, our task
is to find the corresponding vertex of $x$ in $G_2$ (if it exists). We
emphasize that $G_1$ and $G_2$ need not have the same number of
vertices and that the correspondence between the vertices of $G_1$ and
$G_2$ is largely unknown. 
The following is a concrete practical example of the above problem. 
Consider the pair of high school friendship networks
\citep{patsolic2017vertex}. The first social network $G_1$, having
$156$ vertices, represents a Facebook network, in which two vertices
are adjacent if the pair of individuals are friends on Facebook. The
second social network $G_2$ with $134$ vertices, is created based on
the result of survey, and two vertices are adjacent if the students
report they are friends. There are $82$ students appearing in both
social network $G1$ and $G_2$. Given a student of interest in $G_1$,
our goal is to identify the corresponding student in $G_2$. We may
know the corresponding relationship of some other shared students, and
the information can help us to nominate the student of interest in
$G_2$. This is a typical example, and similar examples could be
constructed including (1) identifying which user in Instagram
corresponding to some specific user in Facebook, 
(2) identifying topics of interest across
graphical knowledge bases \citep{sun2013efficiency} and (3) identifying
structural signal across connectomes \citep{sussman2019matched}.

Two algorithms were recently proposed for this multi-graph vertex
nomination problem. The authors of \cite{agterberg2020vertex} proposed
an algorithm based on spectral graph embedding wherein (1) the graphs
are spectrally embedded into Euclidean space, (2) the embedded points
are aligned via orthogonal Procrustes transformation, (3) the embedded
points are simultaneously clustered via Gaussian mixture modeling
\citep{fraley1998mclust}, and (4) output the candidate vertices using
the resulting clustering. In contrast,
the authors of \cite{patsolic2017vertex} proposed an algorithm based
on seeded graph matching \citep{lyzinski2014seeded}
wherein they graph match induced subgraphs generated around
neighbourhoods of the vertices with known correspondence in each
network.


The proposed algorithm in this paper also uses spectral graph
embedding. It is noted, empirically, that spectral graph embedding
approaches are much faster and more scalable, computationally, as
compared to graph matching approaches. Our approach is similar to
\cite{agterberg2020vertex} in that we also spectrally embed the graphs
into a common Euclidean space. However, in contrast to the Gaussian
mixture modeling of \cite{agterberg2020vertex}, our nomination lists
are based on solving a penalized linear assignment problem, and are
thus not dependent on tuning parameters such as the number of clusters
and the shape/orientation of these clusters, both of which are hard to
tune and could have significant impacts on the ordering in the
nomination lists. Indeed, if the embedding dimension $d$ is
moderately large, a Gaussian mixture model with arbitrary covariance
matrices requires estimation of $O(Kd^2)$ parameters where $K$ is the
maximum number of Gaussian component.


We also prove consistency results about our scheme when the
number of vertices goes to infinity under mild assumptions for a wide
class of popular random graph models. Furthermore, for a class of
random graph model where the edges of the two graphs are pairwise
correlated, we analyze how the magnitude of this correlation 
influences the consistency.


\section{Methodology}
\label{sec:methodology}
We first introduce some notations. Let the two
graphs be denoted as $G_1 = (V_1, E_1)$ and $G_2 = (V_2, E_2)$ where
$V_1$ and $V_2$ are the vertices sets and $E_1$ and $E_2$ are the
edges sets. We shall assume that our graphs are {\em undirected}. We
also partition the vertices sets $V_1$ and $V_2$ as $V_1 = U_1 \cup
J_1$ and $V_2 = U_2 \cup J_2$ where $U_1$ and $U_2$ denote the sets of
{\em shared} vertices between $G_1$ and $G_2$. We emphasize that the
correspondence between the vertices in $U_1$ and $U_2$ are only
partially known, i.e, we further partition $U_1$ and $U_2$ as $U_1 =
\{x\} \cup S_1 \cup W_1$ and $U_2 = \{\sigma(x)\} \cup S_2 \cup W_2$
where $x$ denote the {\em known} vertex of interest in $G_1$;
$\sigma:U_1\to U_2$ is a bijection such that for any $v\in U_1$,
$\sigma(v)$ is its corresponding vertex in $U_2$; $S_1$ and $S_2$ are
the seed sets with $|S_1| = |S_2| = K$, where we already know the
corresponding relationship, i.e., we know the bijection
$\sigma|_{S_1}:S_1\to S_2$; and $W_1$ and $W_2$ are the remaining
vertices of interest. We shall assume that the mapping $\sigma$ from
$W_1$ to $W_2$ is unknown, and furthermore, while the recovery of this correspondence between $W_1$ and
$W_2$ is important, it is also potentially less pressing than the recovery of $\sigma(x)$.
In summary, we only know the correspondence between $S_1$ and
$S_2$, and given the known vertex of interest $x$,
we are interested in finding its {\em unknown} correspondence
$\sigma(x) \in V_2 \setminus S_2$. Our goal is thus to seek a nomination
list of the vertices in $V_2\setminus S_2$ , ranked according to our
confidence in how similar they are to $x$.

The following example can help to understand how the partitioning of the vertices in each graph arise.
We consider the high school friendship network dataset containing two observed graphs. The first graph $G_1$ is extracted from the Facebook social network and the second graph $G_2$ is created based on a survey of the students' friendship. 
These two networks share some students in common.
For the first social network $G_1$, we partition the  students set $V_1$ as the union of the set $U_1$ of shared students and the set $J_1$ of remaining students. Similarly, we also have the partition $V_2 = U_2 \cup J_2$ for $G_2$. 
There is a bijection between the shared students of the two social networks, but we only observe this bijection for a subset of the vertices; we denote these sets as $S_1\subset U_1$ and $S_2\subset U_2$.
We have interest in the unknown corresponding relationship.
For a specific student of interest $x\in U_1\setminus S_1$, our goal is to find which vertex $\sigma(x)$ in the second graph corresponds to the same student.
We then use $W_1 = U_1 \setminus \{x\} \setminus S_1$ and $W_2 = U_2 \setminus \{\sigma(x)\} \setminus S_2$ to denote the remaining students.
More details for the high school friendship network data can be found in Section~\ref{sec:High School}.

We now describe our algorithm for finding $\sigma(x)$. Our algorithm
proceeds in three main steps. In the first step we spectrally embed
each graph into some $d$-dimensional Euclidean space. We next aligned
these embeddings either via solving an orthogonal Procrustes problem
in the case when $K$, the number of seeds vertices, is at least as
large as the embedding dimension $d$, or via solving a point set
registration problem. Finally we solve a quadratic program, using the
pairwise distances between the embedded points, to map each vertex $v
\in V_1 \setminus S_1$ to some {\em ordered subset} $\ell(v) \subset
V_2 \setminus S_2$; $\ell(v)$ serves as the {\em nomination list} of
the vertices in $V_2 \setminus S_2$ most similar to $v$.  We now
describe these steps in detail.

\subsection{Adjacency spectral embedding}

We spectrally embed the graphs by truncating their eigenvalue
decomposition. More specifically, given an $n \times n$ adjacency
matrix $\ma$ of a graph and a positive integer $d$ for the embedding
dimension, we compute
$$\ma = \sum_{i=1}^{n} \lambda_i u_i u_i^{\top},$$
where $|\lambda_1| \geqslant |\lambda_2| \geqslant \dots$ are the
eigenvalues and $u_1, u_2, \dots, u_n$ are the corresponding
eigenvectors. The adjacency spectral embedding of $\ma$ (into
$\mathbb{R}^{d})$ is then the $n \times d$ matrix
$$\hat\mx = [|\lambda_1|^{1/2} u_1, |\lambda_2|^{1/2} u_2, \dots, |\lambda_{d}|^{1/2} u_d].$$
The rows of $\hat\mx$ represent the (low-dimensional) embedding of the
vertices of $\ma$ into $\real^d$.

In practice, we can choose $d$ by looking at the eigenvalues of the adjacency matrix.
 A ubiquitous and principled method is to examine the so-called scree plot and look for ``elbow'' or ``knees'' defining the cut-off between the top (signal) $d$ dimensions and the noise dimensions. 
\cite{zhu2006automatic} provides an automatic dimensionality selection procedure to look for the ``elbow'' by maximizing a profile likelihood function.
 \cite{han2019universal} suggests another universal approach to rank inference via residual subsampling for estimating rank $d$.
 We can also determine $d$ by eigenvalue ratio test \citep{ahn2013eigenvalue} or by empirical distribution of eigenvalues \citep{onatski2010determining}. 

\subsection{Orthogonal Procrustes and point set registration}
\label{sec:alignment} 
We applied adjacency spectral embedding to the graphs $G_1$ and $G_2$,
thereby obtaining the $n \times d$ matrices $\hat\mx_1$ and
$\hat\mx_2$, respectively. As the rows of $\hat\mx_1$ and $\hat\mx_2$
represent the low-dimensional embeddings of the vertices in $G_1$ and
$G_2$, we should expect that, for similar vertices, these rows are
close in $\ell_2$ distance. This is, however, not necessarily the case
as the embeddings $\hat\mx_1$ and $\hat\mx_2$ are not unique, i.e.,
$\hat\mx_1$ and $\hat\mx_2$ are only defined up to some orthogonal
transformations as the eigendecomposition of $\ma_1$ is not, in general, unique.
We thus need to align $\hat\mx_1,\hat\mx_2$ by an orthogonal
transformation $\mathbf{W}$ to eliminate this potential
non-identifiability. We describe two methods for finding
$\hat\mw$. The first method is applicable when $K$, the number of seed
vertices, is larger than or equal to $d$, the embedding dimension; the
second method, which is more general but possibly less accurate, is
applicable for $K < d$, including the important case of $K = 0$. 

\subsubsection{Orthogonal Procrustes (when $K \geqslant d$)}

When the number of seed vertices is greater than or equal to the
embedding dimension, we find the orthogonal transformation $\hat{\mw}$
by solving the orthogonal Procrustes problem
\citep{schonemann1966generalized} to align the seeded vertices across
graphs, i.e.,
$$\hat\mw=\mathop{\arg\min}\limits_{\mw\in\mathbb{O}_d}\|(\hat\mx_1)_{S_1}\mw-(\hat\mx_2)_{S_2}\|_F,$$
where $\mathbb{O}_d$ is the set of all $d \times d$ orthogonal
matrices and $(\hat\mx_1)_{S_1}$ is a $|S_1|\times d$ matrix whose
rows are the rows of $\hat\mx_1$ indexed by the seed set $S_1$. With
$S_2=\{\sigma(v):v\in S_1\}$, $(\hat\mx_2)_{S_2}$ is defined
similarly, so that the $i$th row of $(\hat\mx_1)_{S_1}$ corresponds to
the $i$th row of $(\hat\mx_2)_{S_2}$. The minimizer $\hat{\mathbf{W}}$
has an explicit solution as $\hat{\mathbf{W}} = \mathbf{U}
\mathbf{V}^{\top}$ where $\mathbf{U} \mathbf{D} \mathbf{V}^{\top}$ is
the singular value decomposition of the $d \times d$ matrix
$(\hat\mx_2)_{S_2}^{\top} (\hat\mx_1)_{S_1}$.  After finding $\hat\mw$ we set
$\tilde\mx_1=\hat\mx_1\hat\mw$.

\subsubsection{Adaptive rigid point set registration (when $K < d$)}

Even when we do not have enough seed vertices or even no seed set ($K
= 0$), as long as we have a reasonable number of vertices that are
shared between the two graphs, we can apply the coherent point drift algorithm in \cite{myronenko2010point} to align
$\hat\mx_1$ and $\hat\mx_2$.
The algorithm in \cite{myronenko2010point} finds an affine
transformation to best align the centroids of the clusters of
$\hat\mx_1$ to the centroids of the clusters of $\hat\mx_2$. More
specifically, given a $n \times d$ matrix $\hat\mx_1$ and $m \times d$
matrix $\hat\mx_2$, we find $s \in \mathbb{R}$, $t \in \mathbb{R}^{d}$
and $\mathbf{W} \in \mathbb{O}_d$ that minimize the following
objective function
\begin{equation*}
\begin{aligned} &Q\left(\mathbf{W}, \mathbf{t}, s,
\sigma^{2}\right)=\frac{1}{2 \sigma^{2}} \sum_{i=1}^n\sum_{j=1}^{m}
P\left(i\, \big|\, (\hat\mx_2)_{j}\right)\|(\hat\mx_2)_{j}-s
\mathbf{W}^\top(\hat\mx_1)_{i}-\mathbf{t}\|^{2} +\frac{N_{\mathbf{P}}
d}{2} \log \sigma^{2},
\end{aligned}
\end{equation*} where $(\hat\mx_1)_i$ represents the $i$th row of
$\hat\mx_1$, i.e., the $i$th vertex's embedding of $G_1$;
$(\hat\mx_2)_j$ is defined similarly. 
Here
$N_{\mathbf{P}}=\sum_{i=1}^{n} \sum_{j=1}^{m} P\left(i \, \big|\,
(\hat\mx_2)_{j}\right)$ is a normalizing constant, with $P\left(i \, \big|\,
(\hat\mx_2)_{j}\right)$ the correspondence probability between two
vertices' embeddings $(\hat\mx_1)_i$ and $(\hat\mx_2)_j$, defined as
the posterior probability of the centroid given the vertex's embedding
$(\hat\mx_2)_j$, i.e.,

$$P\left(i \, \big|\, (\hat\mx_2)_{j}\right)=\frac{\exp\Bigl(-\tfrac{1}{2}\|((\hat\mx_2)_{j}-s \mathbf{W}^\top(\hat\mx_1)_{i}-\mathbf{t})/\sigma\|^{2}\Bigr)}{c + \sum_{k=1}^{n} \exp\Bigl(-\tfrac{1}{2} \|((\hat\mx_2)_{j}-s \mathbf{W}^\top(\hat\mx_1)_{k}-\mathbf{t})/\sigma\|^2\Bigr)},$$
for some constant $c$ (where $c$ is a function of the model parameters defined above; see \cite{myronenko2010point}). The minimization of $Q$ is done via an
EM-algorithm. For more details, please refer to
\cite{myronenko2010point}.

The resulting minimizer $(\hat{s}, \hat{\bm{t}}, \hat{\mathbf{W}})$
yield an affine transformation $\mathcal{T}$ of $\hat\mx_1$ via
$\mathcal{T}((\hat\mx_1)_{i}) = s \hat\mw^\top (\hat\mx_1)_i + \bm{t}$.  We
note, however, that in the context of our current work, the alignment
of $\hat\mx_1$ and $\hat\mx_2$ does not require the scaling $s$ and
the translation $t$. We thus make a few minor adjustments to the EM
algorithm in \cite{myronenko2010point}. In particular, 1) we always
set $s=1$ and $\mathbf{t} = \bm{0}$; 2) we iteratively update
$\mathbf{W}$ via $\mathbf{W}^\top=\mathbf{U V}^{T}$ instead of
$\mathbf{W}^\top=\mathbf{U C V}^{T}$; 3) we initialize $\mathbf{W}$ as a
diagonal matrix with $1$ or $-1$ diagonal elements so that the initial
error in the EM approach $\sigma^{2}=\frac{1}{d n m} \sum_{i=1}^{n}
\sum_{j=1}^{n}\left\|(\hat\mx_1 \mw)_i-(\hat\mx_2)_j\right\|^{2}$
is as small as possible. Once we get the final orthogonal matrix
$\hat\mw\in\real^{d\times d}$, we set $\tilde\mx_1=\hat\mx_1
\hat\mw$.


\subsection{Quadratic program}
\label{sec:quadratic}
We now formulate a quadratic program to
find, for each vertex $v \in V_1$, a collection of vertices $\ell(v)
\subset V_2\setminus S_2$ that are ``most similar'' to $v$. Here
similarity between $v \in V_1$ and $u \in V_2$ is measured in terms of
the Euclidean distances between their embeddings. In other words,
given a {\em query} vertex $v$ in the first graph, 
our proposed algorithm outputs a nomination list $\ell(v)$ of vertices
in the second graph that are most similar to $v$; these vertices are deemed
``interesting'' in the context of the query vertex $v$. 

Our quadratic program is described as follows.
Given the aligned embeddings $\tilde\mx_1$ and $\hat\mx_2$, we find
$\hat{\mathbf{D}}$ to minimize the following objective function
\begin{equation}
  \label{eq:objective}
  \hat\md=\mathop{\arg\min}\limits_{\md\in\real^{n\times m}}
  \sum_{i=1}^n\sum_{j=1}^m\|(\tilde\mx_1)_i-(\hat\mx_2)_j\|_2\cdot\md_{i,j}+\lambda\|\md\|_F^2,
\end{equation}
subject to the constraints that
\begin{equation*}
\begin{aligned}
	(\mathrm{i}) \quad &\,\sum_{j=1}^m\md_{i,j}=m,\text{ for all }1\leqslant i\leqslant n,\\
	(\mathrm{ii}) \quad &\,\sum_{i=1}^n\md_{i,j}=n,\text{ for all }1\leqslant j\leqslant m,\\
	(\mathrm{iii}) \quad & \,\md_{i,j}\geqslant0,\text{ for all }1\leqslant i\leqslant n,1\leqslant j\leqslant m,\\
	(\mathrm{iv}) \quad &\,\md_{(S_1)_k,(S_2)_k}=\min\{n,m\},\text{ for all }1\leqslant k\leqslant K.
\end{aligned}
\end{equation*}
Here $\|(\hat\mx_1)_i-(\hat\mx_2)_j\|$ is the Euclidean distance
between the $i$th vertex's embedding of $G_1$ and the $j$th vertex's
embedding of $G_2$, $\lambda>0$ is a penalty parameter, and
$(S_1)_k,(S_2)_k$ represents the index of the $k$th seed vertex in
$G_1$ and $G_2$, respectively.

The motivation behind solving the above optimization problem is as
follows. The constraints on $\mathbf{D}$ state that (1) each vertex $i
\in V_1$ is mapped to some collection of vertices in $j \in V_2$,
namely those for which $\hat\md_{ij} > 0$ (constraint iii.); (2) since for any $i \in
V_1$, $\sum_{j} \hat\md_{ij} = m$ where $m = |V_2|$, larger values of
$\hat\md_{ij}$ indicates more "similarity" between $i \in V_1$ and $j
\in V_2$ (constraints i. and ii.); and (3) a seed vertex $s \in S_1$ will get mapped to its
unique correspondence $\sigma(s) \in S_2$ (constraint iv.). 

We now consider the
objective function. 
The first part of the objective function indicates
that the similarity between the $i$th vertex in $V_1$ and the $j$th
vertex in $V_2$ is based on the Euclidean distance
$\|(\tilde{\mathbf{X}}_1)_i - (\hat\mx_2)_j\|$, i.e., larger distance
should lead to smaller $\hat\md_{ij}$. We can then consider, for each
$i \in V_1$, the nomination list for $i$ as being the vertices in
$V_2\setminus S_2$ arranged according to decreasing values of
$\{\hat\md_{ij}\}_{j \in V_2\setminus S_2}$. The second part of the
objective function, i.e., the penalty term $\lambda
\|\mathbf{D}\|_{F}^2$, is to discourage sparsity 
of $\hat\md$.
More specifically, removal of the penalty term $\lambda
\|\mathbf{D}\|_{F}^2$ leads to a linear programming problem for which
the minimizer $\hat\md$ may lie on the boundary of the feasibility
region, i.e., the elements of $\hat\md$ take values only in
$\{0,m,n\}$. If $\hat\md_{i,j} = m$, then the $i$th vertex in $V_1$ is
mapped to the $j$th vertex in $V_2$ and if $\hat\md_{i,j} = n$ then
the $j$th vertex in $V_2$ is mapped to the $i$th vertex in $V_1$. This
gives a nomination list with a single candidate. This type of nomination list,
when accurate, can significantly reduces the burden of post-processing and checking/verifying
multiple candidates. However, it is also likely to be non-robust. 
By adding the penalty term, we encourage the elements in each row of
$\hat\md$ to be more uniform since, for any vector $x \in
\mathbb{R}^{m}$, $\|x\|_{\ell_2} \geqslant m^{-1/2}\|x\|_{\ell_1}$ with
equality if and only if all the elements of $x$ are the same. We
note that the optimization problem in Eq.~\eqref{eq:objective} is
analogous to the {\em quadratically regularized} optimal transport problem between the 
point masses induced by $\tilde{\mathbf{X}}_1$ and
$\hat\mx_2$. %

The resulting optimization problem is a quadratic program with linear
constraints, and the coefficient matrix of the quadratic term is positive
definite. 
Thus, for a fixed $\lambda > 0$, the
optimization function is strongly convex and hence there exists a {\em
  unique} global minimizer $\hat{\mathbf{D}}$ for any given
$\tilde{\mathbf{X}}_1$ and $\hat{\mathbf{X}}_2$. In particular, for a fixed $\lambda
> 0$, the solution $\hat{\mathbf{D}}_{\lambda}$ of Eq.~\eqref{eq:objective} is of the
form \citep{blondel}
$$\hat{\mathbf{D}}_{\lambda}(i,j) = \frac{1}{\lambda} \Bigl[\|(\tilde{\mathbf{X}}_1)_{i} -
(\hat{\mathbf{X}}_2)_j\| - \hat{\alpha}_i - \hat{\beta}_j\Bigr]_{+},$$
where $[z]_+ = \max\{z, 0\}$, and $\hat{\bm{\alpha}} = (\hat{\alpha}_1, \dots, \hat{\alpha}_n) $
and $\hat{\bm{\beta}} = (\hat{\beta}_1, \dots, \hat{\beta}_m)$ solve the
unconstrained dual problem
$$\max_{\substack{\bm{\alpha} \in \mathbb{R}^{n}}, \bm{\beta} \in \mathbb{R}^{m}} n \bm{\alpha}^{\top} \bm{1} + m
\bm{\beta}^{\top} \bm{1} - \frac{1}{2 \lambda} \sum_{i=1}^{n}
\sum_{j=1}^{m} \Bigl[\|(\tilde{\mathbf{X}}_1)_{i} -
(\hat{\mathbf{X}}_2)_j\| - \alpha_i - \beta_j\Bigr]_{+}.$$
The optimal solution can be found, theoretically, in
polynomial time using the ellipsoid algorithm of \cite{kozlov1980polynomial}. 
In practice we use the Gurobi solver
\citep{gurobi2018gurobi}, which is based on an interior-point
algorithm. 

There is no simple and universal approach for choosing $\lambda$ in practice. If there are seed vertices, then we can do cross-validation by leaving out a subset (or all) of the seed vertices and choose $\lambda$ for which the nomination on the seed vertices has smallest MRR or MNR. In general, as vertex nomination is an unsupervised learning problem, the issue of choosing tuning parameter is, in a sense, unsolved (at least without additional information such as having seed vertices or a known collection of pair of vertices that should not be matched together).
Nevertheless, we note that, empirically, the nomination list found by our
algorithm is {\em not} overly sensitive to the choice of the
$\lambda$, as $\lambda$ mainly influences the
magnitudes of the elements in $\hat\md$ but does not changes their relative ordering too much.

The quadratic program considered in Eq.~\eqref{eq:objective} is motivated by the two-dimensional linear assignment problem. The most obvious formulation of the current vertex nomination problem to three or more graphs will lead to an optimization problem that is similar to the multi-dimensional assignment problem \citep{pierskalla1968multidimensional}, which is NP-hard. We thus leave the detailed study of multi-sample vertex nomination for future work.

\subsection{Computational complexity}
We now describe the computational complexity for our proposed
algorithm. Firstly, the embedding step is
roughly $O((n^2+m^2)d)$ where $d$ is the embedding dimension
\citep{van1990partial} and $n$ and $m$ are the number of vertices in
the first and second graph, respectively. The orthogonal Procrustes is
roughly $O(\min\{m,n\} d^2)$. The adaptive point set registration can be done
iteratively with each iteration having computational
complexity of $O(2^d nmd)$, and generally speaking $\ell \leq 50$
iterations suffice from empirical observations. The solution of the
quadratic programming problem can also be solved iteratively with each
iteration having a complexity of $O(nm)$ operations; the number of iterations is (empirically) generally bounded and is also independent of the data
dimension. As such, the expected empirical complexity of the quadratic
programming step is also $O(nm)$ \citep{dessein2018regularized}.

Empirically, we record the running time of the simulation experiments in Section~\ref{sec:simulation} on $n=300$ vertices and $n=1000$ vertices, respectively (here we set $m=n$).
For our algorithm with
orthogonal Procrustes, the {\bf total} running time for the
experiments on $n = 300$ vertices for various choices of $\rho$ (Figure~\ref{rho-RDPG} or Figure~\ref{rho-SBM}) is
about $2$ hours while the total running time for $n = 1000$ vertices (Figure~\ref{rho-RDPG_n=1000} or Figure~\ref{rho-SBM_n=1000} in Appendix~\ref{sec:add res})
is about $23$ hours. The running time of our algorithm with the adaptive
point set registration procedure also has the similar
ratio. In summary, we see that the running time of our algorithm does
scale (approximately) quadratically with $n$, the number of vertices, which is consistent with the above analysis.

\section{Theoretical Results}
\label{sec:theorey}
We now investigate the theoretical properties of our proposed
algorithm. For simplicity we will only consider the case where $G_1$
and $G_2$ have the same number of vertices; the analysis presented
here will also extend to the case when $G_1$ and $G_2$ have different
number of vertices provided that the number of common vertices is sufficiently large.
We first formulate a generative model for generating pairs of random graphs $(G_1, G_2)$ with underlying
latent correspondence $\sigma$ between the vertices $V_1$ of $G_1$ and
$V_2$ of $G_2$. 
We then show consistency results about our scheme under the generative model and
analyze the accuracy of the nomination list obtained by our algorithm.
We next investigate the impact of having seed vertices. In
particular, we propose a procedure for re-ranking the nomination lists
in the presence of seed vertices. The empirical
results in Section~\ref{sec:experiment} indicates that 
if the number of seed vertices is not too small, then this optional
re-ranking step outputs an improved nomination list $\ell(x)$ compared
to the nomination list obtained directly from the quadratic program solution.

Our generative model for pairs of random
graphs depends on the following notion of the generalized random dot product graphs \citep{rubin2017statistical,young2007random}.
\begin{definition}[Generalized random dot product graphs]
  \label{def:grdpg}
  Let $d \geq 1$ be given and let $\mathcal{X}$ be  a  subset  of $\mathbb{R}^{d}$ such that
  $x^{\top} \mathbf{I}_{p,q} y \in [0,1]$. Here $\mathbf{I}_{p,q}$ is
  a $d \times d$ diagonal matrix with diagonal entries containing $p$
  ``+1'' and $q$ ``-1'' for integers $p, q \geq 0$, $p + q = d$. For a given $n \geq 1$, let
  $\mathbf{X}$ be a $n \times d$ matrix with rows $X_i \in
  \mathcal{X}$ for $i =1,2,\dots, n$. A random graph $G$ is said to be 
  an instance of a generalized random dot product graph with latent positions
  $\mathbf{X}$ if the adjacency matrix $\mathbf{A}$ of $G$ is a
  symmetric matrix whose upper triangular entries
  $\{\mathbf{A}(i,j)\}_{i \leq j}$ are independent
  Bernoulli random variables with 
  $$\mathbf{A}(i,j) \sim \mathrm{Bernoulli}(X_i^{\top} \mathbf{I}_{p,q} X_j).$$
  We use GRDPG($\mpp$) to represent such graph, where $\mpp=\mx\mi_{p,q}\mx^\top$.
\end{definition}

Generalized random dot product graphs are a special case of latent
position graphs or graphons
\citep{Hoff2002,diaconis08:_graph_limit_exchan_random_graph,lovasz12:_large}.
In the general latent position graph model, each vertex $v_i$ is
associated with a latent or unobserved vector $X_i$ and, given the
collection of latent vectors $\{X_i\}$, the edges are conditionally
independent Bernoulli random variables with $\mathbb{P}[v_i \sim v_j]
= \kappa(X_i, X_j)$ for some {\em symmetric} link function $\kappa$.
Generalized random dot product graphs can be used to model any latent
position graphs where the link function $\kappa$ is
finite-dimensional, i.e., $\kappa$ is such that for any $n$ and for
any collection of latent vectors $\{X_i\}_{i=1}^{n}$, the $n \times n$
matrix $\mathbf{P}$ with $\mathbf{P}(i,j) = \kappa(X_i, X_j)$ has rank
at most $d$ for some {\em arbitrary} but fixed $d$ not depending on
$n$. Indeed, as $\mathbf{P}$ is a symmetric matrix with rank at most
$d$, $\mathbf{P}$ has an eigendecomposition as $\mathbf{P} =
\mathbf{U} \bm{\Lambda} \mathbf{U}^{\top}$. Hence, taking
$\mathbf{X} = \mathbf{U} |\bm{\Lambda}|^{1/2}$ and letting $p$ and $q$
be the number of positive and negative eigenvalues of $\mathbf{P}$, we
obtained a representation of $\mathbf{P}$ as a GRDPG.

Generalized random dot product graphs include, as
special cases, the popular class of stochastic block model graphs and their
degree-corrected and mixed-membership variants
\citep{holland,karrer2011stochastic,Airoldi2008}.
\begin{definition} (Stochastic block model random graphs).  We say a
random graph $G$ with adjacency matrix $\ma$ is distributed as a
stochastic block model random graph with parameters $L,b,\mb$ if
\begin{enumerate}
\item The vertex set $V$ of $G$ is partitioned into $L$ blocks, $V=V_{1} \cup V_{2} \cup \cdots \cup V_{L}$.
\item The function $b$ is a mapping from $V$ to $\{1,\cdots,L\}$ with
  $b(i)$ denoting the block label of vertex
  $i \in V$.
\item The matrix $\mb\in [0,1]^{L\times L}$ is a symmetric
   matrix of block probabilities. More specifically, given $b$, the entries $\ma(i,j)$ for $i \leq j$ are {\em
    conditionally independent} Bernoulli random variables with $\ma(i,j) \sim \mathrm{Bernoulli}\left(\mb_{b(i),b(j)}\right)$.
\end{enumerate}
We denote a stochastic block model graph as $G \sim \mathrm{SBM}(L, b, \mathbf{B})$.
\end{definition} 
A $L$-blocks stochastic block model graph \citep{holland} corresponds to a GRDPG where $\mathcal{X}$ is a mixture
of $L$ point masses. Similarly, a $L$-blocks degree-corrected
stochastic block model and a $L$-blocks mixed-membership stochastic block model
correspond to a GRDPG where $\mathcal{X}$ is supported on a mixture of
$L$ rays and on a convex hull of $L$ points, respectively.  See
\cite{rubin2017statistical} for a more detailed description of these relationships. The perspective of representing a mixed membership SBM as a GRDPG also gives us a general approach to constructing the domain $\mathcal{X}$ for a GRDPG in Definition~\ref{def:grdpg}. More specifically, we can construct $\mathcal{X}$ by first finding
 some collection of $K$ vectors
$\mathcal{S} = \{X_1, X_2, \dots \}$ for which $X_i^{\top} \mi_{p, q} X_j
      \in [0,1]$ for all $X_i, X_j \in \mathcal{S}$ and then construct
      $\mathcal{X}$ as the convex hull of the $\mathcal{S}$. 

Our generative model for pairs of random
graphs extends the $\mathrm{GRDPG}$ model for the single graph setting to
the setting of two graphs that share a common set of vertices with 
edges that are possibly correlated.
\begin{definition}[$\rho$-correlated GRDPG]
  \label{def:rho-correlated}
  Assume the notation in Defintion~\ref{def:grdpg}. Let $\rho \in
  [-1,1]$ be given. A {\em pair} of
  random graphs $(G_1,G_2)$ is said to be an instance of a
  $\rho$-correlated generalized random dot product graphs with latent
  positions $\mathbf{X}$ if the pair of adjacency matrices $\ma_1$ and
  $\ma_2$ satisfy the following conditions. 
\begin{enumerate}
\item Marginally $G_1 \sim \mathrm{GRDPG}(\mathbf{P})$  {\em and} $G_2
  \sim \mathrm{GRDPG}(\mathbf{P})$.
\item Given $\mathbf{P}$, the bivariate random variables $\{\ma_1(i,j),\ma_2(i,j)\}_{1\leqslant
    i<j\leqslant n}$ are collectively independent and
  $$\mathrm{corr}(\ma_1(i,j),\ma_2(i,j))=\rho,$$
  for any $1\leqslant i<j\leqslant n$.
\end{enumerate}
\end{definition} The correlation $\rho$ in
Definition~\ref{def:rho-correlated} induces a notion of correspondence
between the vertices in $G_1$ and $G_2$. More specifically, if $\rho =
0$ then for any two arbitrary pairs of vertices $(i, j)$ and $(k,
\ell)$, the edges $\mathbf{A}_1(i,j)$ and $\mathbf{A}_2(k,\ell)$ are
independent.  In contrast, if $\rho \not = 0$ then
$\mathbf{A}_{1}(i,j)$ and $\mathbf{A}_2(k,\ell)$ are independent if
and only if $\{i,j\} \not = \{k, \ell\},\{i,j\} \not = \{\ell,k\}$. Now suppose that $\rho \not
= 0$. Then given a vertex of interest $v \in G_1$, we can define the
true correspondence of $v$ in $G_2$ as the {\em unique} vertex
$\sigma(v) \in G_2$ such that the edges $\mathbf{A}_1(v,u)$ and
$\mathbf{A}_2(\sigma(v), \sigma(u))$ are correlated {\em for all} $u
\in G_1$. In summary, if $(G_1, G_2)$ is a pair of $\rho$-correlated
GRDPG graphs with $\rho \neq 0$ then there exists a {\em canonical}
correspondence between the vertices of $G_1$ and the vertices of
$G_2$. We use this correspondence to define our notion of
``interestingness'' when evaluating the proposed methodology on graphs
generated from the $\rho$-GRDPG model, i.e., given a query vertex
$v \in G_1$, we wish to find $\sigma(v) \in G_2$. See
\cite{lyzinski2014seeded,patsolic2017vertex,agterberg2020vertex} for
further discussion of the relationship between $\rho$ and its induced correspondence in graph matching and vertex nomination
problems.

We now state our first theoretical result. The following result
provides an error bound for the $2 \to \infty$ norm difference between
the adjacency spectral embeddings of $G_1$ and $G_2$;  
given a $n \times p$ matrix $\mathbf{M}$ with rows $M_i$, the $2 \to \infty$ norm of $\mathbf{M}$
is the maximum $\ell_2$ norm of the rows $M_i$, i.e.,
$$\|\mathbf{M}\|_{2 \to \infty} = \max_{\|\bm{x}\| = 1}
\|\mathbf{M} \bm{x}\|_{\infty} = \max_{i=1,\dots,n} \|M_i\|.$$
The main feature of the following $2 \to \infty$ bound is that it is monotone decreasing
in both $\rho>0$ and $n$, i.e., larger correlation and/or larger number
of vertices in each graph lead to smaller error bound that holds {\em
  uniformly} for all vertices of $G_1$ and $G_2$. We emphasize that
previous bounds for
$\min_{\mathbf{W}\in\mathbb{O}_{d}} \|\hat{\mathbf{X}}_1 \mathbf{W}  -
\mathbf{X}\|_{2 \to \infty}$ and $\min_{\mathbf{W}\in\mathbb{O}_{d}} \|\hat{\mathbf{X}}_2
\mathbf{W} - \mathbf{X}\|_{2 \to \infty}$ (see e.g., \citet{rubin2017statistical})  do not depend on the
correlation $\rho$, and thus will lead to a potentially sub-optimal bound of the form
$$\min_{\mathbf{W}\in\mathbb{O}_{d}} \|\hat\mx_1 \mathbf{W} - \hat\mx_2\|_{2 \to \infty}
= O_{p}(n^{-1/2}).$$
Below, we will characterise the error between asymptotic latent position estimations under the assumption that $\mpp=\gamma\cdot\mx\mi_{p,q}\mx^\top$ where $\gamma$ is a sparsity factor.
\begin{theorem}
  \label{thm:main}
  Let $(G_1,G_2)\sim\rho$-$GRDPG(\mpp)$, where
$\mpp=\gamma\cdot\mx\mi_{p,q}\mx^\top\in[0,1]^{n\times n}$ is symmetric with $\mathrm{rank}(\mathbf{P})=p+q
= d$. Suppose that 1) $\max_{i} \sum_{j} \mpp_{i,j}(1-\mpp_{i,j}) \geqslant
C \log^{4}n$ for some universal constant $C$; 2) the latent positions
satisfy $\frac{1}{n} \sum_{i=1}^{n} X_{i} X_{i}^{\top}\mi_{p,q}
 \stackrel{\mathrm{a.s.}}{\longrightarrow}  \Gamma$ as $n \rightarrow \infty$. Here $\Gamma$ is  a
{\em fixed} $ d \times d$ matrix not depending on $n$. Denote by $\ma_1,\ma_2\in\rnn$ the adjacency matrices for 
$G_1$ and $G_2$, respectively. Let $\hat\mx_1$ and $\hat\mx_2$ be the
adjacency spectral embedding of $\ma_1$ and $\ma_2$ into
$\mathbb{R}^{d}$, respectively. 
Then there exists a constant $c>0$ such that
$$\min_{\mw\in\mathbb{O}_{d}}\left\|\hat\mx_1\mw-\hat\mx_2\right\|_{2\to\infty}=(1-\rho)^{1/2}\cdot O_p\left(n^{-1/2}\right)+O_p\left((\log n)^{2c}n^{-1}\gamma^{-1/2}\right).$$
\end{theorem}
Condition (1) in the statement of Theorem~\ref{thm:main} is a
condition on the {\em sparsity} of the graph. In particular, Condition
(1) is satisfied provided that the {\em maximum degree} of
$\mathbf{P}$ is of order $\Omega(\log^{4}n)$. Condition (2) is a
condition on the homogeneity of the latent positions
$\mathbf{X}$, i.e., as $n
\rightarrow \infty$, the latent positions are sufficiently homogeneous
so that their sample second moment matrix $\tfrac{1}{n} \sum_{i=1}^n X_i
X_i^{\top} \mathbf{I}_{p,q}$ converges. Assuming the above conditions
are satisfied, Theorem~\ref{thm:main} then implies the existence of an orthogonal $\mw_*$
such that for $\rho$ sufficiently bounded away from $1$, $$\|\mw_*^\top(\hat\mx_1)_{i}-(\hat\mx_2)_{i}\| =(1
- \rho)^{1/2} \cdot O_{p}(n^{-1/2} ), \quad \text{for all $i = 1,2,\dots,n.$}$$
Suppose we are now given a collection of seed vertices $S$ with $|S|\geq d$, the
embedding dimension. Suppose furthermore that the $|S| \times d$ matrices $(\hat\mx_1)_{S}$ and
$(\hat\mx_2)_{S}$ are both of full-column rank, i.e., $(\hat\mx_1)_{S}$
and $(\hat\mx_2)_{S}$ each contains $d$ linearly independent
columns. Then by solving the orthogonal Procrustes problem
$\min_{\mathbf{W}\in\mathbb{O}_{d}} \|(\hat\mx_1)_{S}\mw - (\hat\mx_2)_{S}\|_{F}$, we
will obtain an estimate $\hat \mw$ of $\mathbf{W}$ that still
satisfies the claim in Theorem~\ref{thm:main}. Indeed, we have
$$\|\hat \mw^\top (\hat\mx_1)_{j} - (\hat\mx_2)_{j}\| \leqslant
\|\mw_*^\top(\hat\mx_1)_{j} - (\hat\mx_2)_{j}\| = (1 -
\rho)^{1/2} \cdot O_{p}(n^{-1/2}), \quad \text{for
  all $j \in S$}.$$
Now for each $i \not \in S$, $(\hat\mx_1)_i$ and $(\hat\mx_2)_i$ can be
written as a linear combination of $(\hat\mx_1)_{j}, j \in S$ and
$(\hat\mx_2)_{j}, j \in S$, respectively. Hence, by the
triangle inequality for vector norms, 
$$\|\hat\mx_1\hat\mw  - \hat\mx_2\|_{2 \to \infty}
\leq (1 - \rho)^{1/2} \cdot O_{p}(n^{-1/2}).$$
In summary, estimates of $\mw_*$ using either orthogonal
Procrustes or point set registration will yield a transformation
$\tilde\mx_1$ of $\hat\mx_1$ whose rows are {\em
  uniformly} close to the rows of $\hat\mx_2$. Thus, for the optimization
problem in Eq.~\eqref{eq:objective}, the costs
$\hat{c}_{ij} = \|(\tilde\mx_1)_i-(\hat\mx_2)_j\|$ will be
close to $c_{ij} = \|X_i - X_{j}\|$, i.e.,
\begin{equation}
  \label{eq:bound_main}
\max_{i \not = j} |\hat{c}_{ij} - c_{ij}| = O_{P}(n^{-1/2}), \quad
\max_{i} |\hat{c}_{ii} - c_{ii}| = (1 - \rho)^{1/2} O_{P}(n^{-1/2}).
\end{equation}

We now consider the implications of Theorem~\ref{thm:main} on the
solution of the 
optimization problem in Eq.~\eqref{eq:objective}. Suppose first that $\lambda =
0$ so that the optimization problem in Eq.\eqref{eq:objective} reduces
to a linear programming. Suppose also that $c_{ij} > 0$ whenever $i \not = j$ and $c_{ii} = 0$,
i.e., the latent positions $\{X_i\}$ are unique. Then as $n \rightarrow \infty$, by the above
bounds for $|\hat{c}_{ij} - c_{ij}|$, we have
$\hat\md_{i,i}=n$ for all $1\leqslant
i\leqslant n$ and hence, for any vertex of interest $x$ in $G_1$, our algorithm will give the
nomination list with the true $\sigma(x)$ at the top of the list. 



We next consider the case where $\lambda > 0$. Define $\mathbf{C}$ and $\hat{\mathbf{C}}$
as the $n \times n$ matrices whose elements are $c_{ij}$ and
$\hat{c}_{ij}$ respectively.  Let
$$\mathcal{D}=\{\md:\md\in\real^{n\times
  n}_+,\md1_n=(n,\cdots,n)^\top,\md^\top
1_n=(n,\cdots,n)^\top\}.$$ The optimization problem in
Eq.~\eqref{eq:objective} is then equivalent to
$$\argmin_{\mathbf{D} \in \mathcal{D}} \langle \hat{\mathbf{C}}, \mathbf{D}
\rangle + \lambda \|\mathbf{D}\|_{F}^2 = \argmin_{\mathbf{D} \in
  \mathcal{D}} \|\mathbf{D} + \frac{1}{2 \lambda}
\hat\mc\|_{F}^2, \quad \lambda > 0.$$
Let $\hat{\mathbf{D}}_{\lambda}$ be the {\em unique} solution of the above
problem. Then $\hat{\mathbf{D}}_{\lambda}$ is the {\em projection} of $-\tfrac{1}{2
  \lambda} \hat{\mathbf{C}}$ onto the convex set $\mathcal{D}$. Now consider
the solution of Eq.~\eqref{eq:objective} where we replaced
$\hat{\mathbf{C}}$ by $\mathbf{C}$ and denote that unique solution as $\mathbf{D}_{\lambda}$.
From Eq.~(\ref{eq:bound_main}) we have $\|\hat\mc-\mc\|_F^2=
\sum_{ij} (\hat{c}_{ij} - c_{ij})^{2} = O_p(n)$ and hence
$\|\hat\mc-\mc\|_F = O_{P}(n^{1/2})$. 
Next recall that the projection in Frobenius norm onto convex sets is
$1$-Lipschitz. We therefore have
$$\|\mathbf{D}_{\lambda} - \hat{\mathbf{D}}_{\lambda} \|_{F} \leq \frac{1}{2 \lambda}
\|\hat{\mathbf{C}} - \mathbf{C}\|_{F} = \frac{1}{2 \lambda}\cdot O_{P}(n^{1/2}).$$
Since $\|\mathbf{D}\|_{F} \geq n$ for all $\mathbf{D} \in
\mathcal{D}$, we see that $\|\hat{\mathbf{D}}_{\lambda} -
\mathbf{D}_{\lambda}\|_{F} = o(\|\mathbf{D}_{\lambda}\|)$ for all
$\lambda \gg n^{-1/2}.$ In other words, provided that $\lambda$  is
not too small, the solutions of
Eq.~\eqref{eq:objective} using the true cost matrix $\mathbf{C}$ and
using the estimated cost matrix $\hat{\mathbf{C}}$ are close, i.e.,
the relative error between $\mathbf{D}_{\lambda}$ and $\hat{\mathbf{D}}_{\lambda}$
could be made arbitrarily small for sufficiently large $n$. 



Finally, we consider how the solution $\mathbf{D}_{\lambda}$ using the
true cost matrix $\mathbf{C}$ will look like as $\lambda \rightarrow 0$. Let
$P_{0}$ be the linear programming problem 
$\min_{\mathbf{D} \in \mathcal{D}} \langle \mathbf{C}, \mathbf{D} \rangle$.
Suppose $n$ is now fixed. Then if $\lambda \rightarrow 0$,
$\mathbf{D}_{\lambda}$ will converge to the optimal solution with {\em minimum Frobenius norm} among the set of all
optimal solutions of $P_0$, i.e., letting $\xi_*$ be the minimum
objective value of $P_0$, we have
$$\mathbf{D}_{\lambda} \longrightarrow \argmin_{\mathbf{D} \in
  \mathcal{D}} \{ \|\mathbf{D}\|_{F} \colon \langle \mathbf{C}, \mathbf{D} \rangle
= \xi_*\}.$$
We note that there could be multiple solutions of $\{\mathbf{D} \in \mathcal{D} \colon \langle
\mathbf{C}, \mathbf{D} \rangle = \xi_*\}$. Nevertheless, in the event
that $P_0$ has a {\em unique} minimizer, then since $c_{ii} = 0$ for
all $i$, this unique minimizer will be given by $\mathbf{D}_* =
\mathrm{diag}(n,n,\dots,n)$. Therefore, by the continuity of the optimization problem, there
exists a $\lambda > 0$ such that $\mathbf{D}_{\lambda} =
\mathbf{D}_*$. The previous bound for $\|\hat{\mathbf{D}}_{\lambda} -
\mathbf{D}_{\lambda}\|_{F}$ thus suggests that
$\hat{\mathbf{D}}_{\lambda} \rightarrow \mathbf{D}_*$ as $\lambda
\rightarrow 0$. A precise statement of this result, however, requires
a more detailed analysis of the relationship between $\lambda$ and
$n$. Indeed, the relative error bound for
$\|\hat{\mathbf{D}}_{\lambda} - \mathbf{D}_{\lambda}\|_{F}$ currently requires
$n$ sufficiently large and $\lambda \gg n^{-1/2}$ while the convergence of $\mathbf{D}_{\lambda}$
to $\mathbf{D}_*$ currently requires $\lambda \rightarrow 0$ with $n$
fixed. We leave this analysis for future work. Finally, 
if we can assume that for sufficiently large $n$ we also have $\min_{i\neq j} c_{ij} = \omega(n^{-1/2})$, then with high probability $\hat{\mathbf{D}}_0 =
\mathrm{diag}(n,n,\dots,n)$ and we exactly recover the true
correspondence for all vertices. The condition $\min_{i\neq j} c_{ij} =
\omega(n^{-1/2})$ is likely to be too restrictive; we
can relax this condition by assuming that the indices
$\{1,2,\dots,n\}$ can be partitioned in to $K$ distinct groups such
that $c_{i j'}>c_{ij}+ \omega(n^{-1/2})  $ for all triplets $(i,j,j')$ where $i$ and $j$
are in the same group and $i$ and $j'$ are in different groups.  
Then under this
milder condition, we can guarantee that with high probability,
$\hat{\mathbf{D}}_0$ satisfies 
$\hat{\mathbf{D}}_0(i,j) = 0$ whenever $i$ and $j$ are in different
groups. 
As a special case if $G_1,G_2$ are
$\rho-$correlated stochastic block model graphs then $\md_0$ will be
block-diagonal and hence,  for sufficiently large
$n$, $\hat{\md}_0$ will also be block-diagonal with high
probability. Thus, for $\rho$-correlated stochastic block models, our
algorithm will generally assign each vertex $u \in G_1$ to another
vertex $v \in G_2$ from the same block as $u$. 

We summarize the above discussion in the following result.
\begin{proposition} 
\label{prop1}
Let $P_{\lambda}$ and $\hat{P}_{\lambda}$ be the optimization
  problem in Eq.~\eqref{eq:objective} with cost
  matrices $\mathbf{C} = (\|X_{i} - X_{j}\|)$
  and $\hat{\mathbf{C}} = (\|\hat{\mathbf{W}} (\tilde\mx_1)_{i} -
  (\hat\mx_2)_{j}\|)$, respectively. Then for sufficiently large $n$, 
  $$\frac{\|\mathbf{D}_{\lambda} -
    \hat{\mathbf{D}}_{\lambda}\|_{F}}{\|\mathbf{D}_{\lambda}\|_{F}} = \frac{1}{ \lambda}\cdot O_p(n^{-1/2}).$$
  For a fixed $n$, as $\lambda \rightarrow 0$, we have
  $$\mathbf{D}_{\lambda} \longrightarrow \argmin_{\mathbf{D} \in
  \mathcal{D}} \{ \|\mathbf{D}\|_{F} \colon \langle \mathbf{C}, \mathbf{D} \rangle
= \xi_*\},$$
where $\xi_*$ is the minimum value achieved in $P_0$. 
Furthermore,
suppose that $\lambda = 0$ and, for sufficiently large $n$ the vertices of $G_1$
and $G_2$ can be partitioned into $K$ distinct groups such that
$c_{ij}< c_{ij'}$ with $c_{ij'} - c_{ij} = \omega(n^{-1/2})$ for all triplets $(i,j,j')$ with $i$ and $j$ being
the same group and $i$ and $j'$ being in different groups. 
Then with high
probability $\hat{\mathbf{D}}_0$ is a block diagonal matrix, i.e.,
$\hat{\mathbf{D}}_0(i,j) = 0$ whenever $i$ and $j$ are in different
groups. 
%
\end{proposition}
Note that the
simulation results in Section~\ref{sec:experiment} are more accurate
than what Proposition~\ref{prop1} suggests, i.e., the correlation structure
between the two graphs lead to better nomination than
just nominating vertices from the same block.






\subsection{Reranking based on likelihood}
\label{sec:reranking}
The algorithm in Section~\ref{sec:methodology} output a nomination
list $\ell(x)$ for the vertex of
interest. When the pair $(G_1, G_2)$ is an instance of a
$\rho$-correlated generalized random dot product graph, then for any
fixed $\rho > 0$, Theorem~\ref{thm:main} guarantees that $\sigma(x)$ is located at the
top of the nomination list $\ell(x)$, i.e., $\mathrm{rk}(\sigma(x))/n
\rightarrow 0$ as $n \rightarrow \infty$, and furthermore, if $\rho = 1$
then $\mathrm{rk}(\sigma(x)) = 1$ asymptotically almost surely. 

We now describe a procedure for refining the nomination list so that
$\mathrm{rk}(\sigma(x)) = 1$ even when $\rho < 1$, provided that we
have enough seed vertices. Let $x \in V_1$ be given and let
$v \in V_2$ be arbitrary. Then for any seed vertex $w \in S_1$ with
correspondence $\sigma(w) \in S_2$, we have, by the assumptions on the
$\rho$-correlation structure
\begin{gather*}
  \mathbb{P}(\ma_1(x,w) = 1, \ma_2(v, \sigma(w)) = 1) = \mpp(x,w)^2 +
  \rho \mpp(x,w)(1 - \mpp(x,w)),
  \quad \text{if }v = \sigma(x), \\
  \mathbb{P}(\ma_1(x,w) = 1, \ma_2(v,\sigma(w)) = 1)  = \mpp(x,w) \mpp(v, \sigma(w)),
  \quad \text{if }v \not = \sigma(x).
\end{gather*}
Let $C_{vw} = \ma_1(x,w) \ma_2(v, \sigma(w)) \in \{0,1\}$ and let
$p_{vw}(\rho)$ be
$$p_{vw}(\rho) = \mpp(x,w) \mpp(v, \sigma(w)) 
+ \rho \mpp(x,w)(1 - \mpp(x,w)).$$
Then the collection $\{C_{vw}\}_{w \in S_1}$ are {\em independent}
Bernoulli random variables with mean parameters $\{p_{vw}(\rho)\}_{w \in
  S_1}$ where $\rho \neq 0$, if and only if $v = \sigma(x)$. 
For a fixed $\rho \neq 0$, the likelihood of observing
$\{C_{vw}\}_{w \in S_1}$, assuming the edge probabilities $\{\mpp_{ij}\}_{i < j}$
are known, is then
\begin{equation*}
  \begin{split}
    \mathcal{L}(\rho; \{C_{vw}\}_{w \in S_1}) = \prod_{w \in S_1} (p_{vw}(\rho))^{c_{vw}}
(1 - p_{vw}(\rho))^{1 - c_{vw}}.
\end{split}
  \end{equation*}
Deciding between $v = \sigma(x)$ and $v \not = \sigma(x)$ is thus analogous
to test $\mathbb{H}_0 \colon \rho
= 0$ against $\mathbb{H}_A \colon \rho \not
= 0$.
For our problem, the edge probabilities $\{\mpp_{i,j}\}_{i < j}$ are
unknown. Nevertheless, Theorem~\ref{thm:main} guarantees that
$\{\mpp_{i,j}\}_{i < j}$ can be estimated uniformly well by
$\{\hat{\mpp}_{i,j}\}_{i < j}$. 
In summary, our procedure for refining the nomination list $\ell(x)$
is as follows.

\begin{itemize}
\item For every $v$ in top ranked part of $\ell(x)$, find $\hat{\rho}_v \in [-1,1]$ that
  maximizes the likelihood $\mathcal{L}(\rho; \{C_{vw}\}_{w \in
    S_1})$; here the true edge
    probabilities $\{\mpp_{i,j}\}$ defining $p_{vw}(\rho)$ are replaced by their estimates
    $\{\hat{\mpp}_{i,j}\}_{i < j}$. 
  \item Reorder these $v \in \ell(x)$ according to decreasing values of $|\hat\rho_v|$.
\end{itemize}


\section{Simulation and Real Data Experiments}
\label{sec:experiment}

We illustrate the performance of our algorithm by 
synthetic and real data experiments. For the synthetic data
experiments,
we generate two types of synthetic data which correspond to two special cases of $\rho$-correlated GRDPG model. 
For the real data
experiments, we use the high-school friendship data from
\cite{moreno2013network} and Microsoft Bing entity graph transitions data from \cite{agterberg2020vertex}.
We evaluate the performance of our algorithm using the following two criteria.

\begin{itemize}
	\item \textbf{Mean reciprocal rank (MRR)}
	
	The reciprocal rank (RR) of a nomination list
    $\ell(x)$ is a measure of how far down a ranked list
one must go to find the true corresponding vertex of interest
$\sigma(x)$, i.e., with a slight abuse of notation, $$\mathrm{RR}(x) = \mathrm{rk}(\sigma(x))^{-1}\in(0,1],$$
	where $\mathrm{rk}(\sigma(x))$ is the rank of $\sigma(x)$ in
    $\ell(x)$, with ties broken randomly. 
	For Monte Carlo experiments, we also consider the mean reciprocal
    rank, i.e., the reciprocal rank averaged over the Monte Carlo
    replicates; we denote this as $\mathrm{MRR}$.
    Larger values of $\mathrm{RR}$ or
    $\mathrm{MRR}$ indicate better performance.
	
	\item \textbf{Mean normalized rank (MNR)}
	
	The normalized rank (NR) of a nomination list $\ell(x)$ is another
    measure of the rank of $\sigma(x)$ in $\ell(x)$, and is defined as
	$$\mathrm{NR}(x)=\frac{\mathrm{rk}(\sigma(x))-1}{|V_2\setminus S_2|-1}\in[0,1],$$
	where $|V_2\setminus S_2|$ is the set of all possible {\em non-seed} candidates. Note
    that $\mathrm{NR}(x)=0$ if and only if $\mathrm{rk}(\sigma(x)) =
    1$. For Monte Carlo experiments we also consider the
    mean normalized rank (MNR). Smaller values of $\mathrm{NR}$
    or $\mathrm{MNR}$ indicate better performance.
\end{itemize}

\subsection{Simulation experiments}
\label{sec:simulation}

We consider two special cases of the $\rho$-correlated GRDPG
model. We first assumes that the matrix of edge
probabilities is positive semidefinite, i.e., that $q = 0$ in the
GRDPG model. We refer to this as the $\rho$-$\mathrm{RDPG}$ or
$\rho$-correlated random dot product graphs model.
The second assumes that the edge
probabilities matrix is that of the popular stochastic block model
graphs \cite{holland}.


In the case of
$\rho$-$\mathrm{RDPG}$, we generate pairs of graphs $(G_1, G_2)$ on $n =
300$ vertices where the latent positions $\{\mx_i\}$ are sampled
uniformly on the unit sphere in $\mathbb{R}^{3}$. We then choose,
uniformly at random, a vertex $x \in V_1$ and use our algorithm to
find a nomination list $\ell(x) \subset V_2$.
Recall that the edges of the graphs are pairwise correlated. This
correlation structure then yields a canonical notion of correspondence between the
vertices in $G_1$ and those in $G_2$. In
other words, given any vertex $v \in G_1$ with latent position $X_v$, the most ``interesting'' or
``similar'' vertex to $v$ in $G_2$ is simply the vertex $u \in G_2$ with
latent position $X_u = X_v$. We evaluate our
algorithm using the mean reciprocal rank (MRR) and the mean normalized
rank (MNR) calculated from $500$ Monte Carlo replicates. The results,
as a function of the correlation $\rho \in \{0,0.3,0.5,0.7,1\}$,
are presented in Figure~\ref{rho-RDPG}. The
embeddings of the graphs are aligned either via orthogonal Procrustes (see
Section~\ref{sec:alignment}) or via the adaptive point set registration procedure of \cite{myronenko2010point}.

The setup for $\rho$-$\mathrm{SBM}(L,b,\mb)$ is similar. We generate
pairs of graphs on $n = 300$ vertices with $L = 3$ blocks and
$100$ vertices in each block. The block probabilities
matrix is
$$\mb=\left[\begin{array}{lll}0.7 & 0.3 & 0.4 \\ 0.3 & 0.7 & 0.2
              \\ 0.4 & 0.2 & 0.7\end{array}\right].$$
The mean reciprocal rank (MRR) and the mean normalized
rank (MNR), calculated from $500$ Monte Carlo replicates, are given in Figure~\ref{rho-SBM}. 
 
For these two settings we choose $d=3$ for the adjacency spectral
embedding step, i.e., we embed the graphs into $\real^3$. We recall
that our algorithm requires alignment of these embeddings either via
orthogonal Procrustes or via an adaptive point set
registration. These two choices lead to two slightly different
quadratic program formulation. More specifically, as we are
embedding into $\mathbb{R}^{3}$, the orthogonal Procrustes procedure needs at least $3$ seed vertices to
align the embeddings. These seed vertices can then be incorporated
into the quadratic program in
Section~\ref{sec:quadratic}. In contrast, if the embeddings are aligned
using adaptive point set registration, then seed vertices are not
necessary and hence the quadratic program is formulated with no
seeds. When using orthogonal Procrustes, we also explore the impact of $K$, the number of seed
vertices. We find that increasing $K$ does improve our algorithm,
but that the improvement is not overly substantial in the regime where
$K$ is small. For example, in the $\rho$-$\mathrm{RDPG}$ setting with
$\rho=0.5$, increasing $K$ from $3$ to $9$ increases the MRR from
$0.28$ to $0.37$ for the orthogonal Procrustes step, thus, for simplicity of presentation, we fixed $K =6$ and show the results.

The mean reciprocal rank and mean normalized rank of our algorithm, as
a function of the correlation coefficient $\rho$, are presented in
Fig.~\ref{rho-RDPG} and Fig.~\ref{rho-SBM}. Our
algorithm is generally quite accurate. In particular, even when
$\rho=0$ our algorithm still performs substantially better than the
baseline. We also note that the performance of orthogonal Procrustes (using $K = 6$ seeds) and adaptive point set
registration (with no seeds) are similar, with the difference being even less
pronounced in the $\rho$-SBM setting. We posit that the more obvious community structure
in the SBM setting helps the adaptive
point set registration procedure to align the embeddings more
accurately, thereby reducing the need for seed vertices.
\begin{figure}[tp!]
\centering
\subfigure{\includegraphics[width=0.35\textwidth]{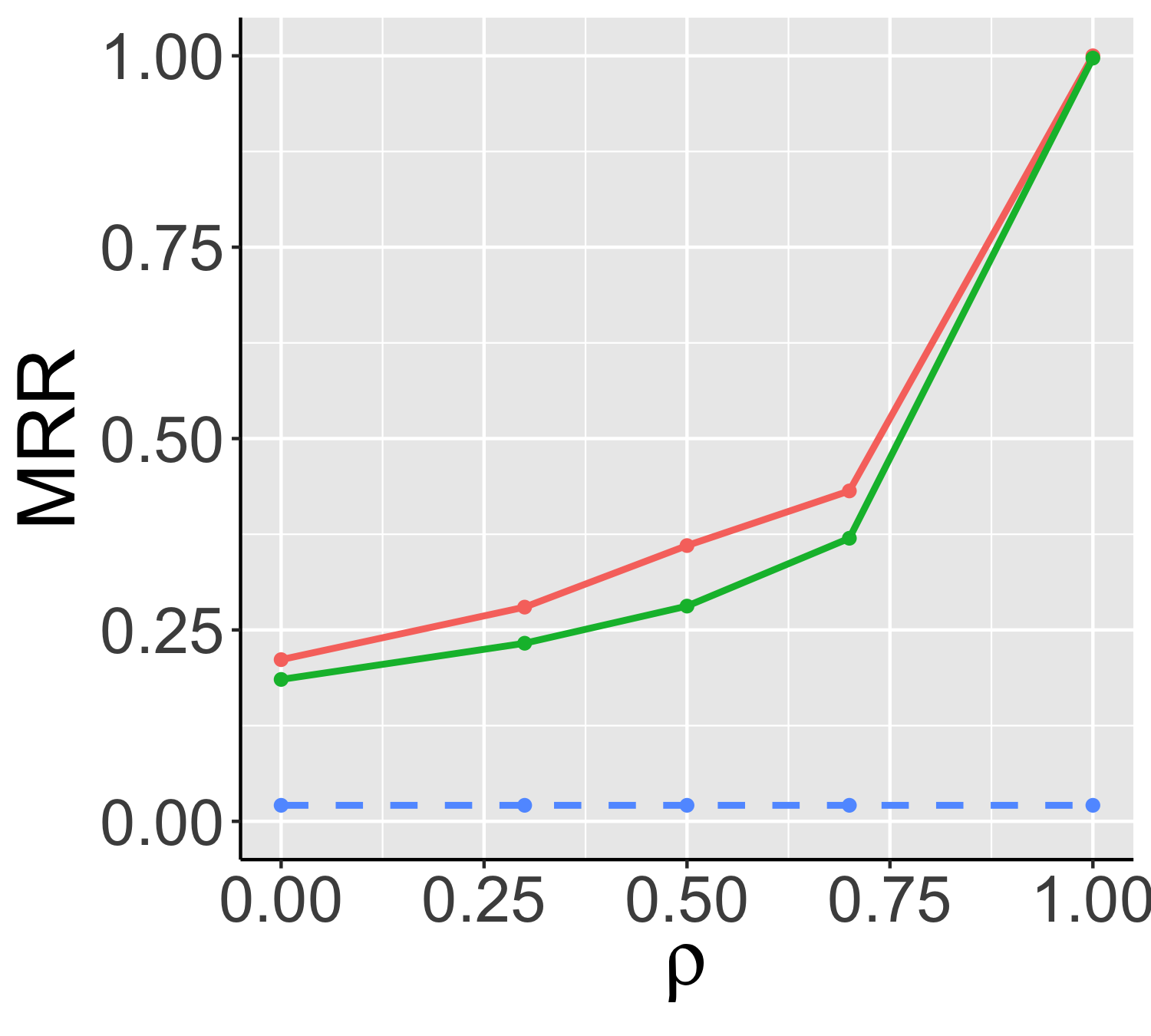}}
\subfigure{\includegraphics[width=0.35\textwidth]{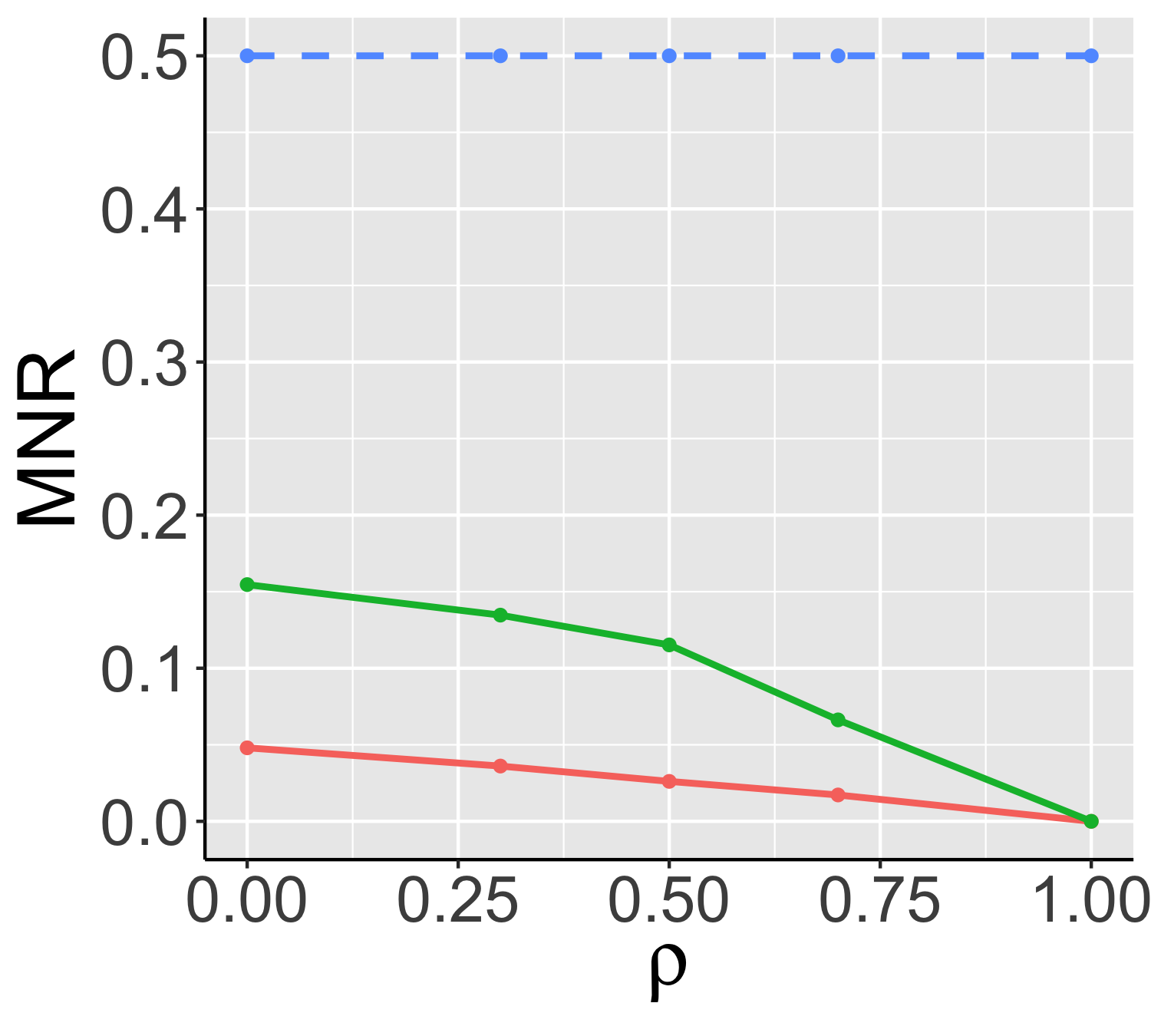}}
\caption{Performance of our algorithm for pairs of
  $\rho$-$\mathrm{RDPG}$ graphs on $n = 300$ vertices. The mean
  reciprocal rank (MRR) and mean normalized rank (MNR) are computed
  based on $500$ Monte Carlo replicates. The MRR and MNR are plotted
  for different values of the correlation coefficient $\rho$. The red
  and green lines
  correspond to the case where the graphs embeddings are aligned via orthogonal Procrustes
  and via the adaptive point set registration procedure,
  respectively.
  The dotted blue lines correspond to the baseline MRR and
  MNR for a nomination list chosen uniformly at random.}
\label{rho-RDPG}
\end{figure}

\begin{figure}[htbp]
\centering
\subfigure{\includegraphics[width=0.35\textwidth]{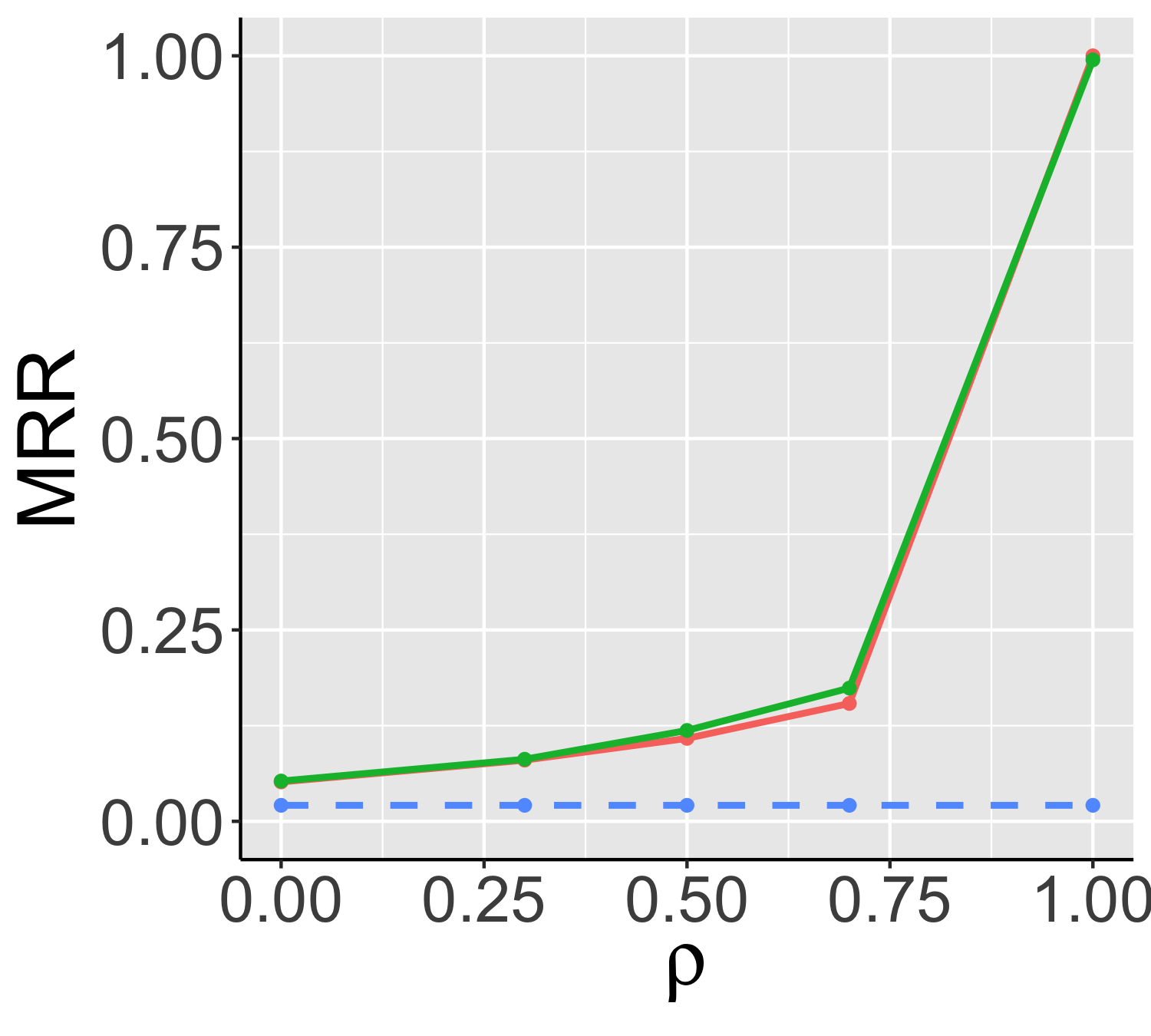}}
\subfigure{\includegraphics[width=0.35\textwidth]{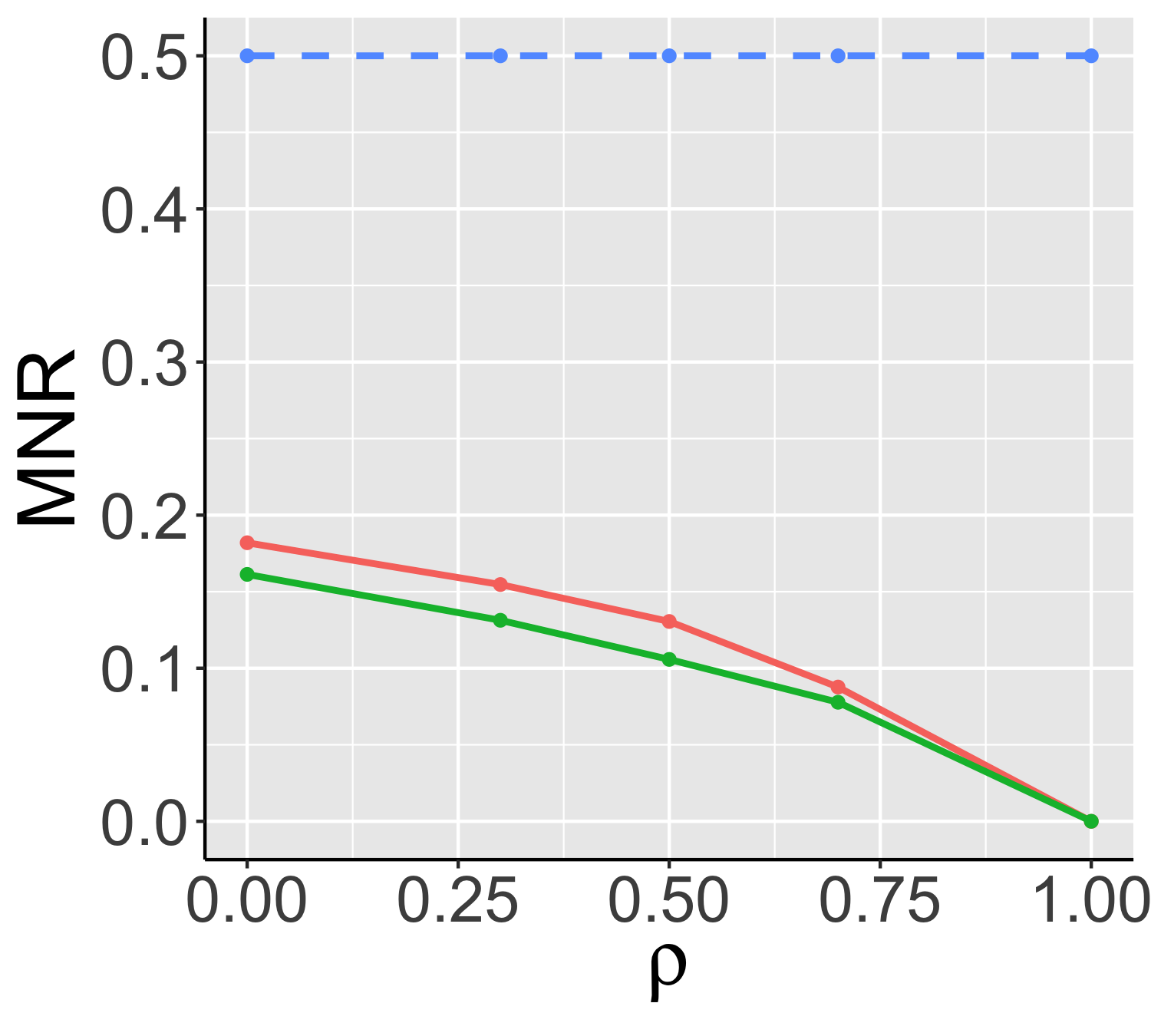}}
\caption{Performance of our algorithm for pairs of
  $\rho$-$\mathrm{SBM}$ graphs on $n = 300$ vertices. The mean
  reciprocal rank (MRR) and mean normalized rank (MNR) are computed
  based on $500$ Monte Carlo replicates. See the caption to
  Figure~\ref{rho-RDPG} for further descriptions of the various
  colored lines.}
\label{rho-SBM}
\end{figure}

We next explore how the reranking step in Section~\ref{sec:reranking}
can improve the performance of our algorithm, especially when there
are enough seed vertices. More specifically, we set $\rho = 0.7$ and vary the number of
seed vertices $K$ from $10$ to $50$. These seed vertices are
incorporated into both the quadratic program formulation and the
reranking step. We then compare the performance of our
algorithm with and without the reranking step. The MRR
averaged over $500$ Monte Carlo replicates are
presented in Fig.~\ref{likelihood_rho-RDPG} for the $\rho$-RDPG
setting and in Fig.~\ref{likelihood_rho-SBM} for the $\rho$-SBM
setting. These figures indicate that the reranking step leads to
significant improvement even for small values of $K$, e.g., compare
the MRR in the $\rho$-$\mathrm{SBM}$ setting with $K = 10$
seeds.

\begin{figure}[htbp]
\centering
\subfigure{\includegraphics[width=0.35\textwidth]{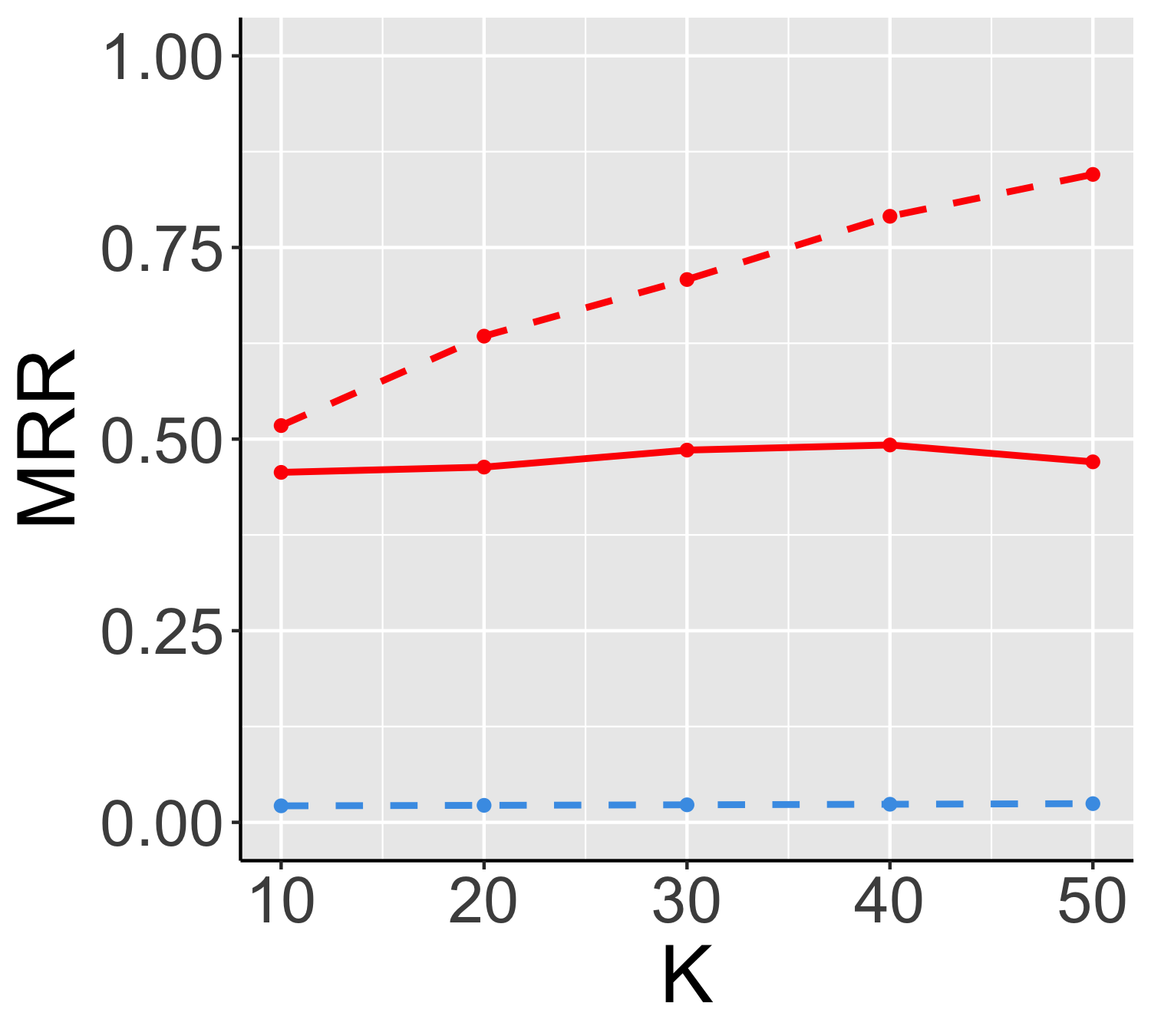}}
\subfigure{\includegraphics[width=0.35\textwidth]{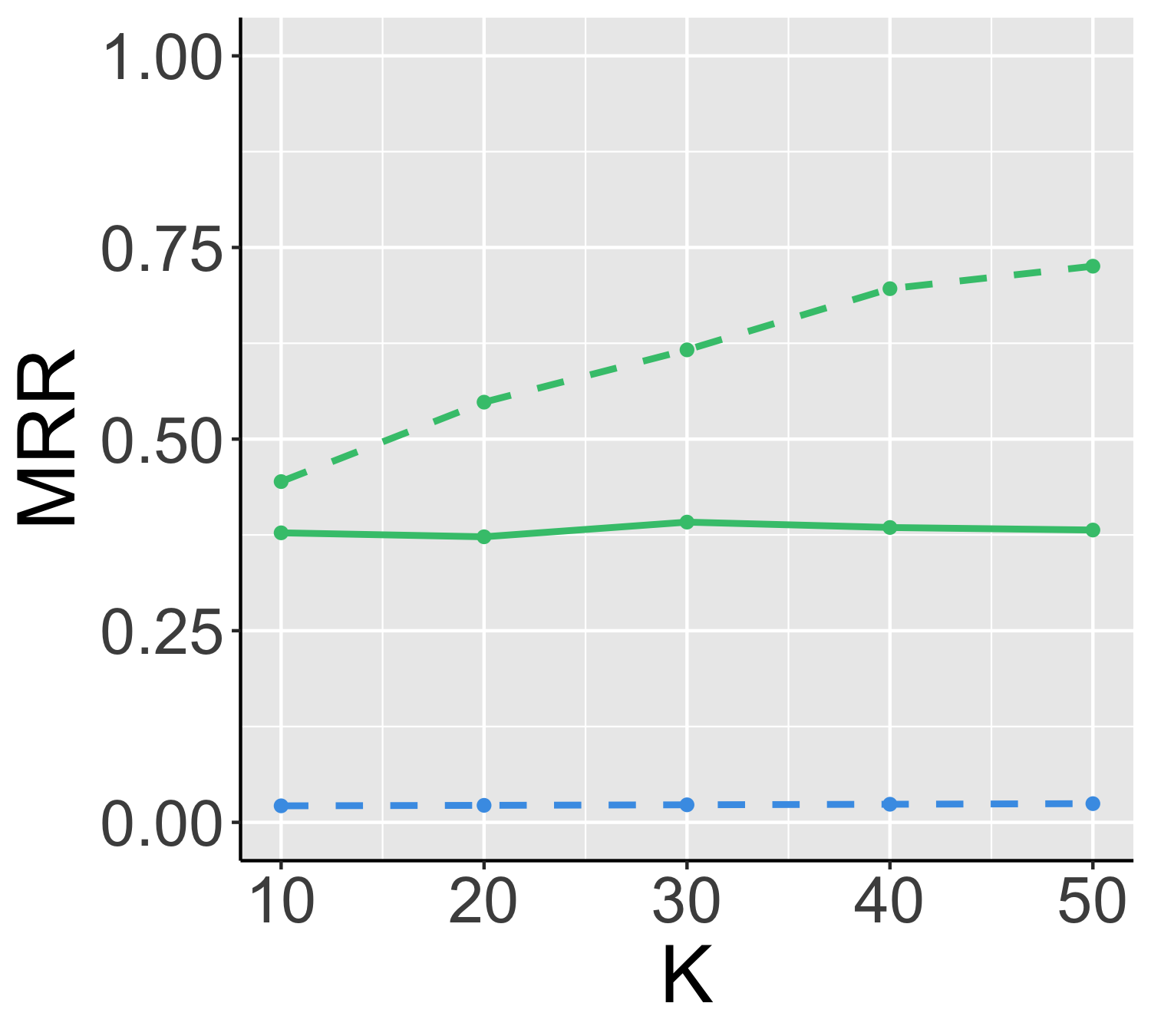}}
\caption{Performance of our algorithm with and without the reranking
step for pairs of $\rho\text{-RDPG}$ graphs on $n=300$ vertices and correlation $\rho = 0.7$. The
mean reciprocal rank (MRR) are computed based on $500$ Monte Carlo
replicates. The red
  and green lines
  correspond to the case where the graphs embeddings are aligned via orthogonal Procrustes
  and via the adaptive point set registration procedure,
  respectively. In each plot, the corresponding dashed red or green line describes the result after the reranking step. The dotted blue lines correspond to the baseline MRR for
a nomination list chosen uniformly at random.}
\label{likelihood_rho-RDPG}
\end{figure}

\begin{figure}[htbp]
\centering
\subfigure{\includegraphics[width=0.35\textwidth]{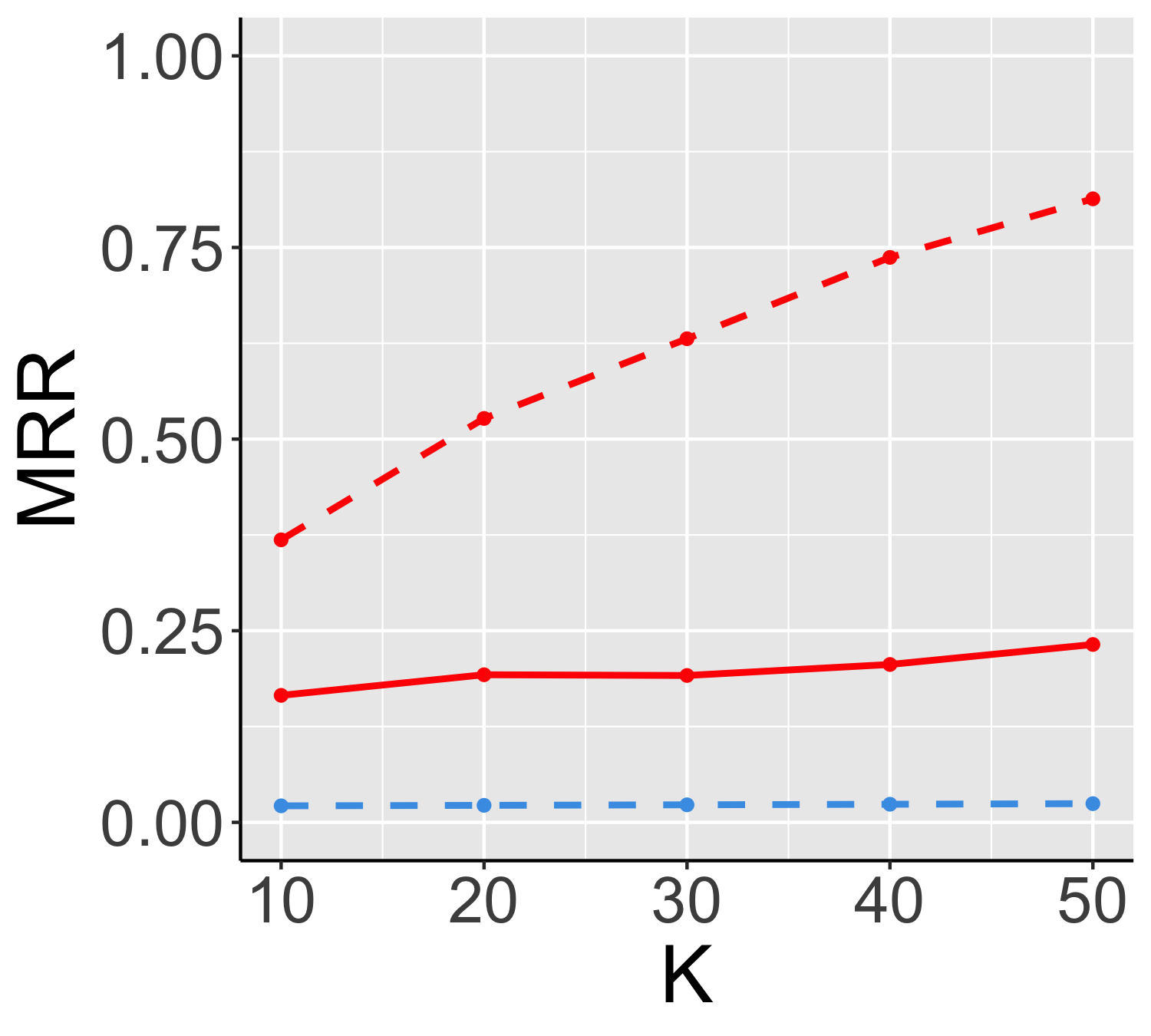}}
\subfigure{\includegraphics[width=0.35\textwidth]{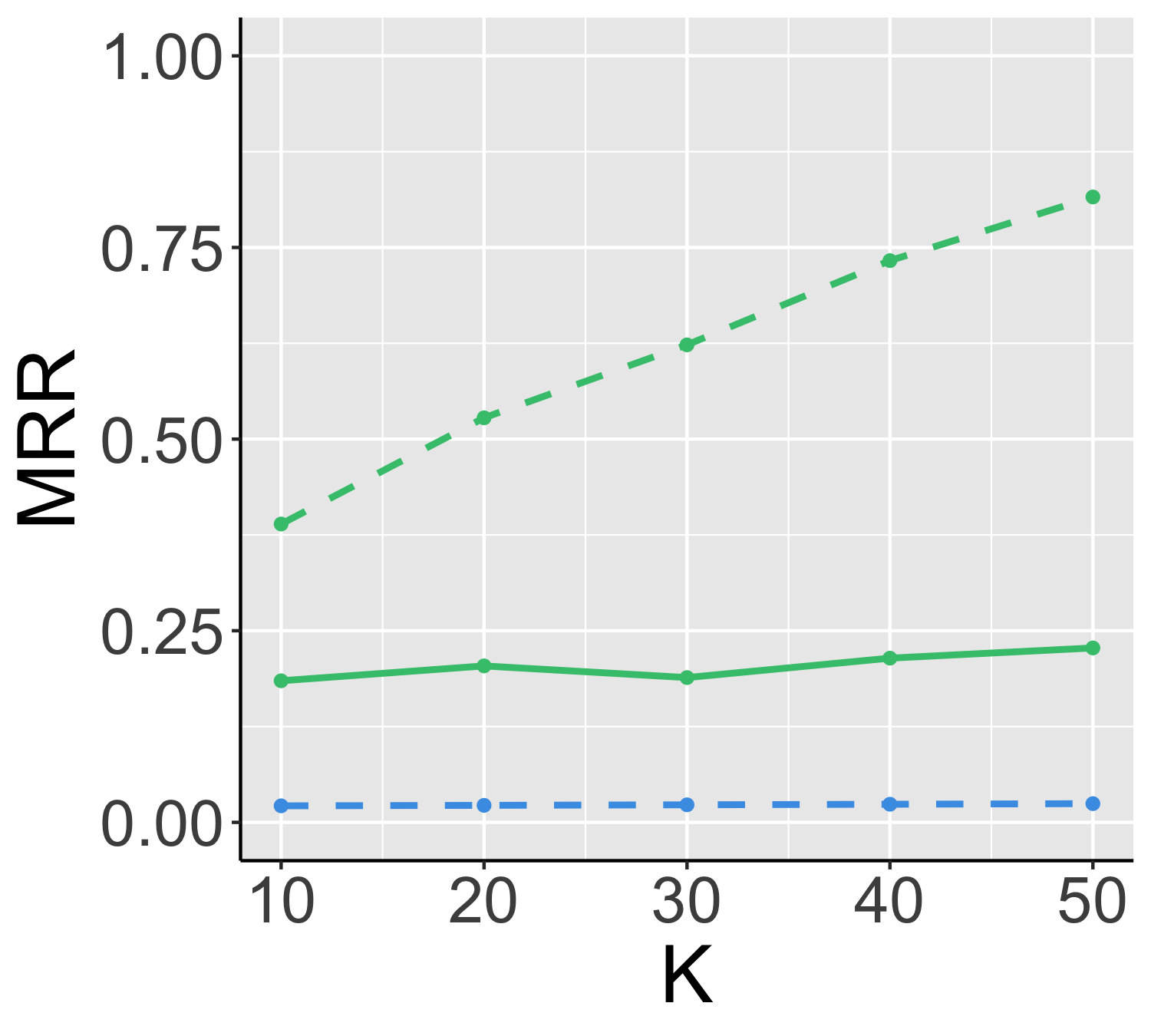}}
\caption{Performance of our algorithm with and without the reranking
  step for pairs of $\rho\text{-SBM}(\mx)$ graphs on $n=300$
  vertices and correlation $\rho = 0.7$. See the caption to
  Figure~\ref{likelihood_rho-RDPG} for further descriptions.}
\label{likelihood_rho-SBM}
\end{figure}

Appendix~\ref{sec:add res} contains additional simulation results illustrating how the choice of embedding dimension $d$, the sparsity parameter $\gamma$, and the penalty parameter $\lambda$ affects the performance of our algorithm. In particular, Figure~\ref{rho-RDPG_d} and Figure~\ref{rho-SBM_d} show that our algorithm is relatively robust to the choice of $d$ while Figure~\ref{rho-SBM_gamma}, Table~\ref{tab:RDPG_lambda} and Table~\ref{tab:SBM_lambda} show that our algorithm is also relatively robust to the choice of $\gamma$ and $\lambda$, respectively. In addition, we also include in Appendix~\ref{sec:add res} comparisons between our algorithm and the embedding followed by Gaussian mixture modeling algorithm of \cite{agterberg2020vertex}, and we see from Figure~\ref{rho-RDPG_compare} and Figure~\ref{rho-SBM_compare} that the accuracy of \cite{agterberg2020vertex}'s method and our algorithm with orthogonal Procrustes are very similar but we note the running time of \cite{agterberg2020vertex}'s method is roughly $6$ times slower than ours, and we see our algorithm with the adaptive point set registration procedure needs no seeds information but also has comparable accuracy.

\subsection{Real data experiments}

We now explore the practical application of our algorithms on real
data. In Section~\ref{sec:High School}, we consider a pair of
high-school friendship networks containing some of the same vertices
and in which we would like to identify the same individuals across the
two networks. In Section~\ref{sec:Bing}, we explore the graphs derived from
Microsoft Bing entity graph transitions.

\subsubsection{High school friendship networks}
\label{sec:High School}

We first focus on the high school friendship network data from
\cite{mastrandrea2015contact}. This dataset contains two observed
graphs and, for each graph, the vertices  represent
students and the edges represent their friendship. The
first graph is extracted from the Facebook social network, i.e., if two
individuals are friends on Facebook, then they are adjacent. The
second graph is created based on the result of a survey of the
students; for every pair of students, they are considered adjacent if
at least one of the students in this pair reports that they are friends with another student. There are $156$ vertices in the first graph, $134$
vertices in the second graph, and $82$ vertices shared between the two graphs.
These $82$ shared vertices will induce the notion of interestingness
for our subsequent analysis. In other words, given a query vertex $x$
in one graph, with $x$ being one of the $82$ shared vertices, we are interested in finding
the same vertex $x$ in the second graph. This application is thus
analogous to that of network deanoymization. 

As the number of unshared vertices is reasonably large, we consider
two experimental setups. In the first setup we used only the subgraphs
induced by the $82$ shared vertices while in the other setup we used
the full graphs on $156$ and $134$ vertices. For the adjacency
spectral embedding step we set $d=2$. Orthogonal Procrustes alignment
of the embeddings then requires at least $2$ seed vertices.


For the experiment using only the shared vertices we
iteratively consider each vertex as the vertex of interest. For each
vertex of interest we choose a pair of seed vertices, align the
embeddings using orthogonal Procrustes, and then solve a quadratic program to obtain a
nomination list (the seed vertices are not used in the quadratic program). We repeat this procedure $100$ times for each vertex
of interest, each time choosing a random pair of seed
vertices. Figure~\ref{highschool_shared_Algo1_count} then illustrates,
for each of the $82$ possible vertex of interest $x$, how often
$\mathrm{NR}(x) \in \{0,(0,0.2],(0.2,0.5],(0.5,1]\}$; the
mean normalized rank for a nomination list chosen uniformly at random is
$0.5$. Figure~\ref{highschool_shared_Algo1_count} indicates that the
nomination lists obtained by our algorithm are in general quite
accurate; indeed, the normalized rank values are small for most of the
nomination lists, with a significant portion of the nomination lists
even having normalized rank values of $0$, i.e., the true correspondence
of the vertex of interest is at the top of the nomination list.

We next consider the impact of increasing the number of seed vertices
$K$. For simplicity, we present our analysis for a randomly chosen
vertex of interest $x = 27$ as an example. Similar results hold for
other vertices. We vary $K$ from $2$ to $10$
and run 500 Monte Carlo replicates to compute the MNR. We tabulate
how often $\mathrm{NR}(x) \in \{0,(0,0.2],(0.2,0.5],(0.5,1]\}$ in Figure~\ref{highschool_shared_Algo1_single}. We see from Figure~\ref{highschool_shared_Algo1_single}
that $K = 7$ seed vertices is sufficient for the NR of the nomination
lists for $x = 27$ to be
between $0$ and $0.2$ always. 

Analogous results are available when we align the embeddings using
adaptive point set registration procedure. However, since adaptive
point set registration does not use any seed vertex, it lead to more
robust performance when compared to using
orthogonal Procrustes.
Finally, we note in passing that our algorithm is quite
computationally efficient, e.g., generating
Figure~\ref{highschool_shared_Algo1_count} takes us only about 7 minutes
on a normal laptop.

\begin{figure}[tp!]
\centering
\includegraphics[width=0.99\textwidth]{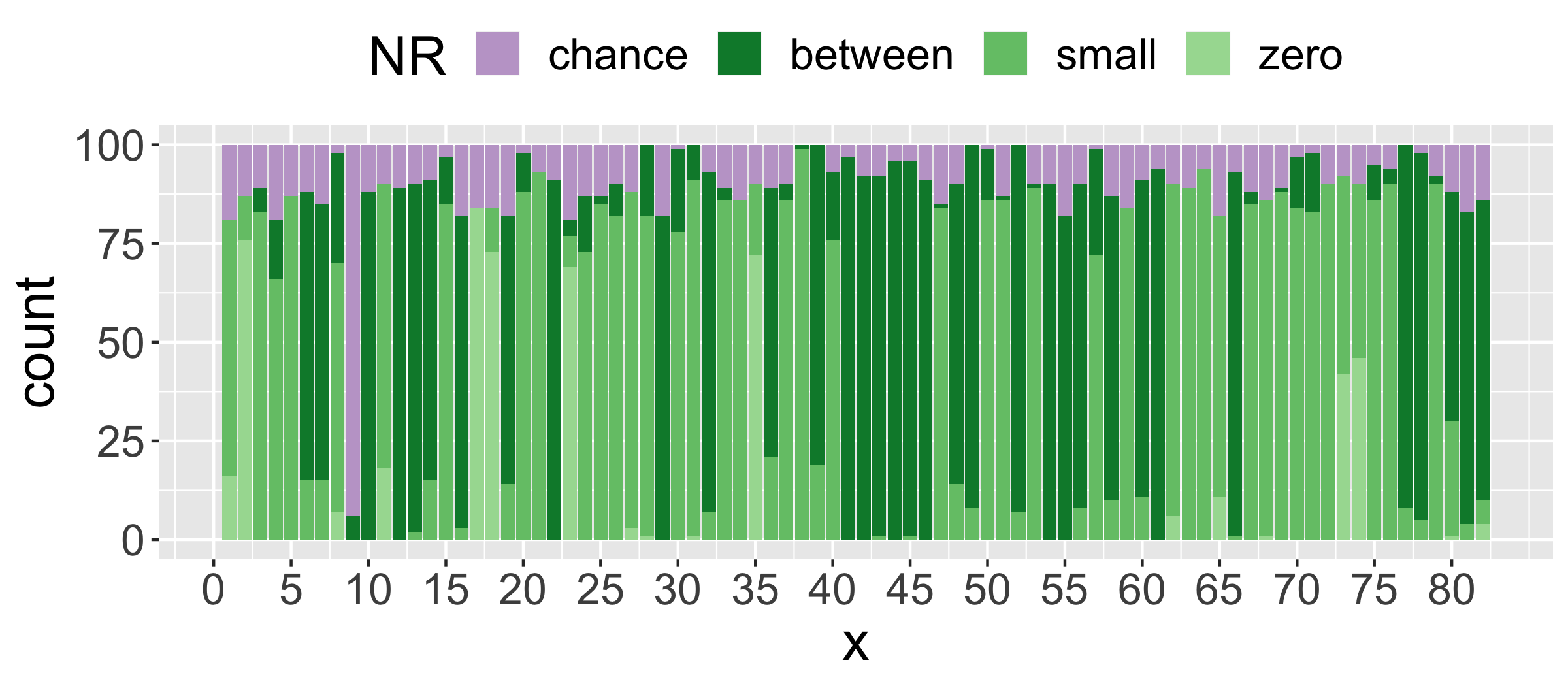}
\caption{Performance of our algorithm for vertex nomination between
the two high-school networks. Here we consider only the subgraphs
induced by the $82$ shared vertices.  The graphs embeddings are
aligned via orthogonal Procrustes transformation using two randomly
selected seeds; these seeds are only used for the alignment and are
not incorporated into the quadratic programming step. For each $x \in
V_1$ we repeat this random seeds selection $100$ times and record the
normalized rank of its correspondence $\sigma(x) \in V_2$. The four
categories correspond to the case when the normalized rank (NR) is equal to
$0$, lying between $0$ and $0.2$, lying between $0.2$ and $0.5$, or
larger than $0.5$.}

\label{highschool_shared_Algo1_count}
\end{figure}

\begin{figure}[htbp!]
\centering
\subfigure{\includegraphics[width=0.49\textwidth]{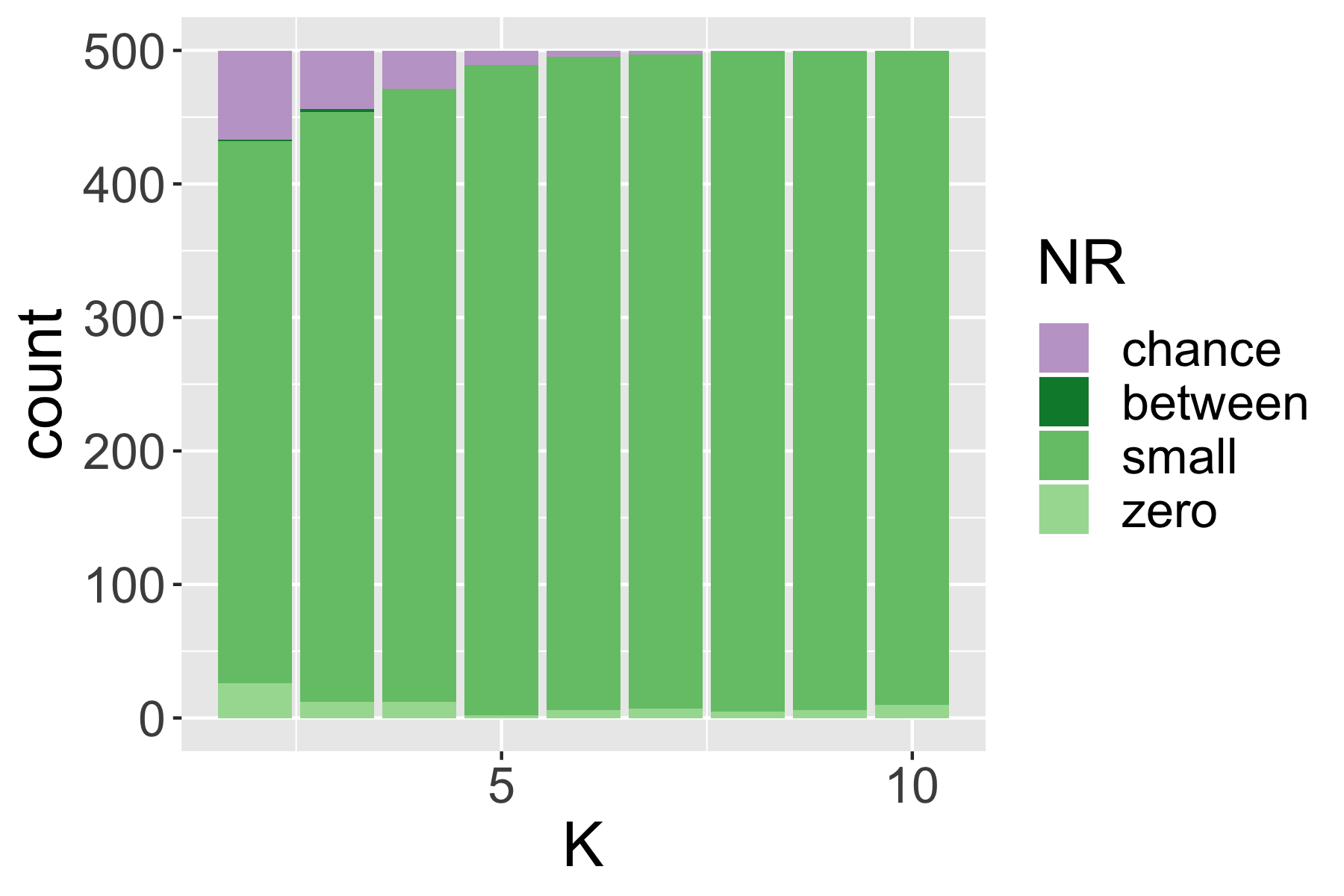}}
\subfigure{\includegraphics[width=0.49\textwidth]{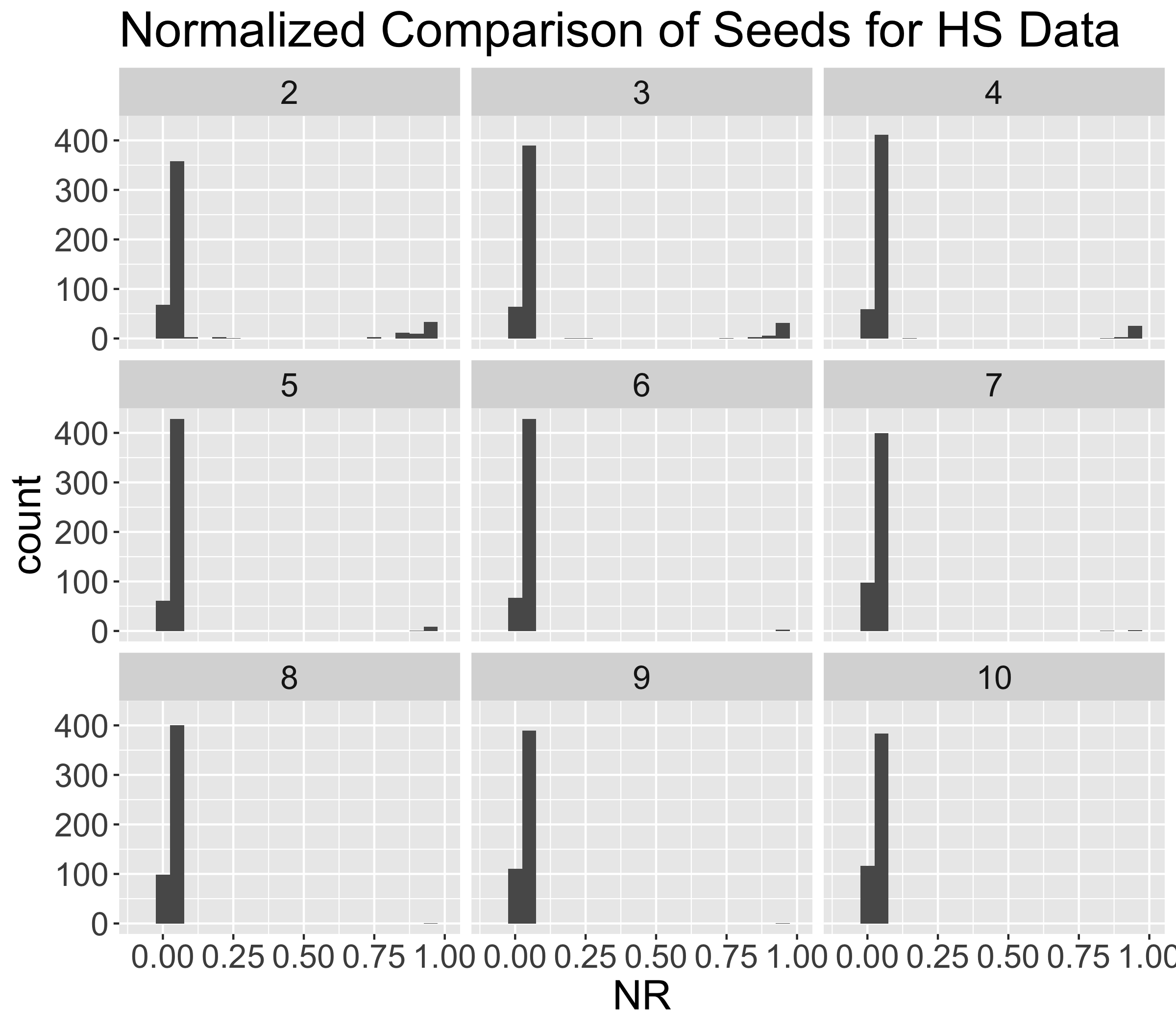}}
\caption{For applying the algorithm with orthogonal Procrustes to
subgraphs of high school network generated by shared vertices, using
$x=27$ as the vertex of interest, we vary the number of seed vertices
$K$ from $2$ to $10$, uniformly at random generate $500$ sets of seed
vertices and plot NR.}
\label{highschool_shared_Algo1_single}
\end{figure}

\begin{figure}[tbp]
\centering
\includegraphics[width=0.99\textwidth]{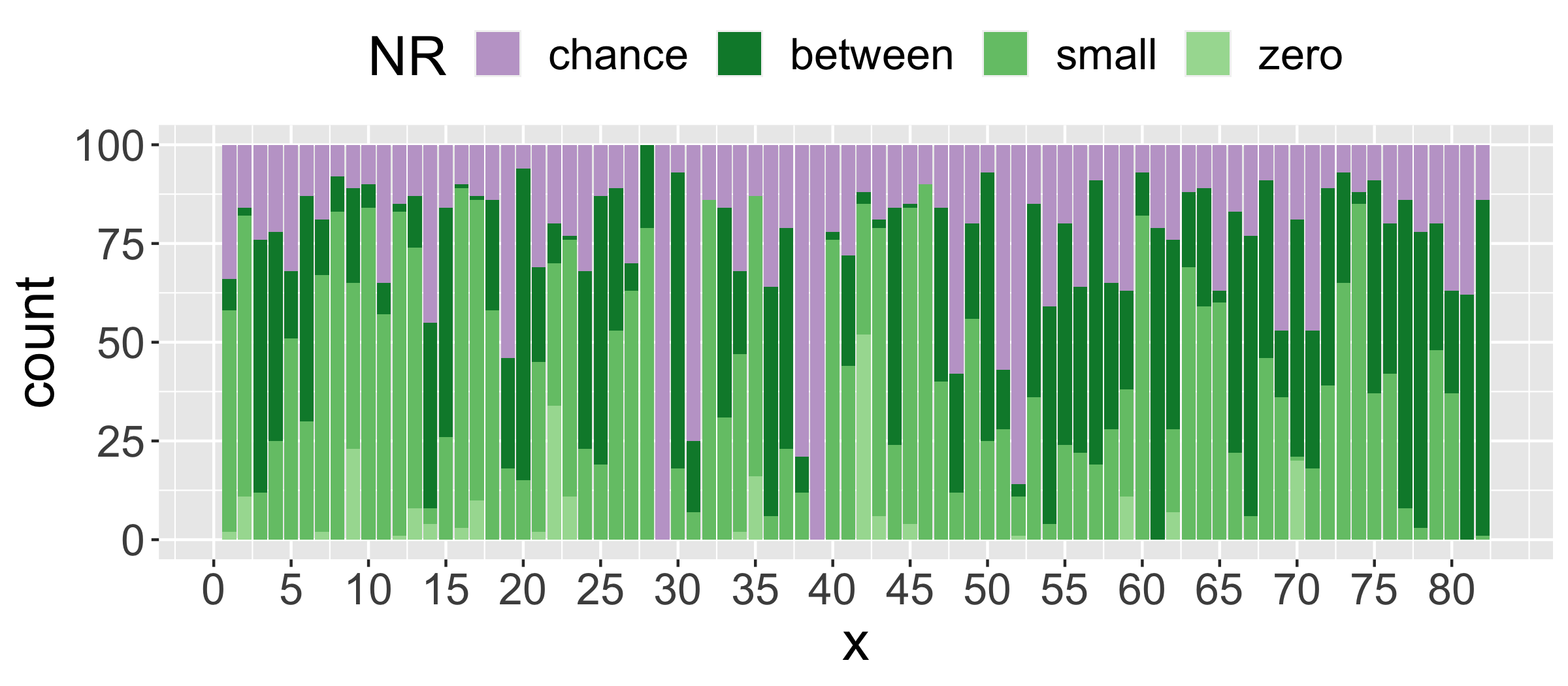}
\caption{Performance of our algorithm for vertex nomination between
the two high-school networks. Here we consider the graphs with full
vertices. The graphs embeddings are aligned via orthogonal Procrustes
transformation using two randomly selected seeds; these seeds are only
used for the alignment and are not incorporated into the quadratic
programming step. See the caption to
Figure~\ref{highschool_shared_Algo1_count} for further descriptions of
the experiment.}
\label{highschool_full_Algo1_count}
\end{figure}
We now consider the setup using the full graphs on $134$ and $156$
vertices. Once again we use orthogonal Procrustes to align the
embeddings. We then consider each of the $82$ shared
vertices as the vertex of interests $x$ and find the nomination list
$\ell(x)$ using the same procedure as that outlined above for the
setup using the induced subgraphs. Note that the main
difference between the current setup and that of the induced subgraphs
is that, for each vertex of interest, there are more candidate
vertices in the current setup; this make the task harder and hence the performance of our
algorithm is likely to be worse in the current setup. The experiment
results in Figure~\ref{highschool_full_Algo1_count}
confirmed this speculation. Indeed, comparing
Figure~\ref{highschool_shared_Algo1_count} and 
Figure~\ref{highschool_full_Algo1_count}, we see that the number of
times in which the obtained nomination list is no better than chance
increases. Nevertheless, our algorithm is still quite accurate since, for
almost all of the vertex of interests,
the true correspondences do appear frequently at the top of the nomination lists. 
\subsubsection{Microsoft Bing entity graph transitions}
\label{sec:Bing}
In this section, we consider graphs derived from one month of Bing
entity graph transitions. The dataset for this example is from
\cite{agterberg2020vertex} and contains two graphs on the \emph{same} set of
vertices; these vertices denote entities. The (weighted) edges in each graph
represent transition rates between the entities during an internet browsing
session, but the types of transitions differ between the two
graphs. More specifically, the edges in the first graph $G_1$ represents transitions that were made using a
suggestion interface while the transitions in the second graph $G_2$ were
made independently of any suggestion interface. As the suggestion interface can only
suggest a few entities at a time, the edges in $G_1$
are much more constrained than those in $G_2$. The first and
second graphs both have $13535$ vertices and approximately $5.2 \times 10^5$
and $5.9 \times 10^5$ edges, respectively. There is, once again, a
one-to-one correspondence between the vertices in both networks and we
use this correspondence to define our notion of interestingness, i.e.,
given a vertex $x$ in one graph, we are interested in finding the same
vertex in the other graph. 

For our first analysis we sub-sample the graphs and only consider the
subgraphs induced by the first $1000$ vertices. These induced
subgraphs are also unweighted, i.e., two vertices are adjacent in a
induced subgraph if the corresponding transition rate in the original
graph is non-zero. Denoting by $G_1$ and $G_2$ the resulting induced
subgraphs, $G_1$ and $G_2$ have $8365$ edges and $10247$ edges, respectively. We
emphasize that there is a 1-to-1 correspondence between the vertex
sets of $G_1$ and $G_2$.

We now explore the performance of our algorithm for vertex nomination
between $G_1$ and $G_2$. In particular, we sequentially consider each
vertex $x \in G_1$ as the vertex of interest, and for a given vertex
of interest we randomly select 10 other vertices as seeds. After
computing the NR for all vertices, we present the histogram of NR to
show the distribution.  The results are given in
Figure~\ref{Bing_first1000} for both the cases where the graph
embeddings are aligned via orthogonal Procrustes and via adaptive
point set registration. We emphasize that there are two variants of
adaptive rigid point set registration used here. In the first variant the
$10$ seed vertices are used in the quadratic programming formulation
while in the second variant the seed vertices are not used at all. 
Figure~\ref{Bing_first1000} indicates that the
normalized rank values are generally quite small and hence the
nomination lists returned by our algorithm are
accurate. Figure~\ref{Bing_first1000} also indicates that there is
almost no difference between using orthogonal Procrustes and using
adaptive point set registration and, more importantly, our
algorithm perform well even when there are no seeds information, i.e.,
the performance of adaptive point set registration
with no seeds is virtually identical to that of orthogonal Procrustes
and adaptive point set registrations with $10$ seeds. Indeed,
Table~\ref{tab:Bing1} summarizes the quantiles of the NR for different
variants of our algorithm and we see from these quantiles that the
performance of the three variants are virtually indistinguishable.
\begin{figure}[tp]
\centering
\subfigure{\includegraphics[width=0.3\textwidth]{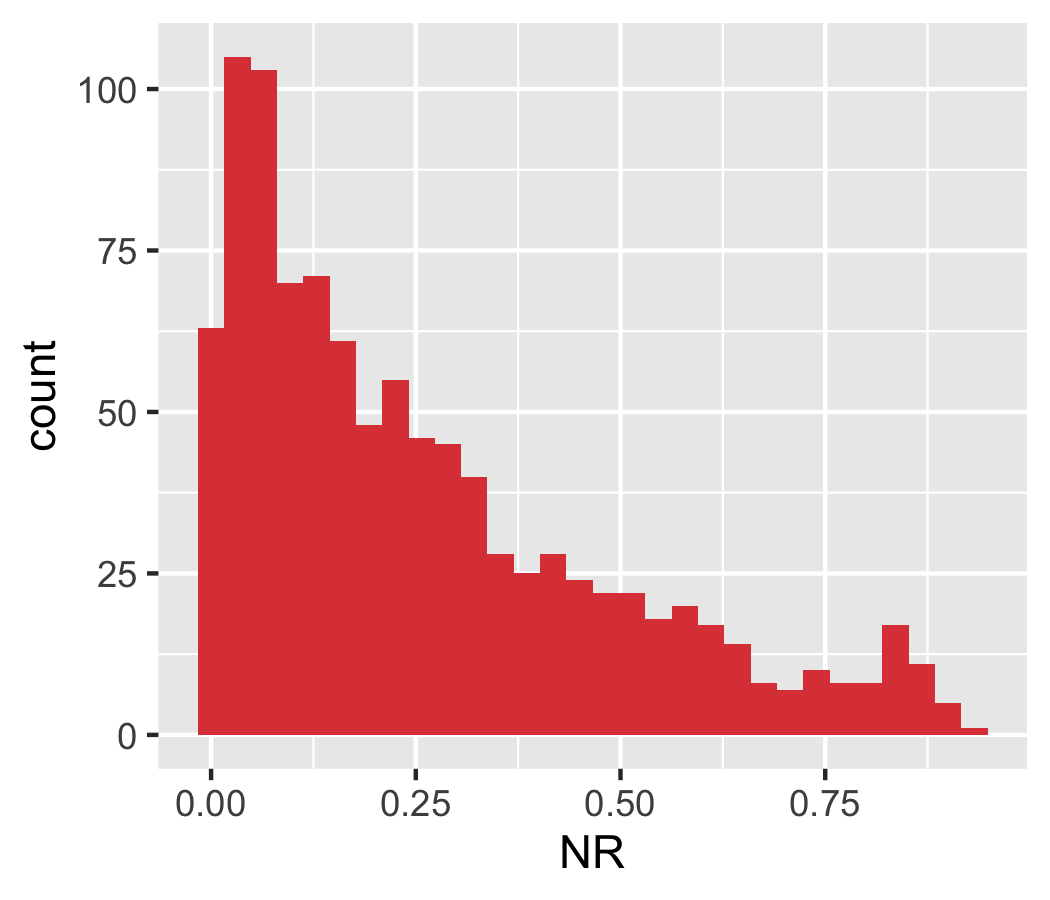}}
\subfigure{\includegraphics[width=0.3\textwidth]{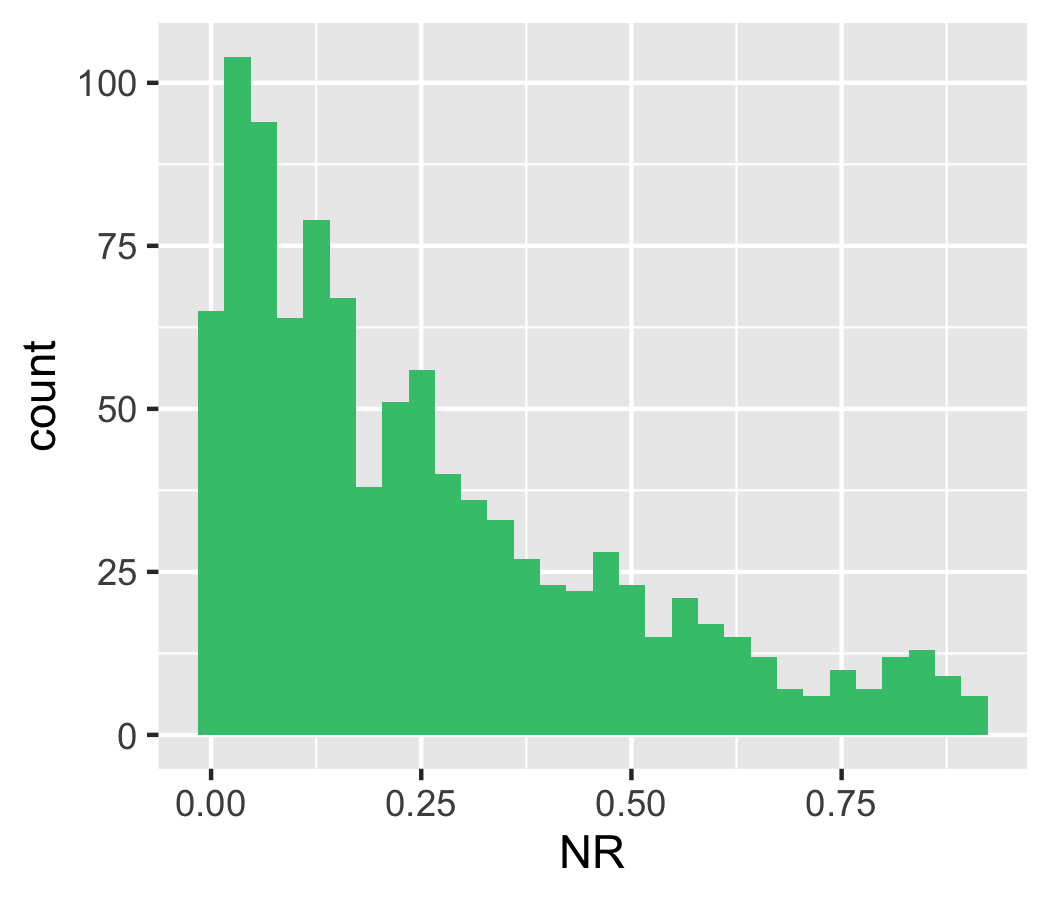}}
\subfigure{\includegraphics[width=0.3\textwidth]{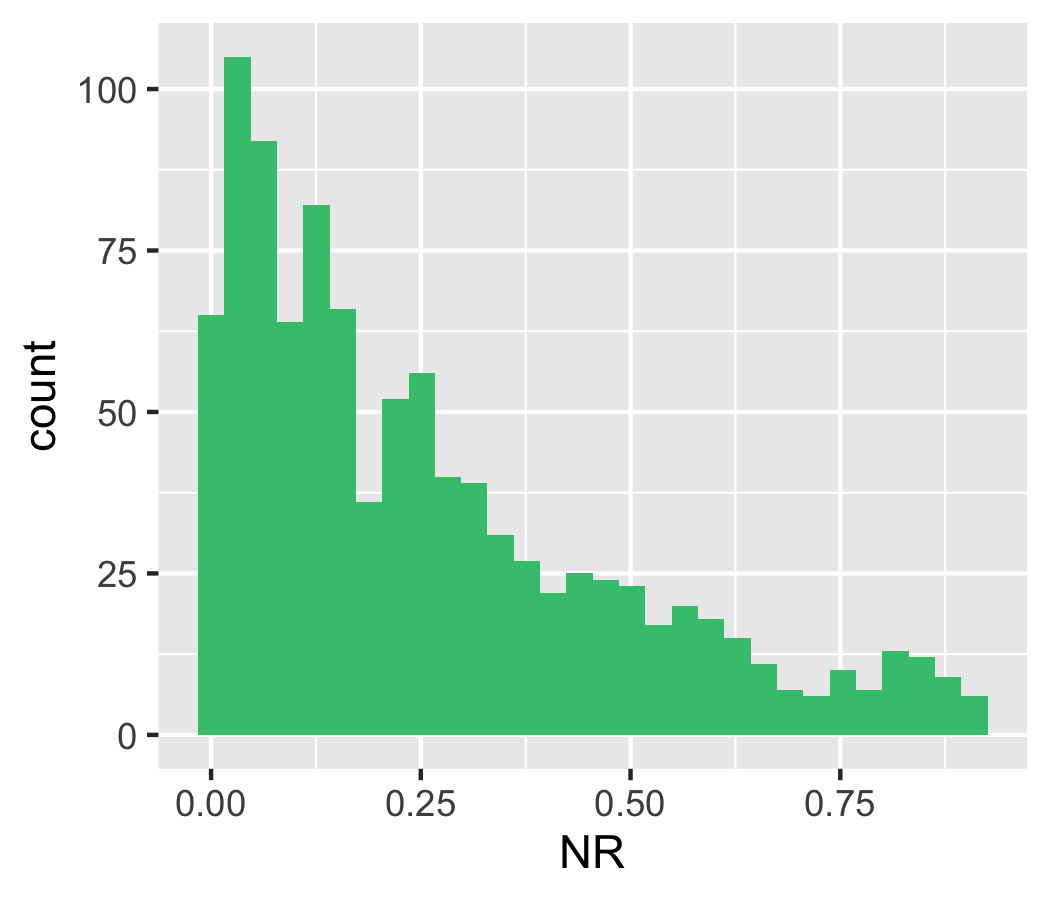}}
\caption{Performance of our algorithm for vertex nomination between
the two Microsoft Bing entities transition networks on $n=1000$
vertices.  For each $x\in V_1$ we randomly selected 10 seeds and
record the normalized rank (NR) of its correspondence $\sigma(x)\in
V_2$.  The red and green histogram of NR correspond to the case where
the graphs embeddings are aligned via orthogonal Procrustes and via
the adaptive point set registration procedure, respectively. The last
green figure corresponds to the case for adaptive point set
registration procedure without any seeds.}
\label{Bing_first1000}
\end{figure}
\begin{table}[pth]
\centering
\begin{tabular}{c|cccccccc}
 & 1\% & 5\%& 10\%&25\%& 50\%& 75\%&95\%& 99\% \\
 \hline
Procrustes (10 seeds) &  0.003 & 0.013 & 0.030 & 0.074 & 0.196 & 0.387 & 0.750 & 0.870   \\
set registration (10 seeds) &  0.002 & 0.012 & 0.025 & 0.073 & 0.196 & 0.387 & 0.757 & 0.877\\
set registration (no seeds) & 0.002 & 0.013 & 0.025 & 0.073 & 0.196 & 0.386 & 0.757 & 0.876 
\end{tabular}
\caption{Quantile levels of normalized rank (NR) values for vertex nomination
  with the Bing entity networks on $n = 1000$ vertices}
\label{tab:Bing1}
\end{table}
\begin{figure}[pth]
\centering
\subfigure{\includegraphics[width=0.4\textwidth]{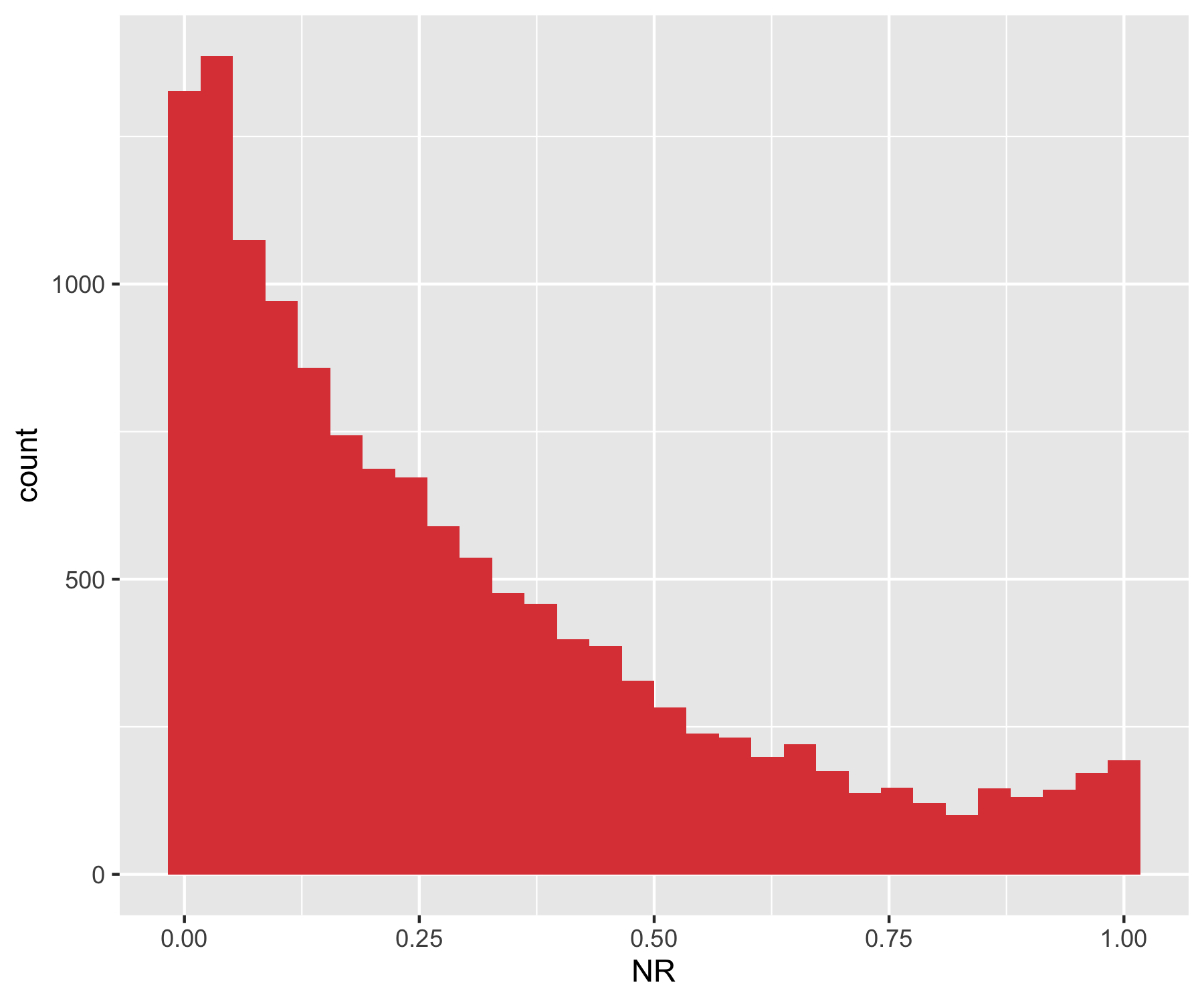}}
\subfigure{\includegraphics[width=0.4\textwidth]{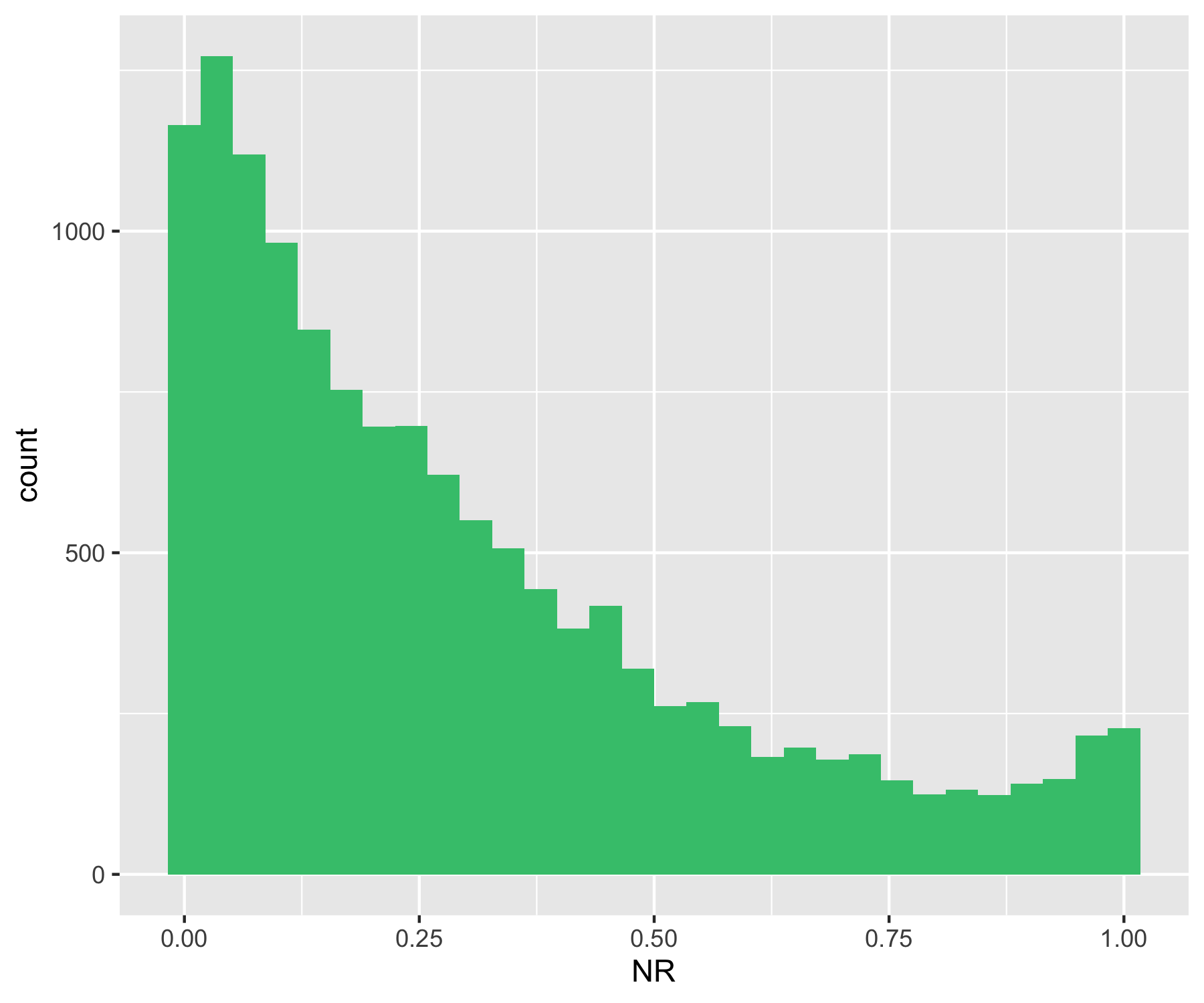}}
\caption{Performance of our algorithm for vertex nomination between
the two Microsoft Bing entities transition networks on $n=13535$
vertices. For each $x\in V_1$ we randomly selected 10 seeds and our
algorithms are applied on the 1-neighborhood of them. We record the
normalized rank (NR) of its correspondence $\sigma(x)$ in the whole
nomination list.  The red and green histogram of NR correspond to the
case where the graphs embeddings are aligned via orthogonal Procrustes
and via the adaptive point set registration procedure, respectively.}
\label{Bing_1neighborhood}
\end{figure}

For our second analysis we do not sub-sample the graphs and hence $G_1$ contains
$13535$ vertices and $519389$ edges while $G_2$ contains the same $13535$
vertices and $595047$ edges. Once again we sequentially consider each vertex $x \in G_1$ is
as the vertex of interest. For a given vertex of interest we randomly select $10$ other
vertices as seeds. We then take the induced subgraph in $G_1$
(respectively $G_2$) formed by the
$1$-neighborhood of these $11$ vertices (the vertex of interest and
the $10$ seed vertices). Letting $G_1(v)$ and $G_2(v)$ denote the
induced subgraph of $G_1$ and $G_2$, we then apply our algorithm to find the nomination
correspondence $\sigma(x) \in G_2(v)$ for the given vertex of interest $x \in
G_1(v)$. The histogram of the NR are summarized in
Figure~\ref{Bing_1neighborhood}. Figure~\ref{Bing_1neighborhood}
indicates that, even though there is no longer an exact 1-to-1 correspondence between the
vertices set in the induced $1$-neighborhood subgraphs, our algorithms still give accurate nomination lists.
\begin{table}[pth]
\centering
\begin{tabular}{c|cccccccc}
 & 1\% & 5\%& 10\%&25\%& 50\%& 75\%&95\%&99\% \\
 \hline
Procrustes (10 seeds) & 0.000& 0.006 &0.018 &0.073& 0.211& 0.429& 0.869 & 0.991   \\
set registration (10 seeds) &  0.000 &0.008 &0.022& 0.081& 0.220& 0.440& 0.894 &0.992\\
\end{tabular}
\caption{Quantile levels of normalized rank (NR) values for vertex nomination
  with the Bing entity networks on $n = 13535$ vertices}
\label{tab:Bing2}
\end{table}

Appendix~\ref{sec:add res} contains additional real data analysis
results illustrating how the choice of embedding dimension $d$ and the
penalty parameter $\lambda$ affects the performance of our algorithm
for the high school friendship data and the Bing data; see
Figure~\ref{highschool_shared_Algo1_d}, Table~\ref{tab:Bing1_d},
Figure~\ref{highschool_shared_Algo1_lambda} and
Table~\ref{tab:Bing1_lambda}. In addition,
Figure~\ref{highschool_shared_Algo1_rerank} and
Table~\ref{tab:Bing1_reranking} illustrate how the reranking step
also improves the performance for these real data applications; for
example for the Microsoft Bing data, Table~\ref{tab:Bing1_reranking}
in Appendix~\ref{sec:real data reranking} shows
that with $100$ seeds the $25$th percentile of the normalized rank
changes from $0.073$ to $0.027$ after we apply the reranking step. It
means that without the reranking step, $25\%$ of the time we can find
the corresponding vertex among the first $73$ vertices of the
nomination list, and with the reranking step, $25\%$ of the time we
can find the corresponding vertex among the first $27$ vertices of the
nomination list. This a substantial improvement and indicates that
while we do not expect the Microsoft Bing entity graph transitions
data to follows a $\rho$-correlated GRDPG model, the $\rho$-correlated
GRDPG model still provides a useful surrogate for analyzing the
pairwise correlations between the edges of the two Bing graphs.
Finally, Appendix~\ref{sec:add res} also presents comparisons between
our algorithm and the algorithm of \cite{agterberg2020vertex}.  For
the high school friendship data, we see from
Figure~\ref{highschool_shared_Algo1_count_compare} that our algorithm
is much more accurate and we note the running time of our algorithm
with orthogonal Procrustes is only about $7$ minutes compared to their
running time which is roughly $200$ minutes on the same
laptop. Meanwhile, for the Bing data, the algorithm in
\cite{agterberg2020vertex} and our algorithm using either the orthogonal
Procrustes alignment or adaptive point set registration alignment have similar normalized rank distribution (see
Table~\ref{tab:Bing_compare}), but \cite{agterberg2020vertex} algorithm
is roughly $8$ times slower than our algorithm. 

\section{Conclusion} In summary, the current paper provides an
algorithm for solving the vertex
nomination problem in the multi-graphs setting. Our algorithm depends
on adjacency spectral embedding and followed by solving a quadratic
programming. To eliminate non-identifiability of spectral embedding
for different graphs, besides an approach based on orthogonal
Procrustes, we propose a method using adaptive point set registration
to align the embedding that also work without needing any information
about seed vertices. Under mild assumption, we establish theoretical
guarantee on the consistency of our nomination scheme. The empirical
results on the simulation and real data analysis demonstrate that our
algorithm is generally quite accurate even when there are only a few
seeds or even no seed vertices. As we allude to in the introduction of
this paper, vertex nomination is an unsupervised learning problem and
thus evaluation of a vertex nomination algorithm usually
requires some underlying ground truth. The real data analysis
examples of this paper are based on pairs of graphs with shared
vertices and we used these shared vertices to define our groundtruth;
the resulting analysis is thus similar to network
deanonymization. When there is no known groundtruth, then our
proposed methodology can be used for exploratory data analysis or for
suggesting possible matches between a query vertex $x$ in one graph
and vertices most ``similar'' to $x$ in the second graph.
To evaluate the accuracy of the resulting nominations will, however, require additional domain knowledge or domain
experts. We believe that our chosen examples are simple to describe and yet
sufficiently rich in scope, thereby providing a clear and compelling illustration of
the effectiveness of our proposed methodology.

While the proposed algorithm is reasonably computationally efficient,
there are still technical challenges in applying the algorithm to large
graphs. For example, the Bing graphs analyzed in this paper are on the order of $10^4$ vertices and $10^5$
to $10^6$ edges and our algorithm takes roughly $30$ minutes for one full
analysis when running on a consumer laptop. For larger-scale graphs, such
as those on $10^5$ vertices and $10^7$ edges, our algorithm breaks
down. In particular, the EM steps in the adaptive point set
registration algorithm can be quite slow to converge and thus might
require sub-sampling of the embedded points before performing the
alignment. Furthermore, the quadratic programming
step requires keeping track of the assignment matrix $\mathbf{D}$; a
naive approach of storing $\mathbf{D}$ will require too much memory,
especially since $\mathbf{D}$ is likely to be sparse throughout the
optimization. Development of iterative procedures for storing and
updating $\mathbf{D}$ is thus essential for scaling our
algorithm to large graphs. We leave these investigations for future work.

\bibliography{Project}
\newpage
\appendix

\renewcommand\thesection{\Alph{section}} 
\renewcommand\thesubsection{\Alph{section}.\arabic{subsection}} 
\setcounter{section}{0}

\renewcommand\thefigure{\Alph{section}\arabic{figure}} 
\setcounter{figure}{0}

\renewcommand\thetable{\Alph{section}\arabic{table}} 
\setcounter{table}{0}

\section{Proof of Proposition \ref{prop1}}

The first part of Proposition \ref{prop1} has been proved in
Section~\ref{sec:theorey}. Now we prove the second part, i.e., we will
show that for a fixed $n$, as $\lambda \rightarrow 0$, we have
  $$\mathbf{D}_{\lambda} \longrightarrow \argmin_{\mathbf{D} \in
  \mathcal{D}} \{ \|\mathbf{D}\|_{F} \colon \langle \mathbf{C}, \mathbf{D} \rangle
= \xi_*\},$$
where $\xi_*$ is the minimum value achieved in $P_0$.

The following argument is adapted from the proof of Proposition 4.1 in
\cite{peyre2019computational}. We consider a sequence
$(\lambda_{\ell})_{\ell}$ such that $\lambda_{\ell} \rightarrow 0$ and
$\lambda_{\ell}>0 .$ Since $\mathcal{D}$ is bounded, we can extract a
sequence (that we do not relabel for the sake of simplicity) such that
$\mathbf{D}_{\lambda_\ell} \rightarrow \mathbf{D}_{\star}.$ Since
$\mathcal{D}$ is closed, $\mathbf{D}_{\star} \in \mathcal{D} .$ We
consider any $\mathbf{D}$ such that $\langle\mathbf{C},
\mathbf{D}\rangle=\xi_{*}$. 
By optimality of such $\mathbf{D}$ and
$\mathbf{D}_{\lambda_\ell}$ for their respective optimization
problems, we have
\begin{equation}
  \label{eq:lambda}
0 \leqslant\left\langle\mathbf{C}, \mathbf{D}_{\lambda_\ell}\right\rangle-\langle\mathbf{C}, \mathbf{D}\rangle \leqslant \lambda_{\ell}\cdot\left(\|\mathbf{D}\|_F-\|\mathbf{D}_{\lambda_\ell}\|_F\right).
\end{equation}

Since $\|\cdot\|_F$ is continuous, taking the limit $\ell
\rightarrow+\infty$ in this expression shows that
$\left\langle\mathbf{C},
\mathbf{D}_\star\right\rangle=\langle\mathbf{C}, \mathbf{D}\rangle$ so
that $\mathbf{D}_{\star}$ is a feasible point of $\{\md \colon \langle
\mathbf{C}, \mathbf{D} \rangle = \xi_*\}$. Furthermore, dividing by
$\lambda_{\ell}$ in Eq. (\ref{eq:lambda}) and taking the limit shows
that $\|\mathbf{D}_\star\|_F \leqslant \|\mathbf{D}\|_F,$ which shows
that $\mathbf{D}_{\star}$ is a solution of $\argmin_{\mathbf{D} \in
\mathcal{D}} \{ \|\mathbf{D}\|_{F} \colon \langle \mathbf{C},
\mathbf{D} \rangle = \xi_*\}$.

Finally we prove the claim that if the vertices $\{1,2,\dots,n\}$ can
be partitioned into $K$ distinct groups/blocks such that $c_{ij'}>c_{ij}+\omega(n^{-1/2})$
for all triplets $(i,j,j')$ with $i$ and $j$ being the same group and
$i$ and $j'$ being in different groups, then with high probability $\hat{\mathbf{D}}_0$ is
block diagonal, i.e., with high probability $\hat{\mathbf{D}}_0(i,j) =
0$ whenever $i$ and $j$ are in different groups. 

Suppose that $\hat{\mathbf{D}}_0$ is not block diagonal. Since
  $n^{-1}\hat{\mathbf{D}}_0$ is a doubly stochastic matrix, by the
  Birkhoff-von Neumann theorem we can write $n^{-1} \hat{\mathbf{D}}_0$ as a convex
  combination of permutation matrices, i.e., $\hat{\mathbf{D}}_0 = n
  \sum_{\sigma} \lambda_{\sigma} \bm{\Pi}_{\sigma}$ where the sum is
  over all permutations $\sigma$ of $\{1,2,\dots,n\}$, the
  $\bm{\Pi}_{\sigma}$ represent permutation matrices
  corresponding to the permutations $\sigma$ and
  $\{\lambda_{\sigma}\}$ are the coefficients for the convex
  combination.  
  We therefore have
  \begin{equation}
    \label{eq:claim0}
  \langle \hat{\mathbf{C}}, \hat{\mathbf{D}}_0 \rangle =
  n \sum_{\sigma} \langle \hat{\mathbf{C}}, \lambda_{\sigma}
  \bm{\Pi}_{\sigma} \rangle
  \geq
  n \min_{\sigma} \langle \hat{\mathbf{C}}, \bm{\Pi}_{\sigma} \rangle.
  \end{equation}

  Now according to the condition that $c_{ij'} > c_{ij}$ with
  $c_{ij'} - c_{ij} = \omega(n^{-1/2})$ for all triplets $(i,j,j')$
  where $i$ and $j$ belong to the same group and $i$ and $j'$
  belonging to different groups, we have, using
  Eq.~(\ref{eq:bound_main}), that with high probability
  $\hat{c}_{ij'} > \hat{c}_{ij}$ for all triplets $(i,j,j')$
  satisfying the above groups condition. We will now assume that this
  condition $\hat{c}_{ij'} > \hat{c}_{ij}$ holds. 

  We recall that any permutation $\sigma$ can be decomposed
  into a product of cycles $\mathcal{C}_1, \mathcal{C}_2, \dots,
  \mathcal{C}_m$ for some $m \geq 1$. Let $\mathcal{C}_r$ be an
  arbitrary cycle and suppose that $\mathcal{C}_r$ is of length
  $s \geq 2$; note that if $\mathcal{C}_r$ has length $1$ then
  $\mathcal{C}_r$ correspond to a fixed point so that
  $\bm{\Pi}_{\sigma}(i,j) = 1$ if and only if $j = i$ where $i$ is the
  sole index appearing in $\mathcal{C}_r$. Let $i_1$ be the smallest index from
  $\{1,2,\dots,n\}$ that appears in the cycle $\mathcal{C}_r$. Then
  the remaining indices appearing in  $\mathcal{C}_r$ are of the form
  $i_2 = \sigma(i_1), i_3 = \sigma(i_2), \dots i_{s} =
  \sigma(i_{s-1}), i_1 = \sigma(i_{s})$. Now if all the indices
  appearing in $\mathcal{C}_r$ are from the same group then we are
  done. Otherwise, we split the cycle $\mathcal{C}_r$ into smaller
  cycles $\{\mathcal{C}_{r1}, \mathcal{C}_{r2}, \dots, \mathcal{C}_{rK}\}$ where each
  cycle $\mathcal{C}_{rk}$ only have indices appearing in group
  $k$. The order of elements in each cycle $\mathcal{C}_{rk}$ are
  arranged according to the order in which they appear within the
  original cycle $\mathcal{C}_r$; for example if
  $\mathcal{C}_{r}$ contains five elements $(i_1, i_2, \dots, i_5)$ with $i_1, i_2, i_4$
  belonging to group $1$ and $i_3, i_5$ belonging to group $2$ then
  $\mathcal{C}_{r1} = (i_1, i_2, i_4)$ and $\mathcal{C}_{r2} = (i_3,i_5)$. 
  Let $\sigma_{1}, \sigma_{2}, \dots,
  \sigma_{K}$ be the permutations corresponding to these cycles; note
  that there could be empty cycles $\mathcal{C}_{rk}$ in which case the corresponding
  $\sigma_{k}$ can be dropped or ignored. We then have
  \begin{equation}
    \label{eq:claim1}
   \sum_{i \in \mathcal{C}_r} \hat{c}_{i \sigma(i)} \geq \sum_{k=1}^{K} \sum_{i
     \in \mathcal{C}_{rk}} \hat{c}_{i \sigma_{k}(i)}
   \end{equation}
 with strict inequality unless $\mathcal{C}_r$ only have elements
  from the same group.

  The justification for Eq.~(\ref{eq:claim1}) is as follows. If $i$
  and $\sigma(i)$ are from the same
  group then $\sigma_{k}(i) = \sigma(i)$ and hence $\hat{c}_{i \sigma(i)} = \hat{c}_{i \sigma_{k}(i)}$. Otherwise, if $i$ and $\sigma(i)$ are from
  different group, say group $k$ and $k'$, then at the time 
  right before the edge $(i, \sigma(i))$ is considered the index $j =
  \sigma(i)$ has in-degree $0$ and out degree $0$
  while the index $i$ has in degree $1$ and out degree $0$. 
  The splitting of $\mathcal{C}_r$ into smaller cycles
  $\mathcal{C}_{r1}, \dots, \mathcal{C}_{rK}$ is equivalent to
  changing the outgoing edge for index $i$ and the incoming edge
  for index $j = \sigma(i)$ using the following rules. 
  \begin{enumerate}
    \item If $i$ is the last index in $\mathcal{C}_{rk}$ then we had replace $c_{ij}$
  with the strictly smaller cost $c_{i i_*}$ where $i_* = \sigma_k(i)$ is the first index that appears in
  $\mathcal{C}_{rk}$. Note that $i_* = i$ is a possibility if
  $\mathcal{C}_{rk}$ contains only a single index. 
  \item If $i$ is not the last index in $\mathcal{C}_{rk}$ 
  then we had replace the cost $c_{ij}$ with the strictly smaller
  cost $c_{ii*}$ where $i_* = \sigma_k(i)$ is the index appearing right after $i$ in
  $\mathcal{C}_{rk}$. 
\item If $j$ is the first index in $\mathcal{C}_{rk'}$ then we had
  replaced $c_{ij}$ with the strictly smaller $c_{j_*j}$ where $j_* =
  \sigma_{k'}^{-1}(j)$
  is the last index that
  appears in $\mathcal{C}_{rk'}$. There is once again the possibility
  that $j_* = j$.
  . 
\item If $j$ is not the first index in $\mathcal{C}_{rk'}$ then we had
  replaced $c_{ij}$ with
  the strictly smaller $c_{j_*j}$ where $j_* = \sigma_{k'}^{-1}(j)$ is the index appearing right before $j$
  in $\mathcal{C}_{rk'}$.
\end{enumerate} 
Note
that in the above steps we neither change the {\em outgoing} edge of
any index $i'$ appearing before $i$ nor change the {\em incoming} edge
of any index $j'$ appearing before $j$.  The sum of the costs $c_{ij}$
for both the incoming and outgoing edges for any cycle $\mathcal{C}_r$
is twice the total cost of the cycle, and hence, by going sequentially
through the indices in the cycle $\mathcal{C}_r$ and applying the
above rules we will never increase the cost for any outgoing edge or
incoming edge; the total sum of the cost for the smaller cycles
$\mathcal{C}_{rk}$ is thus strictly less than that for
$\mathcal{C}_r$, with equality if and only if all of the indices
appearing in $\mathcal{C}_r$ are from a single group.

The above reasoning implies that the permutation
$\sigma_*$ which minimizes $\langle \hat{\mathbf{C}},
\bm{\Pi}_{\sigma} \rangle$ over all permutation $\sigma$ is a
union of disjoint cycles, each of which contains indices from a single
group, i.e., if $\sigma_*$ minimizes $\langle \hat{\mathbf{C}},
\bm{\Pi}_{\sigma} \rangle$ then $\bm{\Pi}_{\sigma_*}(i,j)
= 0$ whenever $i$ and $j$ belong to different groups. Let
$\mathcal{S}$ be the set of all permutations $\sigma$ for which $\bm{\Pi}_{\sigma}(i,j)
= 0$ whenever $i$ and $j$ belong to different groups. Then by
Eq.~(\ref{eq:claim0}), 
$\hat{\mathbf{D}}_0 = n\sum_{\sigma} \lambda_{\sigma}
\bm{\Pi}_{\sigma}$ minimizes $\langle \hat{\mathbf{C}}, \mathbf{D}
\rangle$ over the simplex constraint if and only if $\hat{\mathbf{D}}_0$
is a convex combination of elements in $\mathcal{S}$ and hence
$\hat{\mathbf{D}}_0(i,j) = 0$ whenever $i$ and $j$ belong to different
groups. We had thus justified our last claim from Proposition~\ref{prop1}.

\section{Proof of Theorem \ref{thm:main}}
The following argument is adapted from the proof of Theorem~5 in \cite{rubin2017statistical} for
bounding $\|\hat{\mathbf{X}} - \mathbf{X}\|_{2 \to \infty}$ in the
case of a {\em single} generalized random dot product graph to the
current setting of bounding $\|\hat{\mathbf{X}}_1 -
\hat{\mathbf{X}}_2\|_{2 \to \infty}$ for a {\em pair} of {\em correlated} generalized random dot
product graphs. 

We set the block spectral decomposition of the symmetric matrix $\ma_1$ as 
$
\ma_1=[\muu_1|\muu_1'][\ms_1\oplus\ms_1'][\muu_1|\muu_1']^\top=\muu_1\ms_1\muu_1^\top+\muu_1'\ms_1'\muu_1'^\top,
$
where the diagonal matrix $\ms_1\in\real^{d\times d}$ contains the $d$ largest-in-magnitude nonzero eigenvalues of $\ma_1$. Similarly,
$
\ma_2=\muu_2\ms_2\muu_2^\top+\muu_2'\ms_2^\prime\muu_2'^\top.
$
According to the definition of adjacency spectral embedding, we know $\hat\mx_1=\muu_1|\ms_1|^{\frac{1}{2}}$ and $\hat\mx_2=\muu_2|\ms_2|^{\frac{1}{2}}$. So our goal is to prove
$$
\min_{\mw\in\mathbb{O}_{d}}\left\|\muu_1|\ms_1|^{\frac{1}{2}}\mw-\muu_2|\ms_2|^{\frac{1}{2}}\right\|_{2\to\infty}=(1-\rho)^{1/2}\cdot O_p\left(n^{-1/2}\right)+O_p\left((\log n)^{2c}n^{-1}\gamma^{-1/2}\right).
$$

We set $\mw^*={\mw_1}^\top\mw_2$,where $\mw_1,\mw_2$ are two orthogonal matrices and we will give their specific formula in the following proof. 
Let $\mpp=\muu\ms\muu^\top$ be the eigendecomposition of $\mpp$, where
$\muu\in\real^{n\times d}$ is the matrix whose columns are the
eigenvectors and the diagonal matrix $\ms\in\real^{d\times d}$ contains all the $d$ nonzero eigenvalues of $\mpp$.
Now we split $\muu_1|\ms_1|^{\frac{1}{2}}\mw^*-\muu_2|\ms_2|^{\frac{1}{2}}$ as
\begin{equation}
  \label{eq:T1-6}
\begin{split}
\muu_1 |\ms_{1}|^{\frac{1}{2}} \mw^*-\muu_{2} |\ms_{2}|^{\frac{1}{2}}
=&\underbrace{\left[\muu_{1} |\ms_{1}|^{\frac{1}{2}} \mw_1^\top\mw_2-\muu\muu^\top\muu_{1} |\ms_{1}|^{\frac{1}{2}} \mw_1^\top\mw_2\right] }_{\mt_1}\\
&+\underbrace{\left[\muu\muu^\top\muu_{1} |\ms_{1}|^{\frac{1}{2}} \mw_1^\top\mw_2-\muu|\ms|^{\frac{1}{2}}\muu^\top\muu_{1} \mw_1^\top\mw_2\right] }_{\mt_2}\\
&+\underbrace{\left[\muu|\ms|^{\frac{1}{2}}\muu^\top\muu_{1} \mw_1^\top\mw_2-\muu |\ms|^{\frac{1}{2}}\mw_1\mw_1^\top\mw_2\right]}_{\mt_3} \\
&+\underbrace{\left[\muu |\ms|^{\frac{1}{2}}\mw_2-\muu|\ms|^{\frac{1}{2}}\muu^\top\muu_2\right]}_{\mt_4}\\
&+\underbrace{\left[\muu|\ms|^{\frac{1}{2}}\muu^\top\muu_2-\muu\muu^\top\muu_2|\ms_2|^{\frac{1}{2}}\right]}_{\mt_5}\\
&+\underbrace{\left[\muu\muu^\top\muu_2|\ms_2|^{\frac{1}{2}}-\muu_2|\ms_2|^{\frac{1}{2}}\right]}_{\mt_6}.
\end{split}
\end{equation}
By Lemma~\ref{lm:newT2T5}, Lemma~\ref{lm:newT3T4} and Lemma~\ref{lm:newT1T6}, we have that for some constant $c>0$
$$
\begin{aligned}
	&\|\mt_2+\mt_3+\mt_4+\mt_5\|_{2\to\infty}=O_p\left(n^{-1}\gamma^{-1/2}\right),\\
&\|\mt_1+\mt_6\|_{2\to\infty}=(1-\rho)^{1/2}\cdot O_p\left(n^{-1/2}\right)+O_p\left((\log n)^{2c}n^{-1}\gamma^{-1/2}\right).
\end{aligned}
$$
Theorem~\ref{thm:main} then follows immediately. 

\newtheorem{lemma}{Lemma}

\begin{lemma}
  \label{lm:newT2T5}
  For the term $\mt_2,\mt_5$ in Eq.(\ref{eq:T1-6}), we have
$$
\|\mt_2\|_{2\to\infty}=O_p\left(n^{-1}\gamma^{-1/2}\right),\quad\|\mt_5\|_{2\to\infty}=O_p\left(n^{-1}\gamma^{-1/2}\right).
$$
\end{lemma}

\begin{proof}
	
For $\mt_2$ we have
$$\begin{aligned}
\left\|\mt_2 \right\|_{2 \rightarrow \infty} 
& \leqslant\left\|\muu\right\|_{2\to \infty} \cdot \left\|\muu^\top\muu_1|\ms_1|^{\frac{1}{2}}-|\ms|^{\frac{1}{2}}\muu^\top\muu_1\right\|_{2}\cdot \|\mw_1^\top\mw_2\|_2 \\
& \leqslant\left\|\muu\right\|_{2\to \infty} \cdot \left\|\muu^\top\muu_1|\ms_1|^{\frac{1}{2}}-|\ms|^{\frac{1}{2}}\muu^\top\muu_1\right\|_{2}.
\end{aligned}$$
For the first part, we have $\left\|\muu\right\|_{2\to \infty}=O_p\left(n^{-1/2}\right)$.
For the second part, we notice that for any $i,j=1,\dots, d$, the entry $ij$ of it can be written as 
$$
\begin{aligned}
(\muu^\top\muu_1|\ms_1|^{\frac{1}{2}}-|\ms|^{\frac{1}{2}}\muu^\top\muu_1)_{i,j}
&=\left[\muu^\top\muu_1\right]_{i,j}\cdot\left(\sqrt{|\lambda_{j}(\ma_1)}|-\sqrt{|\lambda_{i}(\mpp)|}\right).
\end{aligned}
$$
So for $i\leqslant p, j\leqslant p$,
$$
\begin{aligned}
(\muu^\top\muu_1|\ms_1|^{\frac{1}{2}}-|\ms|^{\frac{1}{2}}\muu^\top\muu_1)_{i,j}
&=\left[\muu^\top\muu_1\right]_{i,j}\cdot\left(\sqrt{\lambda_{j}(\ma_1)}-\sqrt{\lambda_{i}(\mpp)}\right)\\
&=\left[\muu^\top\muu_1\right]_{i,j}\cdot\left(\lambda_{j}(\ma_1)-\lambda_{i}(\mpp)\right)\cdot\left(\sqrt{\lambda_{j}(\ma_1)}+\sqrt{\lambda_{i}(\mpp)}\right)^{-1}.
\end{aligned}
$$
Similarly, we have for $i> p, j> p$,
$$
\begin{aligned}
(\muu^\top\muu_1|\ms_1|^{\frac{1}{2}}-|\ms|^{\frac{1}{2}}\muu^\top\muu_1)_{i,j}
&=\left[\muu^\top\muu_1\right]_{i,j}\cdot\left(\lambda_{j}(\ma_1)-\lambda_{i}(\mpp)\right)\cdot\left(-\sqrt{-\lambda_{j}(\ma_1)}-\sqrt{-\lambda_{i}(\mpp)}\right)^{-1}.
\end{aligned}
$$
For $i> p, j\leqslant p$,
$$
\begin{aligned}
(\muu^\top\muu_1|\ms_1|^{\frac{1}{2}}-|\ms|^{\frac{1}{2}}\muu^\top\muu_1)_{i,j}
=&\left[\muu^\top\muu_1\right]_{i,j}\cdot\left(\lambda_{j}(\ma_1)-\lambda_{i}(\mpp)\right)\cdot\left(\sqrt{\lambda_{j}(\ma_1)}+\sqrt{-\lambda_{i}(\mpp)}\right)^{-1}\\
&+2\left[\muu^\top\muu_1\right]_{i,j}\cdot\lambda_{i}(\mpp)\cdot\left(\sqrt{\lambda_{j}(\ma_1)}+\sqrt{-\lambda_{i}(\mpp)}\right)^{-1}.
\end{aligned}
$$
For $i\leqslant p, j> p$,
$$
\begin{aligned}
(\muu^\top\muu_1|\ms_1|^{\frac{1}{2}}-|\ms|^{\frac{1}{2}}\muu^\top\muu_1)_{i,j}
=&\left[\muu^\top\muu_1\right]_{i,j}\cdot\left(\lambda_{j}(\ma_1)-\lambda_{i}(\mpp)\right)\cdot\left(-\sqrt{-\lambda_{j}(\ma_1)}-\sqrt{\lambda_{i}(\mpp)}\right)^{-1}\\
&-2\left[\muu^\top\muu_1\right]_{i,j}\cdot\lambda_{i}(\mpp)\cdot\left(\sqrt{-\lambda_{j}(\ma_1)}+\sqrt{\lambda_{i}(\mpp)}\right)^{-1}.
\end{aligned}
$$
We define matrices $\mh_1,\mh_2\in\real^{d\times d}$ as
$$
\begin{aligned}
	&(\mh_1)_{i,j}=\left(\sqrt{|\lambda_{j}(\ma_1)|}+\sqrt{|\lambda_{i}(\mpp)|}\right)^{-1}\cdot \mathbb{I}(j\leqslant p)+\left(-\sqrt{|\lambda_{j}(\ma_1)|}-\sqrt{|\lambda_{i}(\mpp)|}\right)^{-1}\cdot \mathbb{I}(j>p),\\
	&(\mh_2)_{i,j}=\lambda_{i}(\mpp)\cdot\left(\sqrt{|\lambda_{j}(\ma_1)|}+\sqrt{|\lambda_{i}(\mpp)|}\right)^{-1}.
\end{aligned}
$$
According to Lemma \ref{lm:eigenPandA}, $\lambda_i(\ma_1),\lambda_j(\mpp)=O_p(n\gamma),\Omega_p(n\gamma)$ for any $i,j\leq d$, it follows that
$$
(\mh_1)_{i,j}=O_p\left(\frac{1}{\sqrt{n\gamma}}\right),
(\mh_2)_{i,j}=O_p\left(\sqrt{n\gamma}\right).
$$
Letting $\circ$ denote the Hadamard matrix product, we arrive at the decomposition
\begin{equation}
  \label{eq:V1}
\muu^\top\muu_1|\ms_1|^{\frac{1}{2}}-|\ms|^{\frac{1}{2}}\muu^\top\muu_1=(\muu^\top\muu_1\ms_1-\ms\muu^\top\muu_1)\circ\mh_1+ \mv\circ \mh_2,
\end{equation}
where $$\mv=\left(\begin{array}{cc}
0 & -2\mathbf{U}_{(+)}^{\top} \mathbf{U}_{1(-)} \\
2\mathbf{U}_{(-)}^\top \mathbf{U}_{1(+)} & 0
\end{array}\right).$$
The definition of $\muu_{(+)},\muu_{(-)},\muu_{1(+)},\muu_{1(-)}$ can be found in Lemma \ref{lm:UU-W}, and from the proof of Lemma \ref{lm:UU-W}, we know
$\|\mathbf{U}_{(+)}^{\top} \mathbf{U}_{1(-)}\|_F=O_p\left(\frac{1}{n{\gamma}}\right)$, $
\|\mathbf{U}_{(-)}^\top \mathbf{U}_{1(+)}\|_F=O_p\left(\frac{1}{n{\gamma}}\right)$.
So we have $\left\|\mv\right\|_F=O_p\left(\frac{1}{n{\gamma}}\right)$. It follows that
$$
\|\mv\circ \mh_2\|_2\leqslant d\cdot\|\mh_2\|_{max}\cdot \|\mv\|_F=O_p\left(\sqrt{n\gamma}\right)\cdot O_p\left(\frac{1}{n{\gamma}}\right)=O_p\left(\frac{1}{\sqrt{n\gamma}}\right).
$$
According to Lemma \ref{lm:UUS-SUU}, $\|\muu^\top\muu_1\ms_1-\ms\muu^\top\muu_1\|_F=O_p(1)$, thus we have
$$
\|(\muu^\top\muu_1\ms_1-\ms\muu^\top\muu_1)\circ\mh_1\|_2\leqslant d\cdot\|\mh_1\|_{max}\cdot \|\muu^\top\muu_1\ms_1-\ms\muu^\top\muu_1\|_F=O_p\left(\frac{1}{\sqrt{n\gamma}}\right)\cdot O_p(1)=O_p\left(\frac{1}{\sqrt{n\gamma}}\right).
$$
So we bound the spectral norm of $\muu^\top\muu_1|\ms_1|^{\frac{1}{2}}-|\ms|^{\frac{1}{2}}\muu^\top\muu_1$ as
$$
\|\muu^\top\muu_1|\ms_1|^{\frac{1}{2}}-|\ms|^{\frac{1}{2}}\muu^\top\muu_1\|_2=O_p\left(\frac{1}{\sqrt{n\gamma}}\right).
$$
We finally conclude
$$
\begin{aligned}
\|\mt_2\|_{2\to\infty} 
&\leqslant O_p\left(\frac{1}{\sqrt{n}}\right)\cdot O_p\left(\frac{1}{\sqrt{n\gamma}}\right)\cdot 1
=O_p\left(\frac{1}{n\sqrt{\gamma}}\right).
\end{aligned}
$$

The proof for $\|\mt_5\|_{2\to\infty}$ is almost the same.
\end{proof}

\begin{lemma}
  \label{lm:newT3T4}
  For the term $\mt_3,\mt_4$ in Eq.~\eqref{eq:T1-6}, we have
$$
\|\mt_3\|_{2\to\infty}=O_p\left(n^{-1}\gamma^{-1/2}\right),\quad\|\mt_4\|_{2\to\infty}=O_p\left(n^{-1}\gamma^{-1/2}\right)
.$$
\end{lemma}

\begin{proof}
For $\mt_3$ we have
$$
\begin{aligned}
\|\mt_3\|_{2\to\infty} 
&\leqslant \|\muu\|_{2\to\infty}\cdot\||\ms|^{\frac{1}{2}}\|_2\cdot\|\muu^\top\muu_1-\mw_1\|_2\cdot\|\mw_1^\top\mw_2\|_2\\
&\leqslant \|\muu\|_{2\to\infty}\cdot\||\ms|^{\frac{1}{2}}\|_2\cdot\|\muu^\top\muu_1-\mw_1\|_2.
\end{aligned}
$$
We notice $\left\|\muu\right\|_{2\to \infty}=O_p\left(n^{-1/2}\right)$.
From Lemma \ref{lm:eigenPandA}, 
we have $ {\|\mpp\|_{2}}=O_p(n \gamma)$, hence $\||\ms|^{\frac{1}{2}}\|_2=(\||\ms|\|_2)^{\frac{1}{2}} = (\|\mpp\|_{2})^{\frac{1}{2}}=O_p(\sqrt{n \gamma})$.
And according to Lemma \ref{lm:UU-W}, we have
$
\|\muu^\top\muu_1-\mw_1\|_2=O_p\left(\frac{1}{n\gamma}\right)
$.
We immediately conclude
$$
\|\mt_3\|_{2\to\infty}
=O_p\left(\frac{1}{\sqrt{n}}\right)\cdot O_p(\sqrt{n \gamma})\cdot O_p\left(\frac{1}{n\gamma}\right)
=O_p\left(\frac{1}{n\sqrt{\gamma}}\right).
$$

The proof for $\|\mt_4\|_{2\to\infty}$ is almost the same.
\end{proof}

\begin{lemma}
  \label{lm:newT1T6}
  For the term $\mt_1,\mt_6$ in Eq.\eqref{eq:T1-6}, we have
$$
\|\mt_1+\mt_6\|_{2\to\infty}=(1-\rho)^{1/2}\cdot O_p\left(n^{-1/2}\right)+O_p\left((\log n)^{2c}n^{-1}\gamma^{-1/2}\right)
$$ for some constant $c$.
\end{lemma}

\begin{proof}
According to B.2.4 in \cite{rubin2017statistical},
$$
{\mathbf{U}_1}|{\mathbf{S}_1}|^{\frac{1}{2}}=\mathbf{U}|\mathbf{S}|^{\frac{1}{2}} \mathbf{W}_1+(\mathbf{A}_1-\mathbf{P}) \mathbf{U}|\mathbf{S}|^{-\frac{1}{2}} \mathbf{W}_1 \mathbf{I}_{p, q}+\mathbf{R}_1,
$$
for some (residual) matrix $\mathbf{R}_1 \in \mathbb{R}^{n \times d}$ satisfying $\|\mathbf{R}_1\|_{2 \rightarrow \infty}=O_p\left(\frac{(\log n)^{2 c_1}}{n^{}\gamma^{1 / 2}}\right)$ for some constant $c_1>0$.
And we notice that
$
(\mi-\muu\muu^\top)\muu=0
$,
then
$$
\begin{aligned}
	\mt_1
	&=(\mi-\muu\muu^\top){\mathbf{U}_1}|{\mathbf{S}_1}|^{\frac{1}{2}}\mw_1^\top\mw_2\\
&= (\mi-\muu\muu^\top)(\mathbf{A}_1-\mathbf{P}) \mathbf{U}|\mathbf{S}|^{-\frac{1}{2}} \mathbf{W}_1 \mathbf{I}_{p, q}\mw_1^\top\mw_2+(\mi-\muu\muu^\top)\mr_1\mw_1^\top\mw_2\\
&= (\mi-\muu\muu^\top)(\mathbf{A}_1-\mathbf{P}) \mathbf{U}|\mathbf{S}|^{-\frac{1}{2}} \mathbf{W}_2 \mathbf{I}_{p, q}+(\mi-\muu\muu^\top)\mr_1\mw_1^\top\mw_2.
\end{aligned}
$$
With the similar proof, we have
$$
\begin{aligned}
	\mt_6
&= -(\mi-\muu\muu^\top)(\mathbf{A}_2-\mathbf{P}) \mathbf{U}|\mathbf{S}|^{-\frac{1}{2}} \mathbf{W}_2 \mathbf{I}_{p, q}-(\mi-\muu\muu^\top)\mr_2,
\end{aligned}
$$ 
where $\mr_2$ satisfies $\|\mathbf{R}_2\|_{2 \rightarrow \infty}=O_p\left(\frac{(\log n)^{2 c_2}}{n^{}\gamma^{1 / 2}}\right)$ for some constant $c_2>0$.
We set $\me=\ma_1-\ma_2$. We therefore have 
$$
\begin{aligned}
	\mt_1+\mt_6
=& (\mi-\muu\muu^\top)\me\mathbf{U}|\mathbf{S}|^{-\frac{1}{2}} \mathbf{W}_2 \mathbf{I}_{p, q}+\mr,
\end{aligned}
$$
where we set
$$\mr=\mr_1\mw_1^\top\mw_2-\muu\muu^\top\mr_1\mw_1^\top\mw_2-\mr_2+\muu\muu^\top\mr_2.$$

We bound the terms of $\mr$ one by one. Taking $\muu\muu^\top\mr_1\mw_1^\top\mw_2$ as an example, we can bound the $2\to\infty$ norm as
$$
\begin{aligned}
\|\muu\muu^\top\mr_1\mw_1^\top\mw_2\|_{2\to\infty}
&\leqslant \|\muu\muu^\top\|_\infty\cdot\|\mr_1\|_{2\to\infty}\cdot\|\mw_1^\top\mw_2\|_2\\
&\leqslant O_p(1)\cdot O_p\left(\frac{(\log n)^{2 c_1}}{n^{}\gamma^{1 / 2}}\right)\cdot 1
=O_p\left(\frac{(\log n)^{2 c_1}}{n^{}\gamma^{1 / 2}}\right).
\end{aligned}
$$
Then with the similar analysis, we have $\|\mr_1\mw_1^\top\mw_2\|_{2\to\infty}=O_p\left(\frac{(\log n)^{2 c_1}}{n^{}\gamma^{1 / 2}}\right)$, $\|\mr_2\|_{2\to\infty}=O_p\left(\frac{(\log n)^{2 c_2}}{n^{}\gamma^{1 / 2}}\right)$, $\|\muu\muu^\top\mr_2\|_{2\to\infty}=O_p\left(\frac{(\log n)^{2 c_2}}{n^{}\gamma^{1 / 2}}\right)$. Hence by setting $c=\text{max}\{c_1,c_2\}$, we have
$$
\|\mr\|_{2\to\infty}=O_p\left(\frac{(\log n)^{2 c}}{n^{}\gamma^{1 / 2}}\right).
$$

Now we bound the main part of $\mt_1+\mt_6$. We first have
$$
\begin{aligned}
	\|(\mi-\muu\muu^\top)\me \mathbf{U}|\mathbf{S}|^{-\frac{1}{2}} \mathbf{W}_2 \mathbf{I}_{p, q}\|_{2\to\infty}
	\leqslant & \|\me \mathbf{U}|\mathbf{S}|^{-\frac{1}{2}} \mathbf{W}_2 \mathbf{I}_{p, q}\|_{2\to\infty}+\|\muu\muu^\top\me\mathbf{U}|\mathbf{S}|^{-\frac{1}{2}} \mathbf{W}_2 \mathbf{I}_{p, q}\|_{2\to\infty}.
\end{aligned}
$$
According to Lemma \ref{lm:EUandUEU}, we have $\|\me\muu\|_{2\to\infty}=\sqrt{1-\rho}\cdot O_p(\sqrt{\gamma})$,$\|\muu^{\top}\me\muu\|_F=\sqrt{1-\rho}\cdot O_p(\sqrt{\gamma})$. We then have
$$
\begin{aligned}
	\|\me \mathbf{U}|\mathbf{S}|^{-\frac{1}{2}} \mathbf{W}_2 \mathbf{I}_{p, q}\|_{2\to\infty}
	&\leqslant \|\me\muu\|_{2\to\infty}\cdot \||\mathbf{S}|\|_2^{-\frac{1}{2}}\cdot \|\mathbf{W}_2 \mathbf{I}_{p, q}\|_2\\
	&\leqslant \sqrt{1-\rho}\cdot O_p(\sqrt{\gamma})\cdot O_p\left(\frac{1}{\sqrt{n\gamma}}\right)\cdot 1
	=\sqrt{1-\rho}\cdot O_p\left(\frac{1}{\sqrt{n}}\right),
\end{aligned}
$$
$$
\begin{aligned}
	\|\muu\muu^\top\me\mathbf{U}|\mathbf{S}|^{-\frac{1}{2}} \mathbf{W}_2 \mathbf{I}_{p, q}\|_{2\to\infty}
	&\leqslant \|\muu\|_{2\to\infty}\cdot \|\muu^\top\me\muu\|_F\cdot \||\mathbf{S}|\|_2^{-\frac{1}{2}}\cdot \|\mathbf{W}_2 \mathbf{I}_{p, q}\|_2\\
	&\leqslant O_p\left(\frac{1}{\sqrt{n}}\right)\cdot \sqrt{1-\rho}\cdot O_p(\sqrt{\gamma})\cdot O_p\left(\frac{1}{\sqrt{n\gamma}}\right)\cdot 1
	=\sqrt{1-\rho}\cdot O_p\left(\frac{1}{n}\right).
\end{aligned}
$$
Thus we derive the bound of the main part of $\mt_1+\mt_6$ as
$$
\|(\mi-\muu\muu^\top)\me \mathbf{U}|\mathbf{S}|^{-\frac{1}{2}} \mathbf{W}_2 \mathbf{I}_{p, q}\|_{2\to\infty}=\sqrt{1-\rho}\cdot O_p\left(\frac{1}{\sqrt{n}}\right).
$$
We therefore have
$$
\|\mt_1+\mt_6\|_{2\to\infty}=\sqrt{1-\rho}\cdot O_p\left(\frac{1}{\sqrt{n}}\right)+O_p\left(\frac{(\log n)^{2 c}}{n^{}\gamma^{1 / 2}}\right).
$$
\end{proof}

\begin{lemma}
  \label{lm:eigenPandA}
  Let $|\lambda_{1}(\ma_1)| \geq |\lambda_2(\ma_1)| \geq \dots$ be the
  eigenvalues of $\ma_1$, ordered in decreasing modulus. We then have
  $$\lambda_{k}(\ma_1)=\left\{\begin{array}{ll} {\Omega_{p}( n \gamma),O_p({n \gamma})}
                                & \text{for $k = 1,2, \dots, d$}
                                \\
                                {O_p(\sqrt{n \gamma})} & \text{for $k=d+1, \dots, n$}
\end{array},\right.$$
The same bounds are true for the eigenvalues of $\mathbf{A}_2$. And for the eigenvalues of $\mathbf{P}$, we have
$$\lambda_{k}(\mpp)=\left\{\begin{array}{ll}
{\Omega_p(n \gamma),O_p({n \gamma})} & \text{for $k = 1,2, \dots, d$} \\
{0} & \text{for $k=d+1, \dots, n$}
\end{array}.\right.$$

\end{lemma}

\begin{proof}
  Recall the assumption on the maximum expected degree of $\mpp$, i.e.,
  $$\max _{1 \leqslant i \leqslant n}
\sum_{j=1}^{n}\mpp_{i,j}\geqslant\max _{1 \leqslant i \leqslant n}
\sum_{j=1}^{n}\mpp_{i,j}(1-\mpp_{i,j})\geqslant C\ln ^{4} n.$$
We then have, from Theorem 1 in
\cite{lu2013spectra} and Weyl's inequality, that
\begin{equation}
  \label{eq:lu1}
\begin{split}
 \max_{k=1,2\dots,n} \left|\lambda_{k}(\ma_1)-\lambda_{k}(\mpp)\right|
 & \leqslant \|\ma_1-\mpp\|_{2}
 \leqslant[2+o(1)] \sqrt{\max _{1 \leqslant i \leqslant n} \sum_{j=1}^{n}\mpp_{i,j}}
 =O_p\left(\sqrt{n \gamma}\right).
\end{split}
\end{equation}
Since $\mpp$ is symmetric and rank$(\mpp)=d$, there exists a
decomposition $\mpp= \gamma \mx\mi_{p,q}\mx^\top,$ where $\mx\in\real^{n\times
  d}$ and each row of $\mx$ correspondents the latent position of each
vertex in $G_1$, and $\mi_{p, q}=\operatorname{diag}(1, \ldots, 1,-1,
\ldots,-1)$ with $p$ ones followed by $q$ minus ones on its diagonal
satisfying $p+q=d$. We now consider the eigenvalues of
$\mpp$. As $\mathbf{P}$ is rank $d$, we have $\lambda_k(\mpp)=0$ for
$k > d$. Furthermore, for $k=1,\cdots, d$
$$\begin{aligned}
\lambda_{k}(\mpp) &=\lambda_{k}\left(\gamma
  \mx\mi_{p,q}\mx^\top\right)= \gamma \lambda_{k}\left(\mx^{\top}
  \mx\mi_{p,q}\right)=n \gamma \cdot \lambda_{k}\left(\frac{1}{n} \sum_{i=1}^{n} X_{i} X_{i}^{\top}\mi_{p,q}\right),
\end{aligned}$$
where $X_i$ represents the $i$th row of $\mx$, i.e., the latent
position of $i$th vertex. Since
$\tfrac{1}{n} \sum_{i=1}^{n} X_{i} X_{i}^{\top}\mi_{p,q}$ converges to
a constant matrix, we have
$$\lambda_{k}(\mpp)=\left\{\begin{array}{ll}
{\Omega_p(n \gamma),O_p({n \gamma})} & \text{for $k = 1,2, \dots, d$} \\
{0} & \text{for $k \geqslant d + 1$}
\end{array}.\right.$$
Eq.~\eqref{eq:lu1} then implies
$$\lambda_{k}(\ma_1)=\left\{\begin{array}{ll}
{\Omega_p(n \gamma),O_p({n \gamma})} & \text{for $k=1, \dots, d$} \\
{O_p(\sqrt{n \gamma})} & \text{for $k \geqslant d+1$}
\end{array}.\right.$$
The proof for the eigenvalues of $\ma_2$ is identical.
\end{proof}

\begin{lemma}
  \label{lm:UUS-SUU}
  For the term in Eq.~\eqref{eq:V1}, we have
  $$\|\muu^\top\muu_1\ms_1-\ms\muu^\top\muu_1\|_F=O_p(1).$$
\end{lemma}

\begin{proof}

Notice from the block spectral decomposition $\ma_1=\muu_1\ms_1\muu_1^T+\muu_1'\ms_1'\muu_1'^\top$, we have $\ma_1\muu_1=\muu_1\ms_1$. Similarly, we have $\mpp\muu=\muu\ms$. Then we can split $\muu^\top\muu_1\ms_1-\ms\muu^\top\muu_1$ as
$$\begin{aligned}
	\muu^\top\muu_1\ms_1-\ms\muu^\top\muu_1
	&=\muu^\top\ma_1\muu_1-\muu^\top\mpp\muu_1\\
	&=\muu^\top(\ma_1-\mpp)\muu_1\\
	&=\muu^\top(\ma_1-\mpp)(\mi-\muu\muu^\top)\muu_1+\muu^\top(\ma_1-\mpp)\muu\muu^\top\muu_1
\end{aligned}
$$
And according to Lemma \ref{lm:EUandUEU} and Lemma \ref{lm:(I-UU)U}, we have
$
\|\muu^\top(\ma_1-\mpp)\|_2=O_p(\sqrt{n\gamma}),
\|(\mi-\muu\muu^\top)\muu_1\|_F=O_p(\frac{1}{\sqrt{n\gamma}}),
\|\muu^\top(\ma_1-\mpp)\muu\|_F=O_p(\sqrt{\gamma})
$. We therefore have
$$
\begin{aligned}
	\|\muu^\top\muu_1\ms_1-\ms\muu^\top\muu_1\|_F
&\leq\|\muu^\top(\ma_1-\mpp)\|_2\cdot \|(\mi-\muu\muu^\top)\muu_1\|_F+\|\muu^\top(\ma_1-\mpp)\muu\|_F\cdot\|\muu^\top\muu_1\|_2
\\&=O_p(\sqrt{n\gamma})\cdot O_p(\frac{1}{\sqrt{n\gamma}})+O_p(\sqrt{\gamma})\cdot 1
=O_p(1).
\end{aligned}
$$
\end{proof}

\begin{lemma}
  \label{lm:EUandUEU}
  Let $\me=\ma_1-\ma_2$, then
  $$
  \begin{aligned}
  	&\|\me \muu \|_{2\to\infty}=\sqrt{1-\rho}\cdot O_p(\sqrt{\gamma}), \quad
  \|\me \muu \|_2=\sqrt{1-\rho}\cdot O_p(\sqrt{n\gamma})\\
  &\|\muu^{\top}\me\muu\|_F=\sqrt{1-\rho}\cdot O_p(\sqrt{\gamma}).
  \end{aligned}$$
For $\ma_1-\mpp$ and $\ma_2-\mpp$, we have the same results but without $\sqrt{1-\rho}$.
\end{lemma}

\begin{proof}
For $\muu^{\top}\me\muu$, we note that the $ij$th element of $\muu^{\top}\me \muu$ is of the form
$$
\left(\muu^{\top}\me \muu\right)_{i,j}=\sum_{k<l}2\muu_{k,i}\me_{k,l}\muu_{l,j}+\sum_k\muu_{k,i}\me_{k,k}\muu_{k,j},
$$
which is a sum of {\em independent} mean $0$ random variables, and hence by
Bernstein's inequality we have
$$\left(\muu^{\top}\me \muu\right)_{i,j}=\sqrt{1-\rho}\cdot O_p(\sqrt{\gamma}).$$
Since $\left(\muu^{\top}\me \muu\right)$ is a $d\times d$ matrix where
$d$ is fixed with $n$, by the union bound we have
$$\|\muu^{\top}\me \muu\|_F=\sqrt{1-\rho}\cdot O_p(\sqrt{\gamma}).$$
For $\me\muu$, we notice
$$\begin{aligned}
 \|\me \muu \|_2
 \leqslant \sqrt{n}\cdot \|\me \muu \|_{2\to\infty}	=\sqrt{n}\cdot \max_{1\leqslant i\leqslant n} \|(\me \muu)_i \|_{2}=\sqrt{n}\cdot \max_{1\leqslant i\leqslant n} \sqrt{\sum_{j=1}^d(\me\muu)_{i,j}^2},
 \end{aligned}
$$
where $(\me \muu)_i$ represents the $i$th row of $(\me \muu)$, and by
Bernstein's inequality we have
$$(\me\muu)_{i,j}=\sqrt{1-\rho}\cdot O_p(\sqrt{\gamma}).$$
We therefore have $$\|\me \muu \|_{2\to\infty}=\sqrt{1-\rho}\cdot O_p(\sqrt{\gamma}), \quad
 \|\me \muu \|_2=\sqrt{1-\rho}\cdot O_p(\sqrt{n\gamma}).$$
For $\ma_1-\mpp$ and $\ma_2-\mpp$, the proof is similar.
\end{proof}

\begin{lemma}
	\label{lm:(I-UU)U}
	
	For terms $\muu,\muu_1,\muu_2$ in the block spectral decompositions of $\mpp,\ma_1,\ma_2$, we have
	$$
	\begin{aligned}
		&\|\left(\mi-\muu \muu^{\top}\right)\muu_{1}\|_F=O_{p}\left(\frac{1}{\sqrt{n\gamma}}\right),\quad
		\|\left(\mi-\muu \muu^{\top}\right)\muu_{2}\|_F=O_{p}\left(\frac{1}{\sqrt{n\gamma}}\right),\\
	    &\|\left(\muu_{1} \muu_{1}^{\top}-\mi\right)\muu_2\|_F=\sqrt{1-\rho}\cdot O_{p}\left(\frac{1}{\sqrt{n\gamma}}\right).
	\end{aligned}
	$$
\end{lemma}

\begin{proof}
By applying Theorem 2 in \cite{yu2014useful}, we have $$\arg\min_{\mo\in\mathbb{O}_{d}}\|\muu_{1}-\muu \mo\|_F
\leqslant \frac{2^{3/2}\cdot d ^{1/2}\cdot \|\ma_1-\mpp\|_2}{\lambda_d(\mpp)-\lambda_{(d+1)}(\mpp)},$$
According to Lemma \ref{lm:eigenPandA} and Lemma \ref{lm:E}, we have
$$
\lambda_d(\mpp)=\Omega_p(n\gamma),\lambda_{(d+1)}(\mpp)=0, \|\ma_1-\mpp\|_2=O_p(\sqrt{n\gamma}).
$$
Thus
$
\arg\min_{\mo\in\mathbb{O}_{d}}\|\muu_{1}-\muu \mo\|_F=O_{p}\left(\frac{1}{\sqrt{n\gamma}}\right).$
We therefore have
$$
\begin{aligned}
\|\left(\mi-\muu \muu^{\top}\right)\muu_{1}\|_F=\|\muu_{1}-\muu (\muu^{\top}\muu_{1})\|_F&=\arg\min_\mt \|\muu_{1}-\muu \mt\|_F\\
&\leqslant \arg\min_{\mo\in\mathbb{O}_{d}}\|\muu_{1}-\muu \mo\|_F=O_{p}\left(\frac{1}{\sqrt{n\gamma}}\right).
\end{aligned}
$$
The proof for $\left(\mi-\muu \muu^{\top}\right)\muu_{2}$ is identical.

From Lemma \ref{lm:E}, we have $\|\me\|_2=\sqrt{1-\rho}\cdot O_p(\sqrt{n\gamma})$. Then with the similar analysis, we have $\|\left(\muu_{1} \muu_{1}^{\top}-\mi\right)\muu_2\|_F=\sqrt{1-\rho}\cdot O_{p}\left(\frac{1}{\sqrt{n\gamma}}\right)$.
\end{proof}

\begin{lemma}
  \label{lm:E}
  Let $\me=\ma_1-\ma_2$, then
	$$\|\me\|_2=\sqrt{1-\rho}\cdot O_p(\sqrt{n\gamma}).$$
 And similarly, we have
	$$\|\ma_1-\mpp\|_2=O_p(\sqrt{n\gamma}), \quad\|\ma_2-\mpp\|_2=O_p(\sqrt{n\gamma}).$$
\end{lemma}
\begin{proof}
  Recall that $\me_{ij} = \ma_{1,ij} - \ma_{2,ij}$ where $\ma_{1,ij}$ and
  $\ma_{2,ij}$ are $\rho$-correlated Bernoulli random variables. We thus have
$\mathrm{Var}[\me_{i,j}]=2\mpp_{i,j}(1-\mpp_{i,j})(1-\rho)$.
Furthermore, according to our assumption, 
$$\max _{1 \leqslant i \leqslant n}
\sum_{j=1}^{n}\text{Var}[\me_{i,j}]=2(1-\rho)\max _{1 \leqslant i
  \leqslant n} \sum_{j=1}^{n}\mpp_{i,j}(1-\mpp_{i,j})\geqslant
2(1-\rho)\cdot C \log^{4}n.$$
On the other hand, because $X_i$ magnitude does not change with $n$, $\mpp_{i,j}=\gamma\cdot X_i\mi_{p,q}X_j^\top=O_p(\gamma)$. We therefore have 
$\max _{1 \leqslant i \leqslant n}
\sum_{j=1}^{n}\text{Var}[\me_{i,j}]=(1-\rho)\cdot O_p(n\gamma).$
Applying Theorem 7 in \cite{lu2013spectra} yields the stated claim. The proof for $\|\ma_1-\mpp\|_2,\|\ma_2-\mpp\|_2$ is similar.
\end{proof}

\begin{lemma}
  \label{lm:UU-W}
  With proper setting of $\mw_1,\mw_2$, we can bound $\muu^\top\muu_1-\mw_1,\muu^\top\muu_2-\mw_2$ as
	$$\|\muu^\top\muu_1-\mw_1\|_F= O_p\left(\frac{1}{n \gamma}\right),\quad
	\|\muu^\top\muu_2-\mw_2\|_F= O_p\left(\frac{1}{n \gamma}\right).$$
\end{lemma}
\begin{proof}
Without loss of generality, we let
$\mathbf{U}_1 =\left[\mathbf{U}_{1(+)} | \mathbf{U}_{1(-)}\right]$
such that the columns of $\mathbf{U}_{1(+)}$ and $\mathbf{U}_{1(-)}$
consist of orthonormal eigenvectors corresponding to the largest $p$
positive and $q$ negative non-zero eigenvalues of $\ma_1$,
respectively. Similarly, we set $\mathbf{U} =\left[\mathbf{U}_{(+)}
| \mathbf{U}_{(-)}\right]$. We therefore have
$$
\muu_{}^{\top} \muu_1=\left(\begin{array}{cc}
\muu_{(+)}^{\top} \muu_{1(+)}  & \muu_{(+)}^{\top} \muu_{1(-)} \\
\muu_{(-)}^{\top} \muu_{1(+)} &  \muu_{(-)}^{\top} \muu_{1(-)} 
\end{array}\right)\in\real^{d\times d}. $$
Let $ \muu_{(+)}^{\top}
\muu_{1(+)}=\mv_{(+)}\mathbf{\Lambda_{(+)}}\mv_{(+)1}^\top$ be the
singular value decomposition of 
$\mathbf{U}_{(+)}^{\top} {\mathbf{U}}_{1(+)} \in \mathbb{R}^{p \times
  p}$ and define $\mw_{1(+)}=\mv_{(+)} \mv_{(+)1}^{\top} \in
\mathbb{O}_{p}$. Similarly, define $\mw_{1(-)}\in\mathbb{O}_{q}$
using the singular value
decomposition of $\mathbf{U}_{(-)}^{\top} {\mathbf{U}}_{1(-)}$ and
define $\mw_1$ as the block-diagonal matrix
$$
\mw_1=\left(\begin{array}{cc}
\mw_{1(+)}  & 0 \\
0 &  \mw_{1(-)}
\end{array}\right)\in\mathbb{O}_{d}.
$$
We now analyze each block of $\muu_{}^{\top} \muu_1-\mw_1$. For the first diagonal block,
$$
\begin{aligned}
\left\|\muu_{(+)}^{\top} \muu_{1(+)}-\mw_{1(+)}\right\|_{F}^{2}
&=\operatorname{tr}\left[\left(\muu_{(+)}^{\top} \muu_{1(+)}-\mw_{1(+)}\right)\left(\muu_{(+)}^{\top} \muu_{1(+)}-\mw_{1(+)}\right)^{\top}\right] \\
&=\operatorname{tr}\left(\mathbf{\Lambda}_{(+)}^2-2\mathbf{\Lambda}_{(+)}+\mathbf{I}\right)\\
&=\sum_{i=1}^{p}(1-\sigma_{i})^2,
\end{aligned}
$$
where $\sigma_{1}, \dots, \sigma_{p}$ are the singular values of
$\muu_{(+)}^{\top} \muu_{1(+)}$ and $\mathbf{\Lambda}_{(+)}=\text{diag}\{\sigma_1, \dots,
\sigma_p\}$. Since $\muu_{(+)}$ and $\muu_{1(+)}$ both have
orthonormal columns, 
$\|\muu_{(+)}^\top\muu_{1(+)}\|_2\leqslant1$ and hence, for all $i= 1,2\dots,p$,
\begin{equation*}
\begin{aligned}
0 \leqslant 1-\sigma_{i} &=\frac{1-\sigma_{i}^2}{1 + \sigma_i}
\leqslant 1-\sigma_i^2.
\end{aligned}
\end{equation*}
We therefore have
\begin{equation}
  \label{eq:rel1}
  \sum_{i=1}^{p} (1-\sigma_{i})^2
\leq \sum_{i=1}^{p} (1 - \sigma_{i}^2)^2 \leq \Bigl(\sum_{i=1}^{p} 1
- \sigma_{i}^2 \Bigr)^2.
\end{equation}
Recalling the relationship between the $\sin$-$\Theta$ distance and
singular values (see e.g., Lemma~1 in \cite{cai2018rate}), we have
$$\begin{aligned}
\sum_{i=1}^{p} (1 - \sigma_{i}^2) \leq \inf_{\mo\in \mathbb{O}_p}\|\muu_{(+)}-\muu_{1(+)}\mo\|_F^{2}.
\end{aligned}$$
Eq.~\eqref{eq:rel1} then implies
\begin{equation*}
\|\muu_{(+)}^{\top} \muu_{1(+)}-\mw_{1(+)}\|_{F} \leq \inf_{\mo\in \mathbb{O}_p}\|\muu_{(+)}-\muu_{1(+)}\mo\|_F^{2}.
\end{equation*}
By applying Theorem 2 in \cite{yu2014useful}, we have
$$\begin{aligned}
\inf_{\mo\in \mathbb{O}_p}\|\muu_{(+)}-\muu_{1(+)}\mo\|_F &\leqslant \frac{2^{3 / 2} \cdot d^{1 / 2} \cdot\|\ma_1-\mpp\|_{2}}{\lambda_{(+)p}\left(\mathbf{P}\right)},
\end{aligned}
$$
where $\lambda_{(+)p}(\mathbf{P})$  is the $p$-th largest positive
eigenvalue of $\mathbf{P} = \mathbb{E}[\ma_1]$. According to Lemma \ref{lm:eigenPandA} and
Lemma~\ref{lm:E}, we have
$$\lambda_{(+)p}\left(\mathbf{P}\right)=\Omega_p(n \gamma), \quad \|\ma_1-\mpp\|_2=O_p(\sqrt{n \gamma}).$$
We therefore have
\begin{gather*}\inf_{\mo\in
  \mathbb{O}_p}\|\muu_{(+)}-\muu_{1(+)}\mo\|_F=
O_p\left(\frac{1}{\sqrt{n \gamma}}\right), \quad 
\left\|\muu_{(+)}^{\top}
  \muu_{1(+)}-\mw_{1(+)}\right\|_{F}=
O_p\left(\frac{1}{n \gamma}\right).
\end{gather*}
The same argument also yield
$\left\|\muu_{(-)}^{\top}
\muu_{1(-)}-\mw_{1(-)}\right\|_{F}=
O_p((n
\gamma)^{-1})$. 

We now bound
  $\muu_{(+)}^{\top} \muu_{1(-)}$. Let $u_{(+)}^i$ and $u_{1(-)}^j$
  be the $i$th column of $\muu_{(+)}$ and $j$th column of
  $\muu_{1(-)}$, respectively. The $ij$-th entry of
  $\muu_{(+)}^{\top} \muu_{1(-)}$ is $(u_{(+)}^{i})^{\top} u_{1(-)}^j$ and hence
$$(u_{(+)}^{i})^{\top} u_{1(-)}^j
=\frac{(u_{(+)}^{i})^{\top} (\mpp-\ma_1) u_{1(-)}^j}{\lambda_{(+)i}(\mpp)-\lambda_{(-)j}(\ma_1)}.$$
As the positive eigenvalues of $\mathbf{P}$ are separated from the
negative eigenvalues of $\mathbf{A}_1$, we have, by
Lemma~\ref{lm:eigenPandA},
$\left[\lambda_{(+)i}(\mpp)-\lambda_{(-)j}(\ma_1)\right]^{-1}=O_p((n
\gamma)^{-1})$. 
Furthermore, since $(u_{(+)}^{i\top})(\mpp-\ma_1)
u_{1(-)}^j$ is the $ij$-th entry of
$\muu_{(+)}^{\top}(\mpp-\ma_1)\muu_{1(-)}$, we have
$$\begin{aligned}
	\|\muu_{(+)}^{\top} \muu_{1(-)}\|_F
	=O_p\left(\frac{1}{n \gamma}\right)\cdot
    \|\muu_{(+)}^{\top}(\mpp-\ma_1)\muu_{1(-)}\|_F\leqslant
    O_p\left(\frac{1}{n \gamma}\right)\cdot \|\muu_{}^{\top}(\mpp-\ma_1)\muu_{1}\|_F.
\end{aligned}
$$
Finally, in Lemma \ref{lm:newT2T5}, we have proved
$\|\muu_{}^{\top}(\mpp-\ma_1)\muu_{1}\|_F=O_p(1)$ and hence
$$
\|\muu_{(+)}^{\top} \muu_{1(-)}\|_F
=O_p\left(\frac{1}{n \gamma}\right).
$$
An identical argument also yield 
$\|\muu_{(-)}^{\top} \muu_{1(+)}\|_F=O_p((n
\gamma)^{-1})$. 
Combining the various blocks together, we derive
$$
\|\muu^\top\muu_1-\mw_1\|_F=O_p\left(\frac{1}{n \gamma}\right).
$$

The proof for $\|\muu^\top\muu_2-\mw_2\|_F$ is identical.
\end{proof}

\newpage

\section{Additional Results of Data Analysis}
\label{sec:add res}

Additional results of the simulated and real data are provided in this
section.  In particular, Appendix~\ref{sec:d changed} and
Appendix~\ref{sec:lambda changed} contain additional results of the
simulated and real data to illustrate how the choice of the embedding
dimension $d$ and the penalty parameter $\lambda$ affects the
performance of our algorithm, respectively; see
Figure~\ref{rho-RDPG_d}, Figure~\ref{rho-SBM_d},
Figure~\ref{highschool_shared_Algo1_d} and Table~\ref{tab:Bing1_d} for
the embedding dimension $d$, and see Table~\ref{tab:RDPG_lambda},
Table~\ref{tab:SBM_lambda},
Figure~\ref{highschool_shared_Algo1_lambda} and
Table~\ref{tab:Bing1_lambda} for the penalty parameter $\lambda$.
Figure~\ref{rho-SBM_gamma} in Appendix~\ref{sec:gamma changed} shows
the robustness of our algorithm with different values of the sparsity
parameter $\gamma$ using $\rho$-$\mathrm{SBM}$ setting as an example.
Figure~\ref{highschool_shared_Algo1_rerank} and
Table~\ref{tab:Bing1_reranking} in Appendix~\ref{sec:real data
reranking} illustrate how the reranking step improves the performance
of our algorithm for the real data.  The comparisons between our
algorithm and the embedding followed by Gaussian mixture modeling
algorithm of \cite{agterberg2020vertex} are provided in
Appendix~\ref{sec:compare}. 
Finally, Figure~\ref{rho-RDPG_n=1000}
and Figure~\ref{rho-SBM_n=1000} in Appendix~\ref{sec:n=1000} show the
simulation results when we increase the number of vertices from $n =
300$ (as used in Section~\ref{sec:simulation} of the main paper) to $n
= 1000$. In particular the performance of our algorithm improves for $n = 1000$
in the $\rho$-RDPG setting and is stable in the $\rho$-SBM
setting. This is as expected because the latent positions for a
$\rho$-SBM are sampled from a mixture of point masses and hence as $n$
increases the number of points from the same block/classes also
increases. 

\subsection{Results with $d$ changed}
\label{sec:d changed}


\begin{figure}[h!]
\centering
\subfigure{\includegraphics[width=0.35\textwidth]{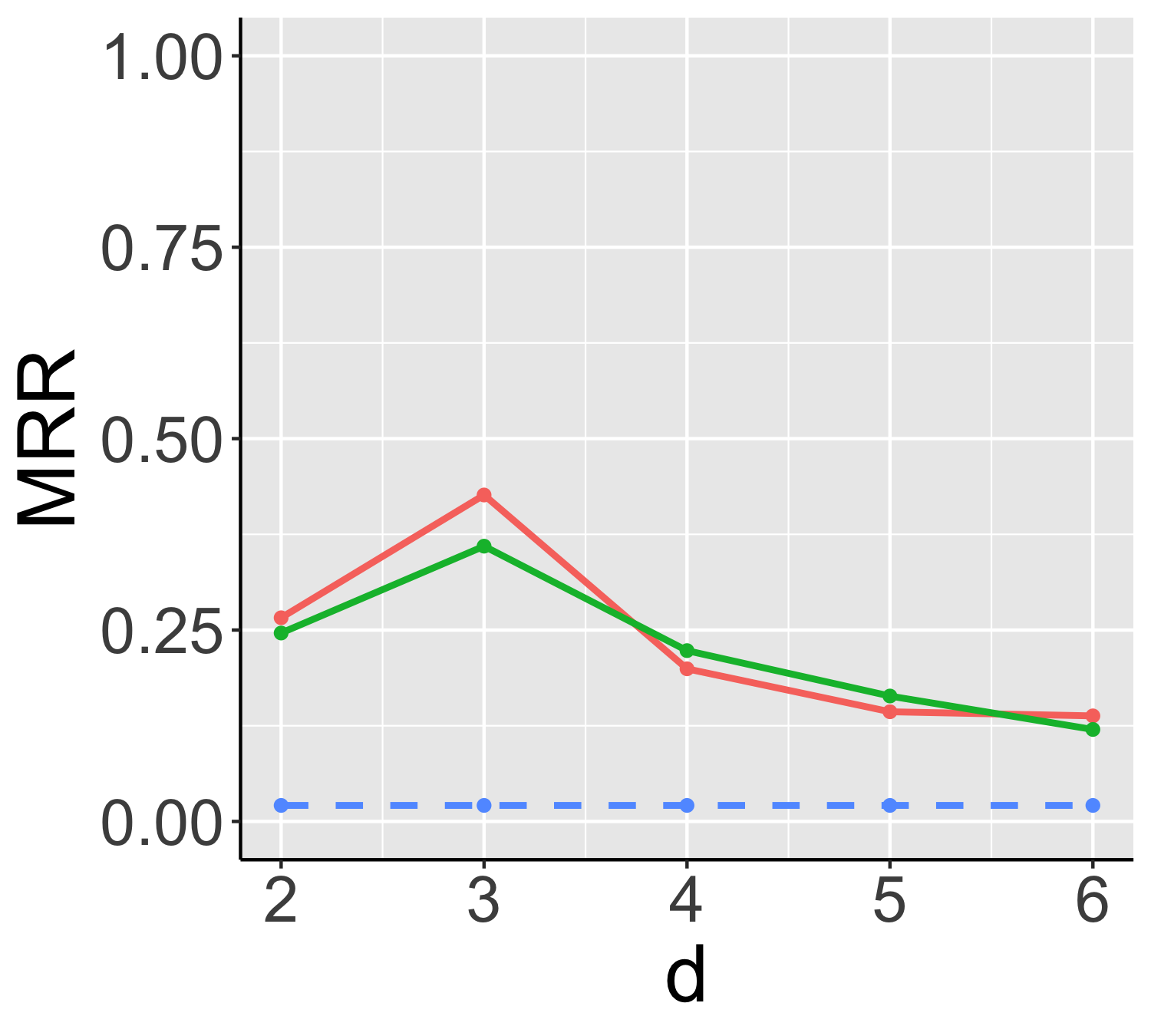}}
\subfigure{\includegraphics[width=0.35\textwidth]{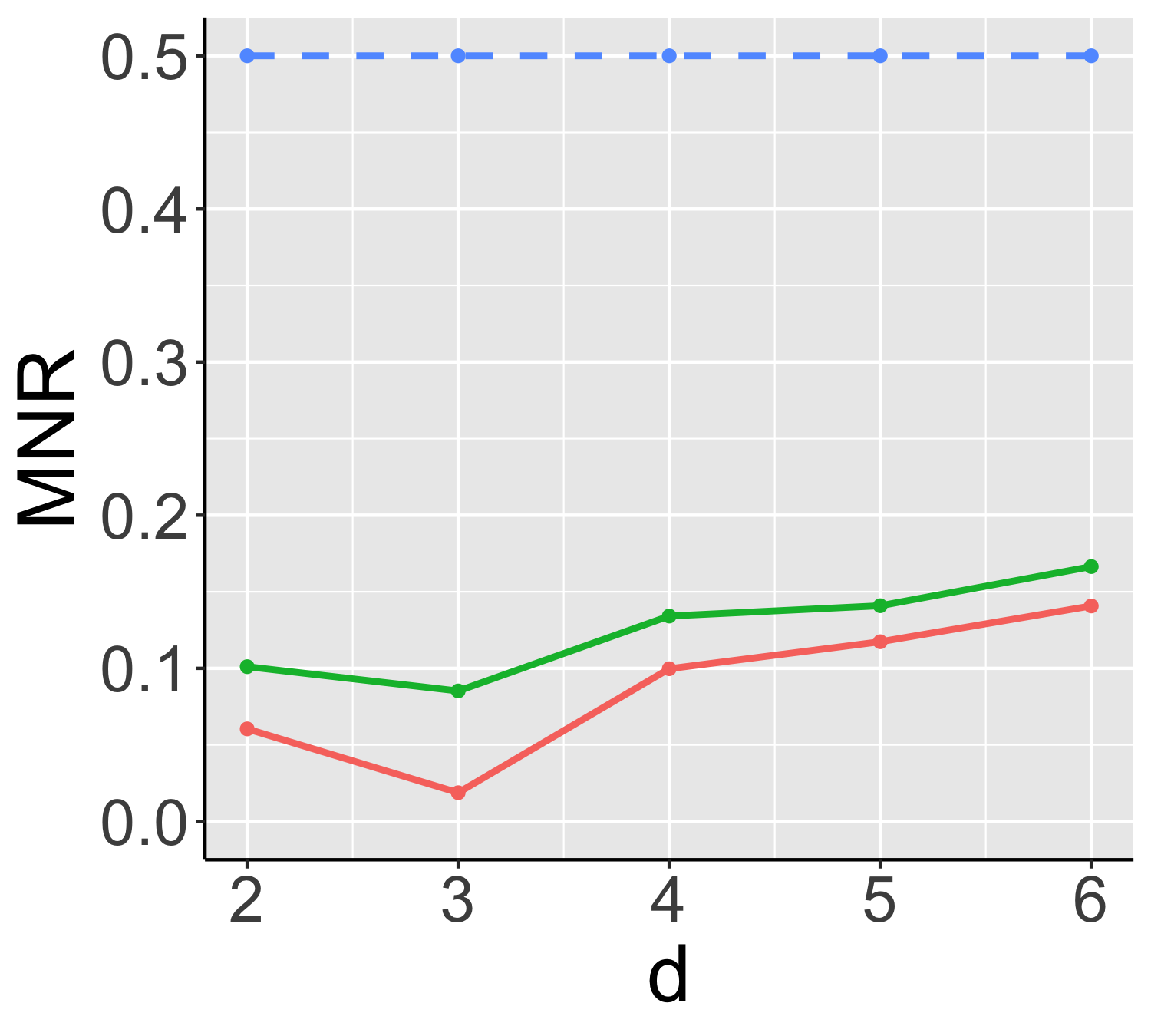}}
\caption{Performance of our algorithm for pairs of
  $\rho$-$\mathrm{RDPG}$ graphs on $n = 300$ vertices. The mean
  reciprocal rank (MRR) and mean normalized rank (MNR) are computed
  based on $500$ Monte Carlo replicates. The MRR and MNR are plotted
  for different values of the embedding parameter $d$. }
\label{rho-RDPG_d}
\end{figure}


\begin{figure}[h!]
\centering
\subfigure{\includegraphics[width=0.35\textwidth]{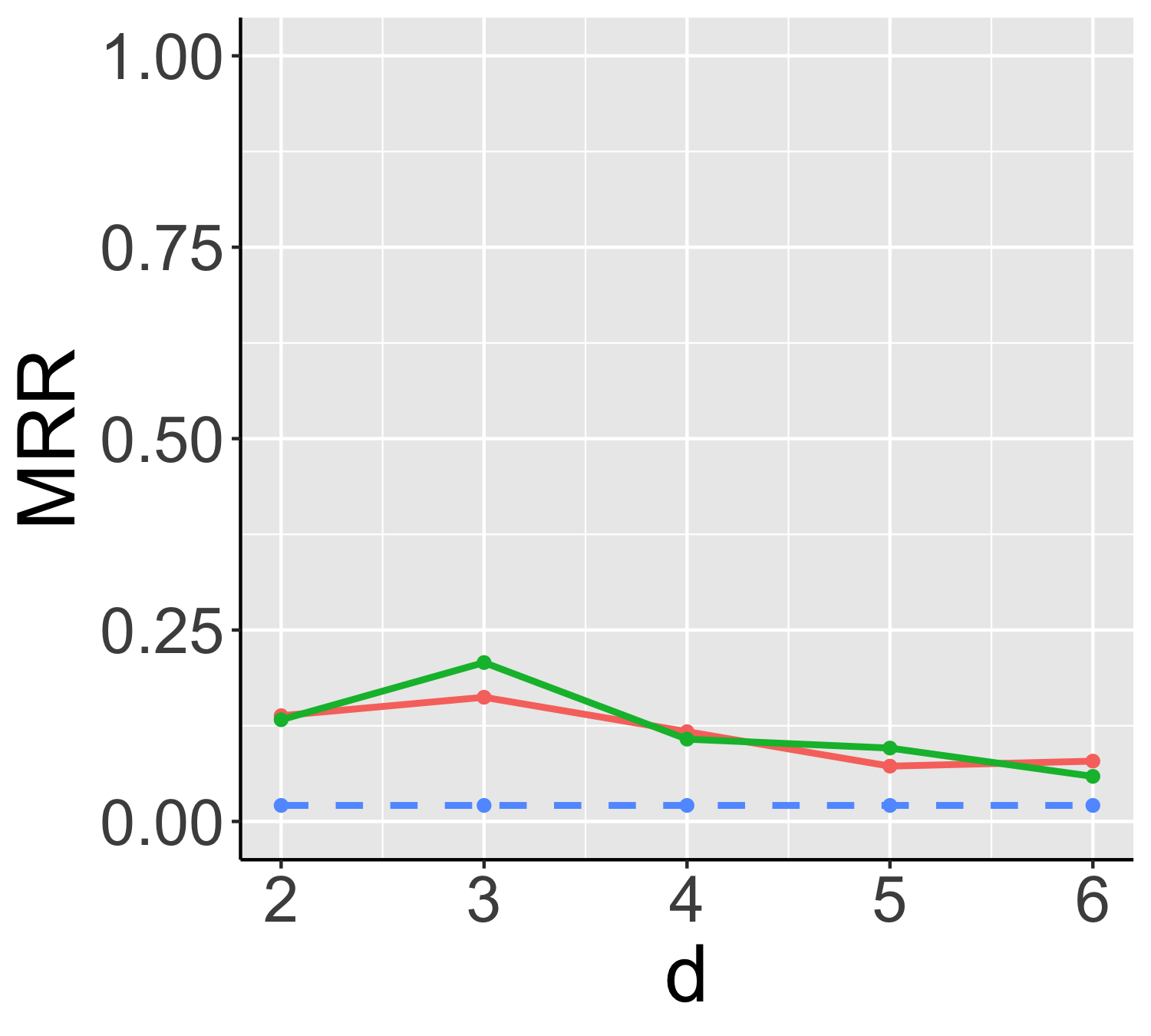}}
\subfigure{\includegraphics[width=0.35\textwidth]{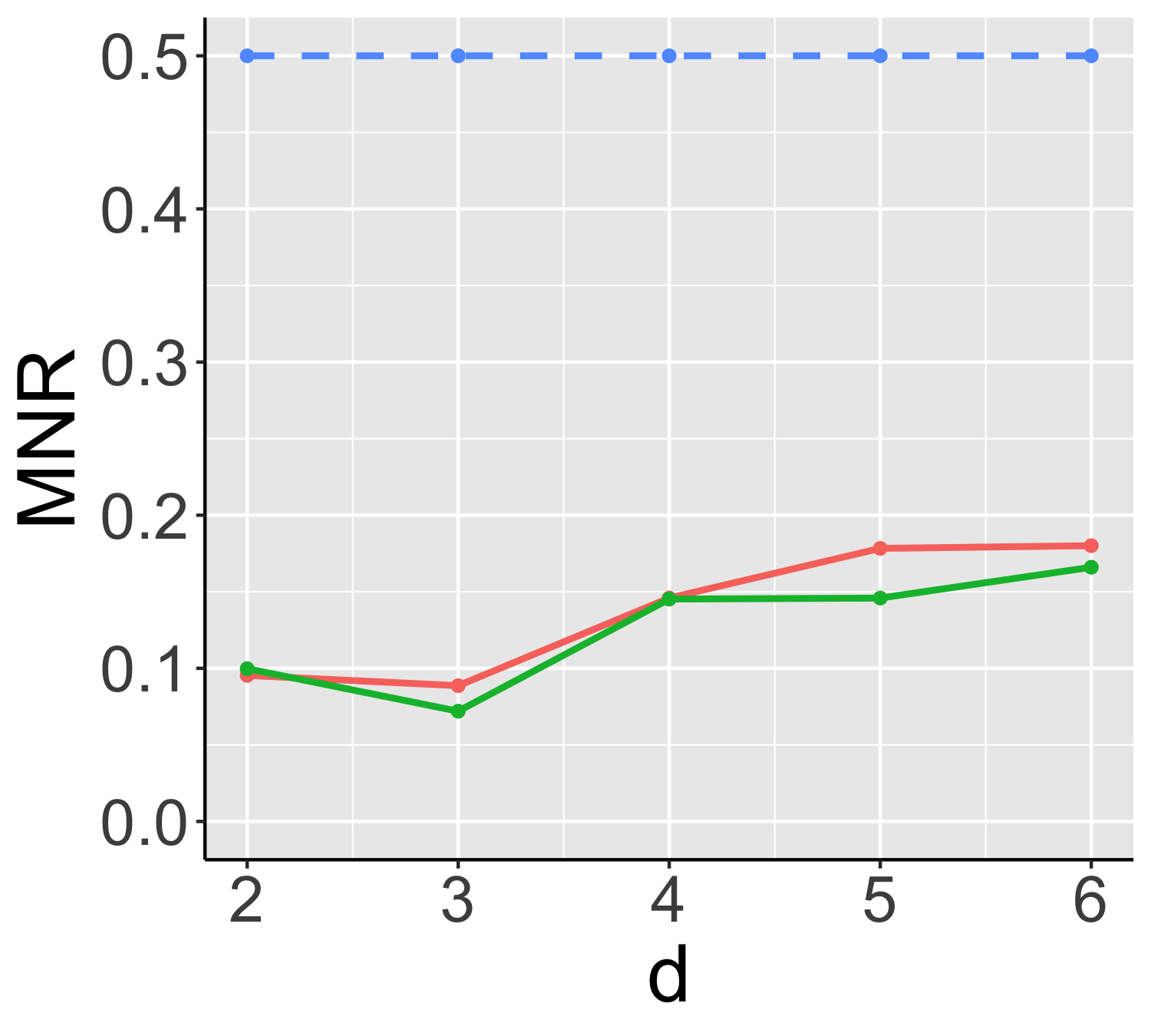}}
\caption{Performance of our algorithm for pairs of
  $\rho$-$\mathrm{SBM}$ graphs on $n = 300$ vertices. The mean
  reciprocal rank (MRR) and mean normalized rank (MNR) are computed
  based on $500$ Monte Carlo replicates. The MRR and MNR are plotted
  for different values of the embedding parameter $d$.}
\label{rho-SBM_d}
\end{figure}


\begin{figure}[ht!]
\centering
\subfigure[$d=2$]{\includegraphics[width=0.49\textwidth]{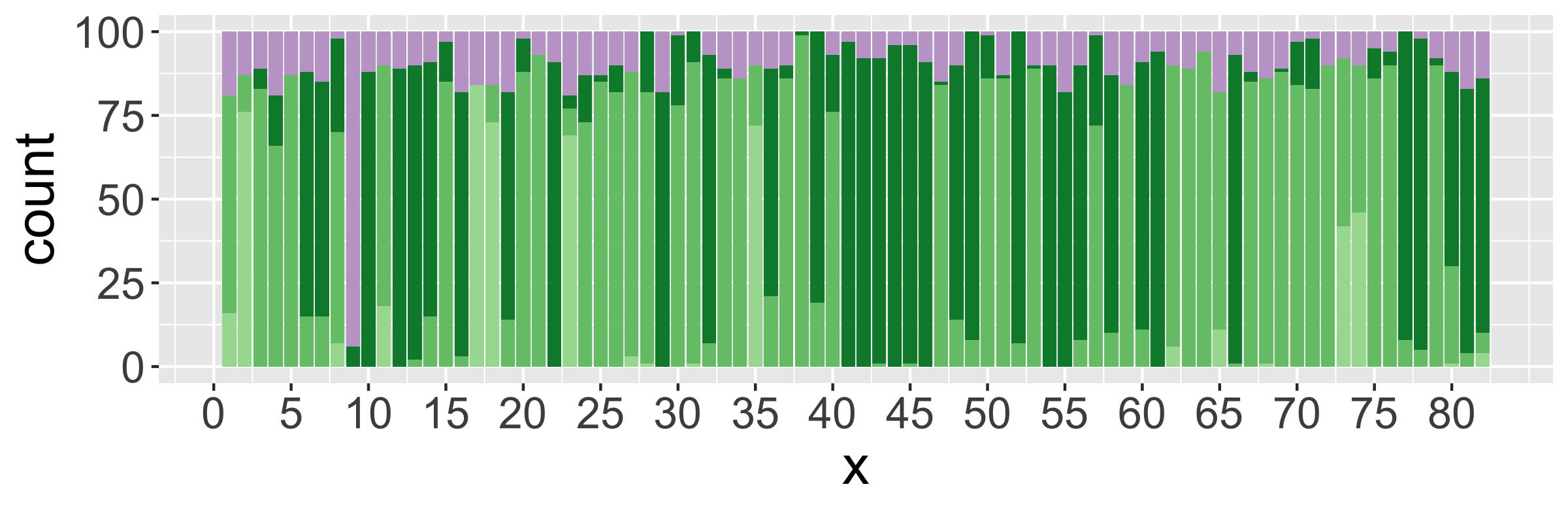}}
\subfigure[$d=3$]{\includegraphics[width=0.49\textwidth]{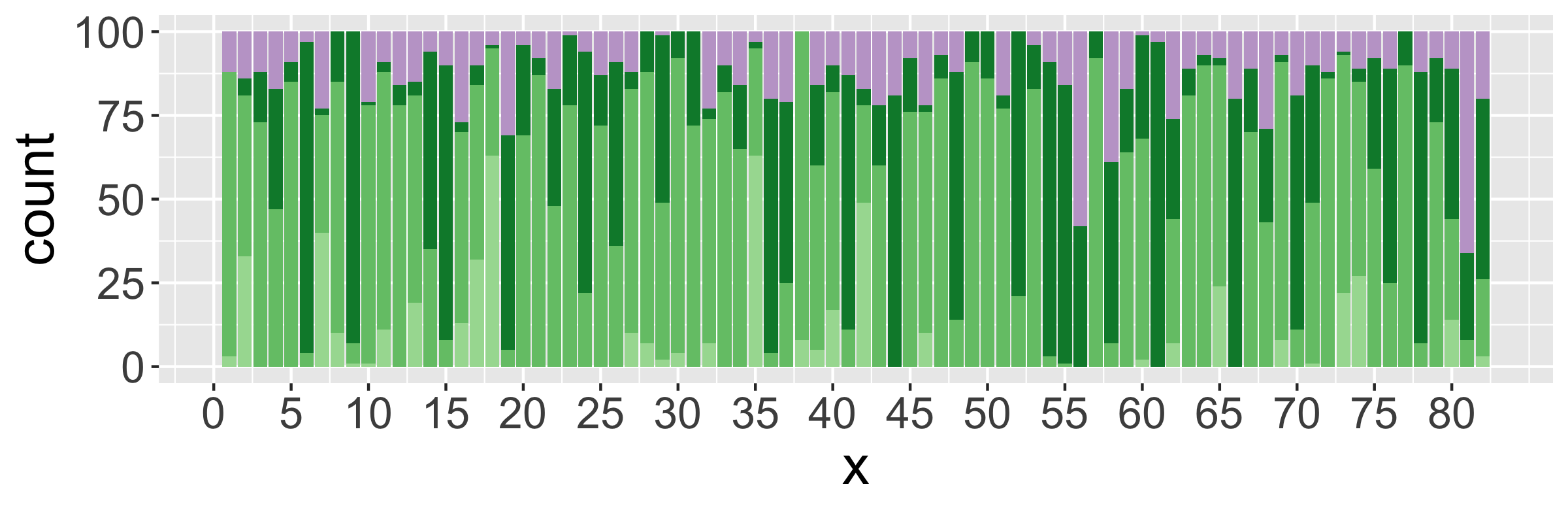}}

\subfigure[$d=4$]{\includegraphics[width=0.49\textwidth]{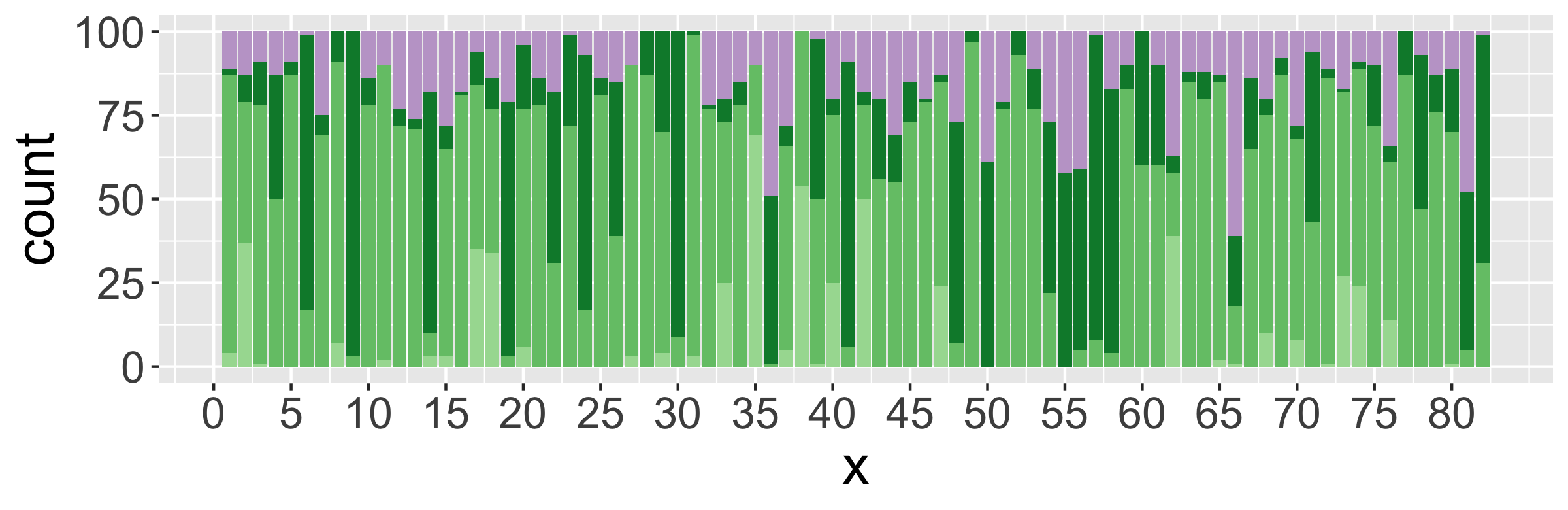}}
\subfigure[$d=5$]{\includegraphics[width=0.49\textwidth]{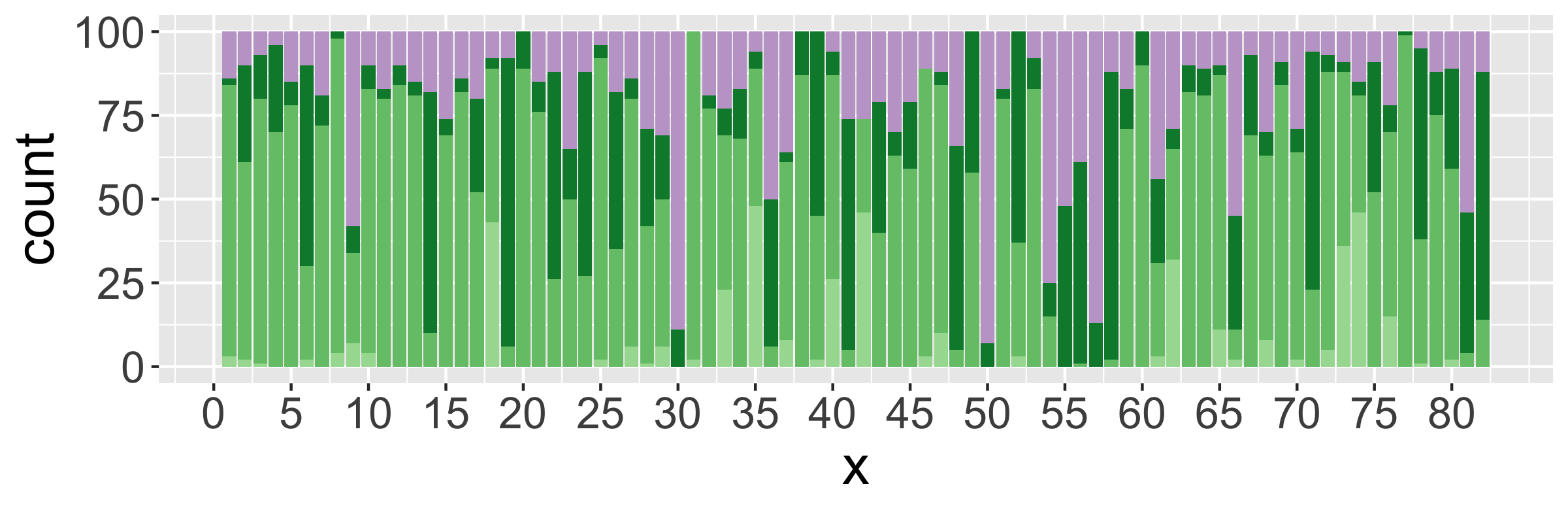}}
\caption{Performance of our algorithm with $d\in\{2,3,4,5\}$ for vertex nomination between the two high-school networks. Here we consider only the subgraphs induced by the $82$ shared vertices. The graphs embeddings are
aligned via orthogonal Procrustes transformation.}
\label{highschool_shared_Algo1_d}
\end{figure}


\begin{table}[h!]
\centering
\begin{tabular}{c|cccccccc}
 & 1\% & 5\%& 10\%&25\%& 50\%& 75\%&95\%& 99\% \\
 \hline
Procrustes ($d=2$) & 0.003 & 0.013 & 0.030 & 0.074 & 0.196 & 0.387 & 0.750 & 0.870   \\
Procrustes ($d=3$) & 0.002 & 0.009 & 0.021 & 0.071 & 0.183 & 0.356 & 0.792 & 0.953\\
Procrustes ($d=4$) & 0.002 & 0.010 & 0.021 & 0.063 & 0.173 & 0.361 & 0.771 & 0.966 \\
Procrustes ($d=5$) & 0.002 & 0.009 & 0.021 & 0.058 & 0.164 & 0.353 & 0.763 & 0.950 
\end{tabular}
\caption{Quantile levels of normalized rank (NR) values for vertex nomination with the Bing entity networks on $n = 1000$ vertices}
\label{tab:Bing1_d}
\end{table}

\clearpage

\subsection{Results with $\lambda$ changed}

\label{sec:lambda changed}

\begin{table}[h!]
\centering
\begin{tabular}{c|ccccc}
 & $0.01$ & $0.1$& $1$&$10$ &$100$\\
 \hline
Procrustes &  0.494 & 0.444 & 0.428 & 0.432  &0.447\\
set registration & 0.427 & 0.372 & 0.367 & 0.370& 0.345
\end{tabular}
\caption{The mean
  reciprocal rank (MRR) with different setting of the penalty parameter $\lambda\in\{0.01,0.1,1,10,100\}$ for pairs of $\rho$-$\mathrm{RDPG}$ graphs on $n = 300$ vertices}
\label{tab:RDPG_lambda}
\end{table}

\begin{table}[h!]
\centering
\begin{tabular}{c|ccccc}
 & $0.01$ & $0.1$& $1$&$10$ &$100$\\
 \hline
Procrustes &  0.206 & 0.182 & 0.172 & 0.154  &0.159\\
set registration & 0.209 & 0.208 & 0.184 & 0.205& 0.168
\end{tabular}
\caption{The mean
  reciprocal rank (MRR) with different setting of the penalty parameter $\lambda\in\{0.01,0.1,1,10,100\}$ for pairs of $\rho$-$\mathrm{SBM}$ graphs on $n = 300$ vertices}
\label{tab:SBM_lambda}
\end{table}

\begin{figure}[ht!]
\centering
\subfigure[$\lambda=0.1$]{\includegraphics[width=0.49\textwidth]{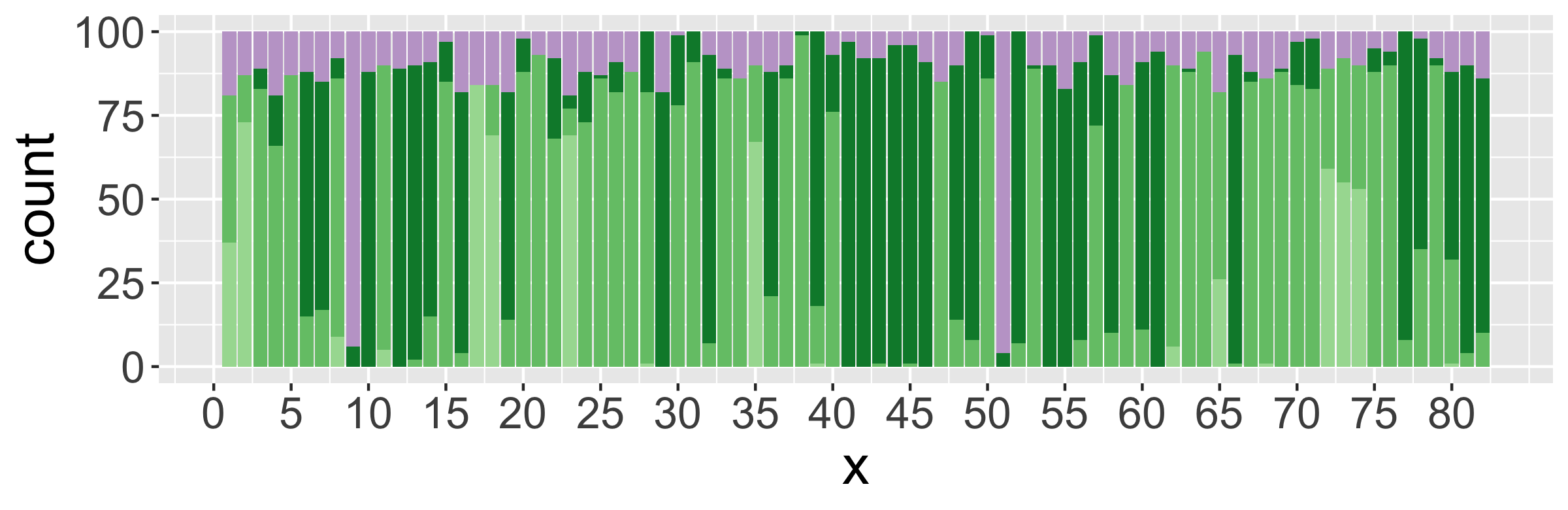}}
\subfigure[$\lambda=1$]{\includegraphics[width=0.49\textwidth]{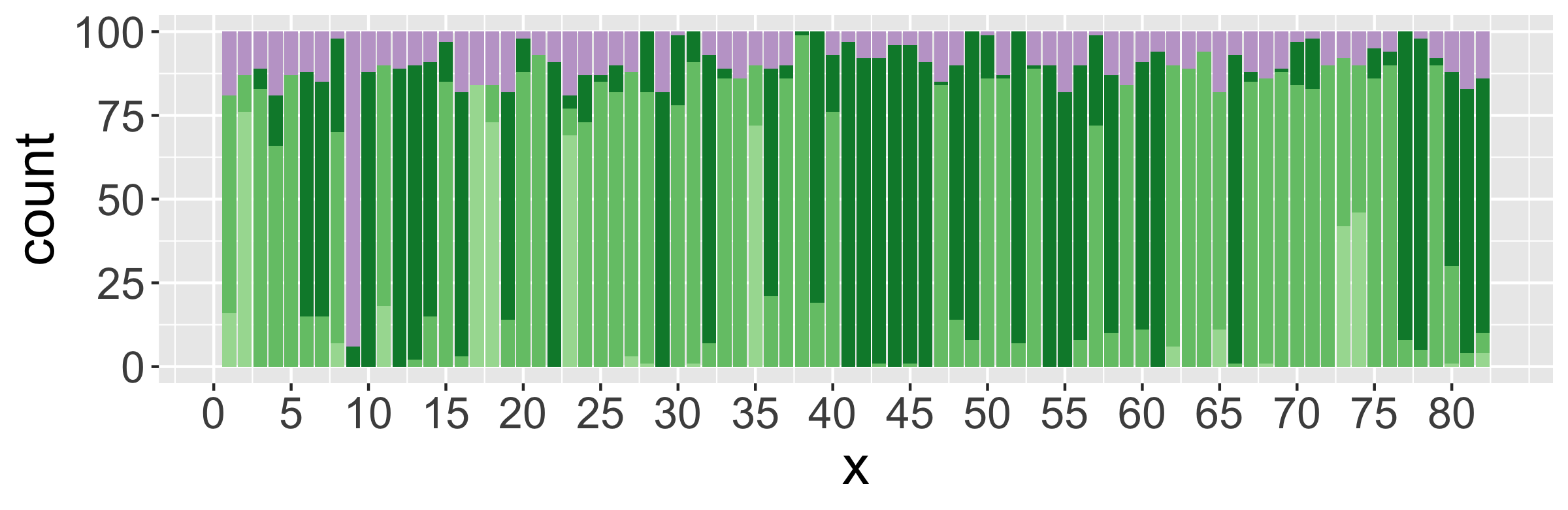}}

\subfigure[$\lambda=10$]{\includegraphics[width=0.49\textwidth]{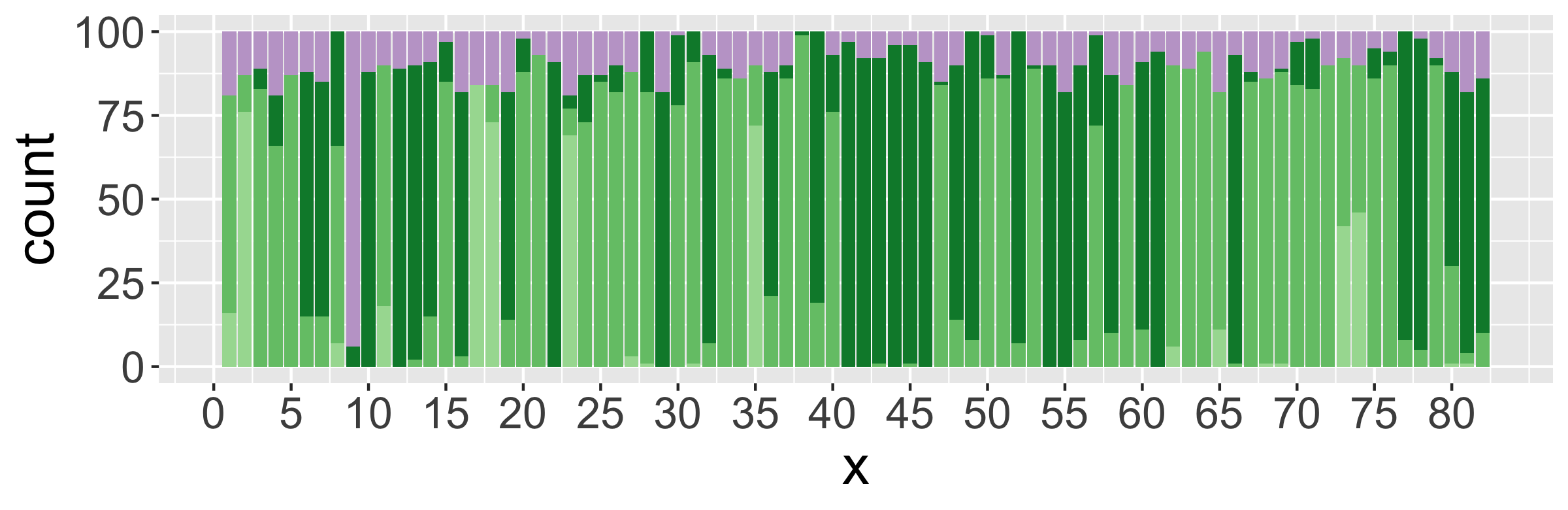}}
\subfigure[$\lambda=100$]{\includegraphics[width=0.49\textwidth]{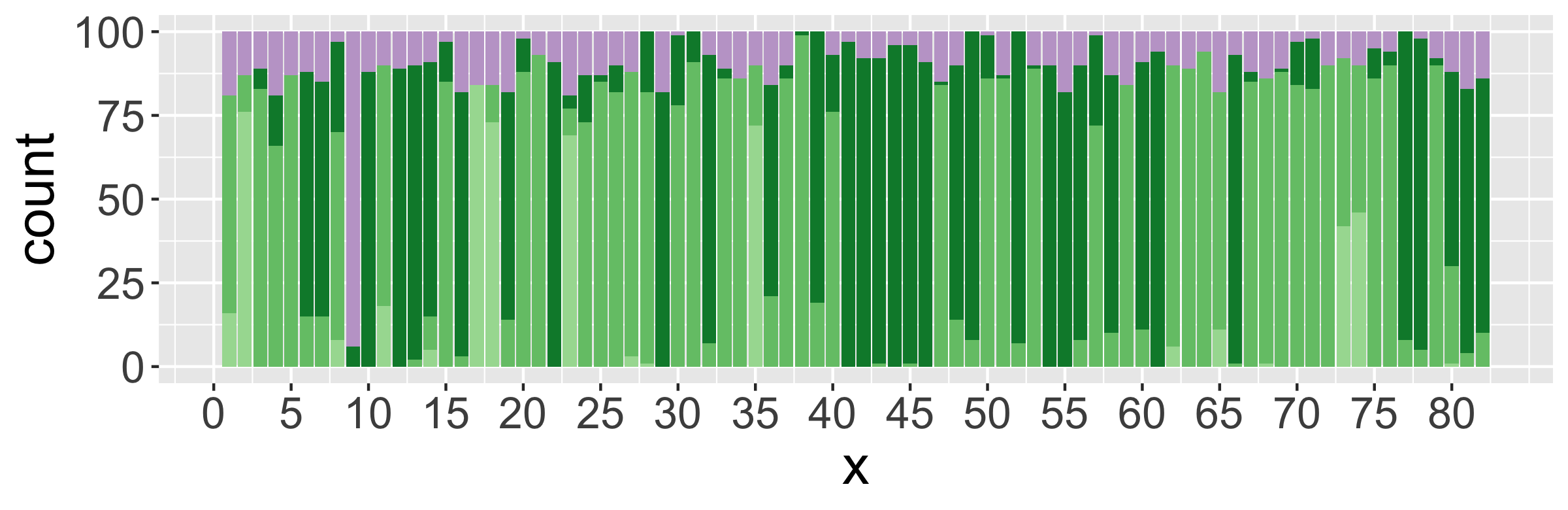}}
\caption{Performance of our algorithm with the penalty parameter $\lambda\in\{0.1,1,10,100\}$ for vertex nomination between the two high-school networks. Here we consider only the subgraphs induced by the $82$ shared vertices. The graphs embeddings are aligned via orthogonal Procrustes transformation.}
\label{highschool_shared_Algo1_lambda}
\end{figure}

\begin{table}[h!]
\centering
\begin{tabular}{c|cccccccc}
 & 1\% & 5\%& 10\%&25\%& 50\%& 75\%&95\%& 99\% \\
 \hline
Procrustes ($\lambda=0.01$) & 0.003 &0.012 &0.023 &0.063 &0.180 &0.378 &0.829 & 0.904   \\
Procrustes ($\lambda=0.1$) & 0.003 &0.013 &0.030 &0.074 &0.196 &0.387 &0.750 &0.870\\
Procrustes ($\lambda=1$) & 0.003 &0.013 &0.030 &0.074 &0.196 &0.387 &0.750 &0.870 \\
Procrustes ($\lambda=10$) & 0.003 &0.013 &0.030 &0.074 &0.196 &0.387 &0.750 &0.870 \\
Procrustes ($\lambda=100$) & 0.003 &0.013 &0.030 &0.074 &0.196 &0.387 &0.750 &0.870
\end{tabular}
\caption{Quantile levels of normalized rank (NR) values for vertex nomination with the Bing entity networks on $n = 1000$ vertices}
\label{tab:Bing1_lambda}
\end{table}

\subsection{Results with $\gamma$ changed}

\label{sec:gamma changed}

\begin{figure}[h!]
\centering
\subfigure{\includegraphics[width=0.35\textwidth]{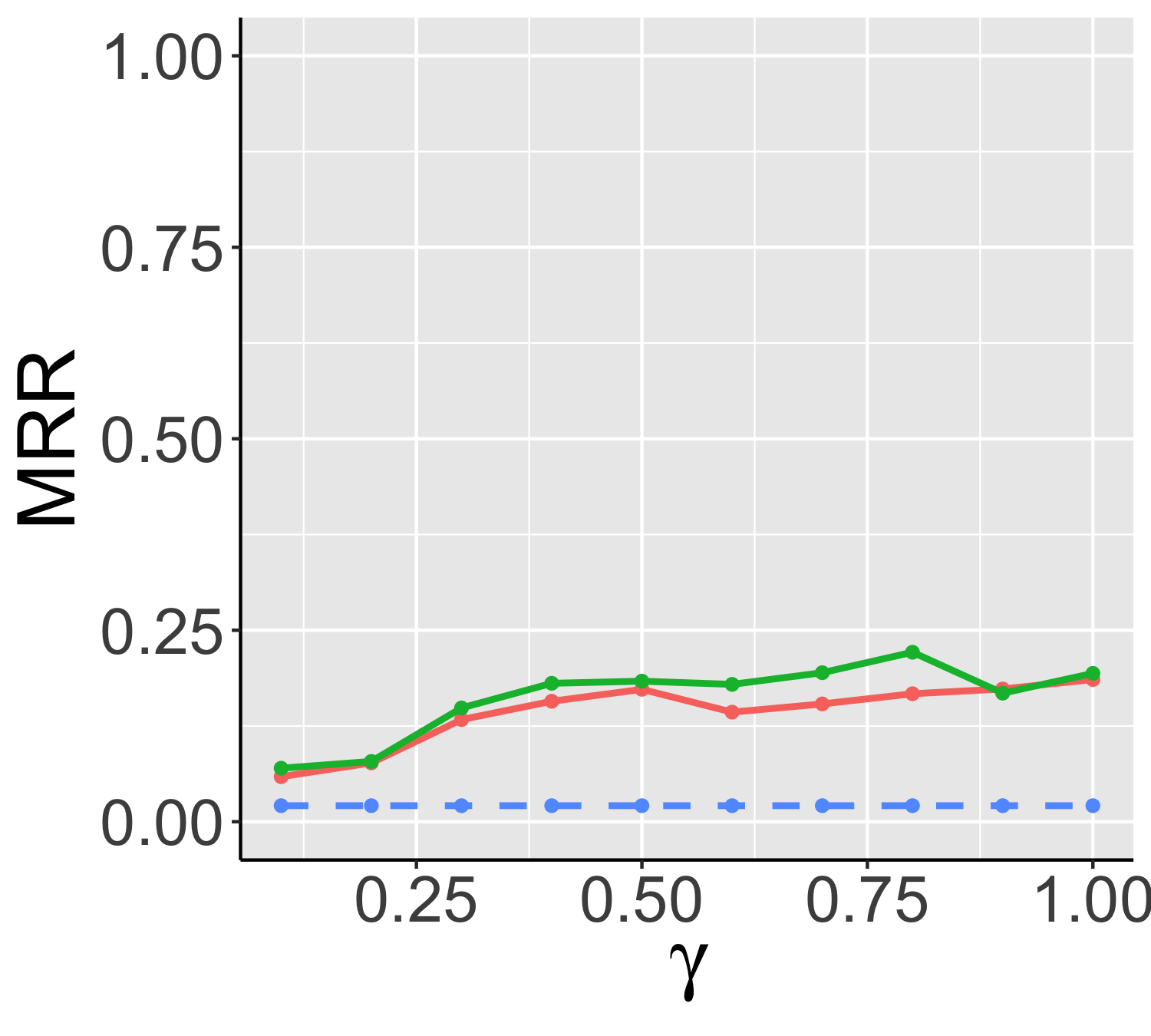}}
\subfigure{\includegraphics[width=0.35\textwidth]{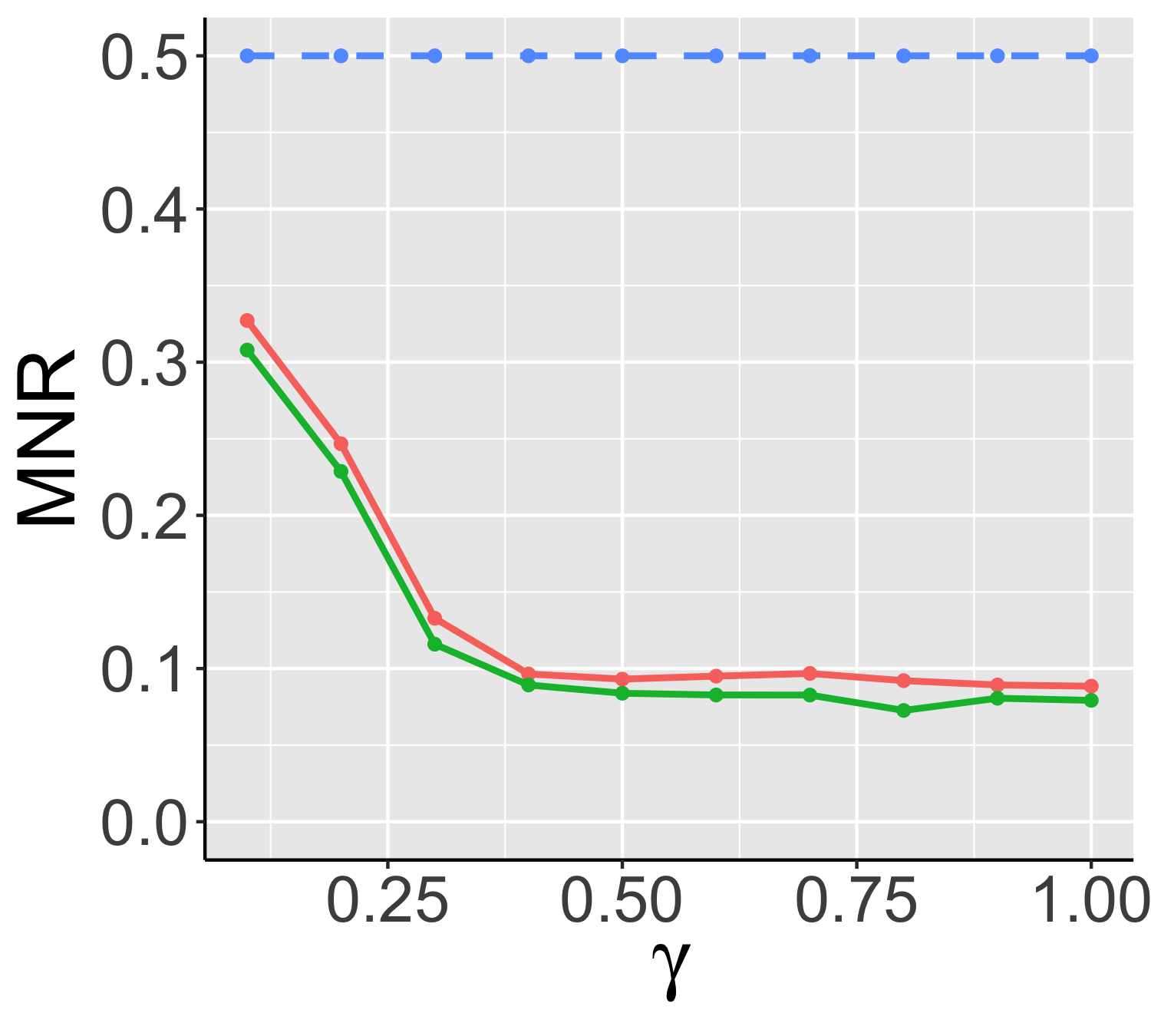}}
\caption{Performance of our algorithm for pairs of
  $\rho$-$\mathrm{SBM}$ graphs on $n = 300$ vertices. The mean
  reciprocal rank (MRR) and mean normalized rank (MNR) are computed
  based on $500$ Monte Carlo replicates. The MRR and MNR are plotted
  for different values of the sparsity parameter $\gamma \in \{0.1,0.2,0.3,\dots,1\}.$ We note that different values of $\gamma \geq 0.3$ yield almost identical accuracy.}
\label{rho-SBM_gamma}
\end{figure}

\clearpage

\subsection{Results of reranking step}

\label{sec:real data reranking}

\begin{table}[h!]
\centering
\begin{tabular}{c|cccc}
 & 1\% & 5\%& 10\%&25\% \\
 \hline
Procrustes ($100$ seeds) &  0.002 & 0.011 & 0.026 & 0.073  \\
Procrustes with reranking ($100$ seeds) &  0.001 & 0.002 & 0.004 & 0.027\\
set registration ($100$ seeds)& 0.002 & 0.012 & 0.027 & 0.073\\
set registration with reranking ($100$ seeds) & 0.001 & 0.002 & 0.004 & 0.027
\end{tabular}
\caption{Quantile levels of normalized rank (NR) values for vertex nomination with and without the reranking step with the Bing entity networks on $n = 1000$ vertices}
\label{tab:Bing1_reranking}
\end{table}

\begin{figure}[ht!]
\centering
\subfigure[without the reranking step]{\includegraphics[width=0.49\textwidth]{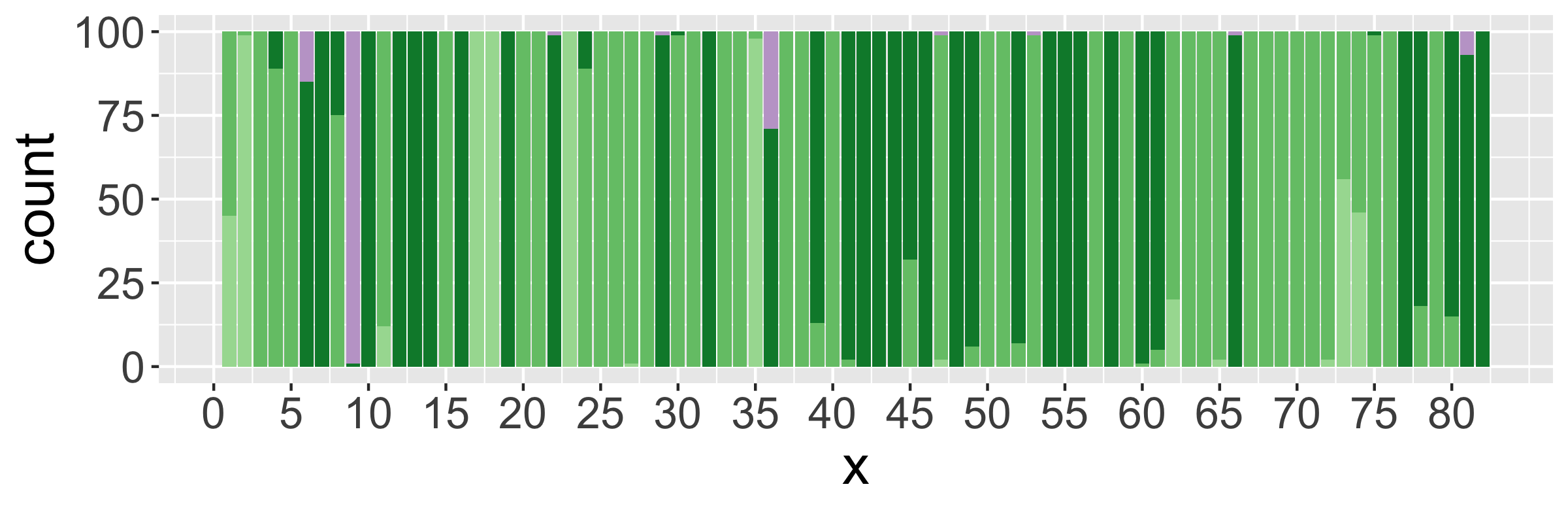}}
\subfigure[with the reranking step]{\includegraphics[width=0.49\textwidth]{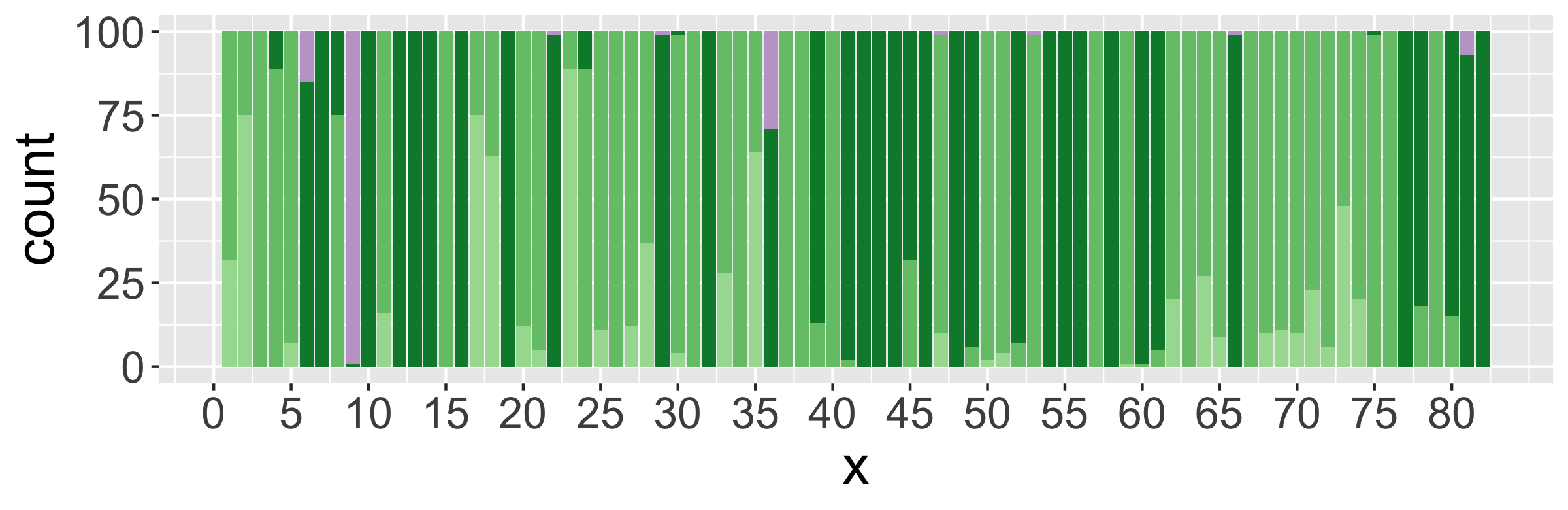}}
\caption{ Performance of our algorithm with and without the reranking
step for vertex nomination between the two high-school networks. Here
we consider only the subgraphs induced by the $82$ shared
vertices. The graphs embeddings are aligned via orthogonal Procrustes
transformation using $10$ randomly selected seeds. If a
reranking step was done then it was also done using these same seed
vertices. The reranking step leads to improved accuracy. In
particular, if there was no reranking step then
the number of vertices of interest in the first graph for which the
corresponding vertex in the second graph has a certain probability of being at the top of the nomination list is $14$. This number increases to $29$ with the use of the reranking step.  
}
\label{highschool_shared_Algo1_rerank}
\end{figure}

\newpage
\subsection{Comparison with other methods}

\label{sec:compare}

\begin{figure}[h!]
\centering
\subfigure{\includegraphics[width=0.35\textwidth]{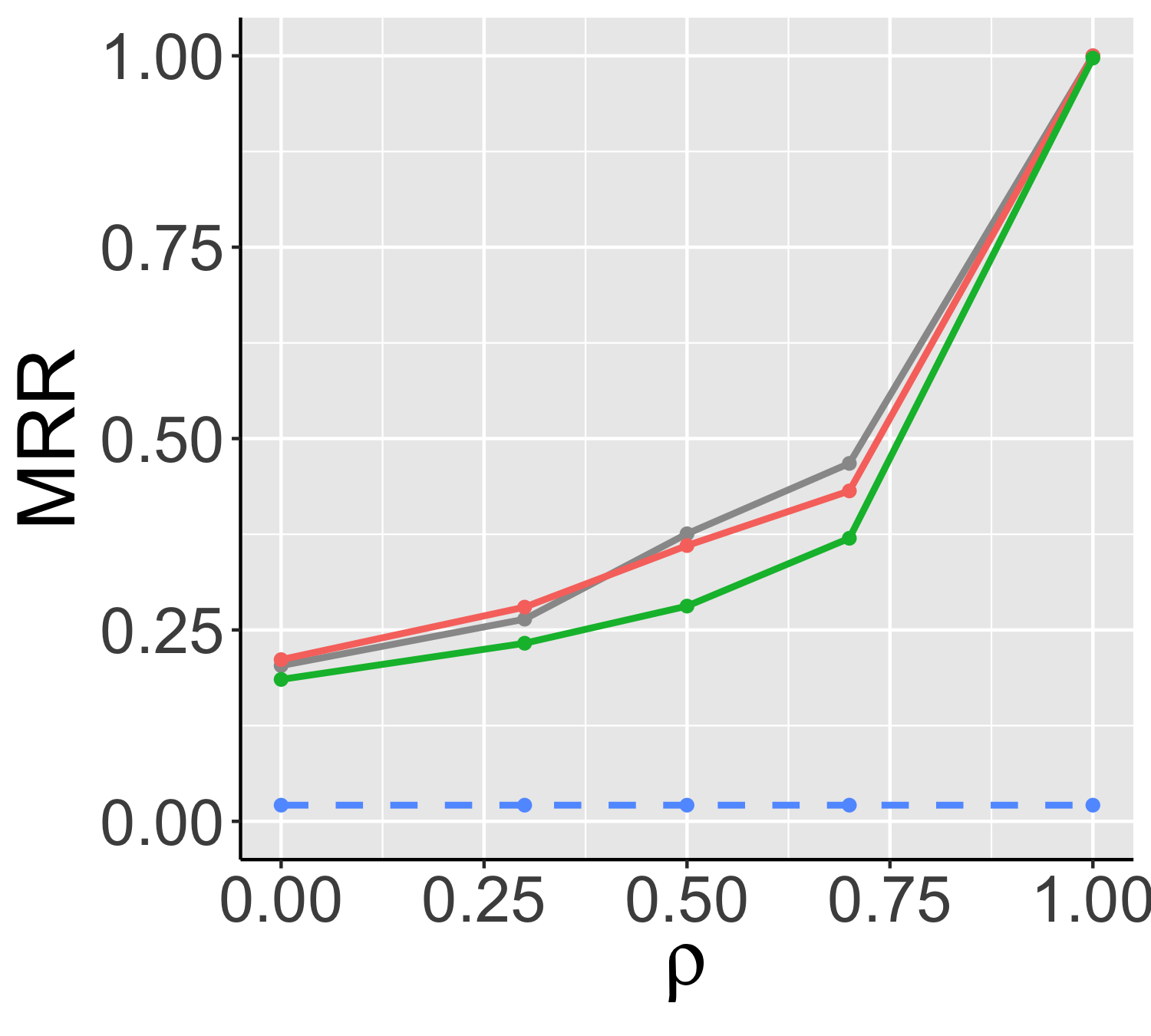}}
\subfigure{\includegraphics[width=0.35\textwidth]{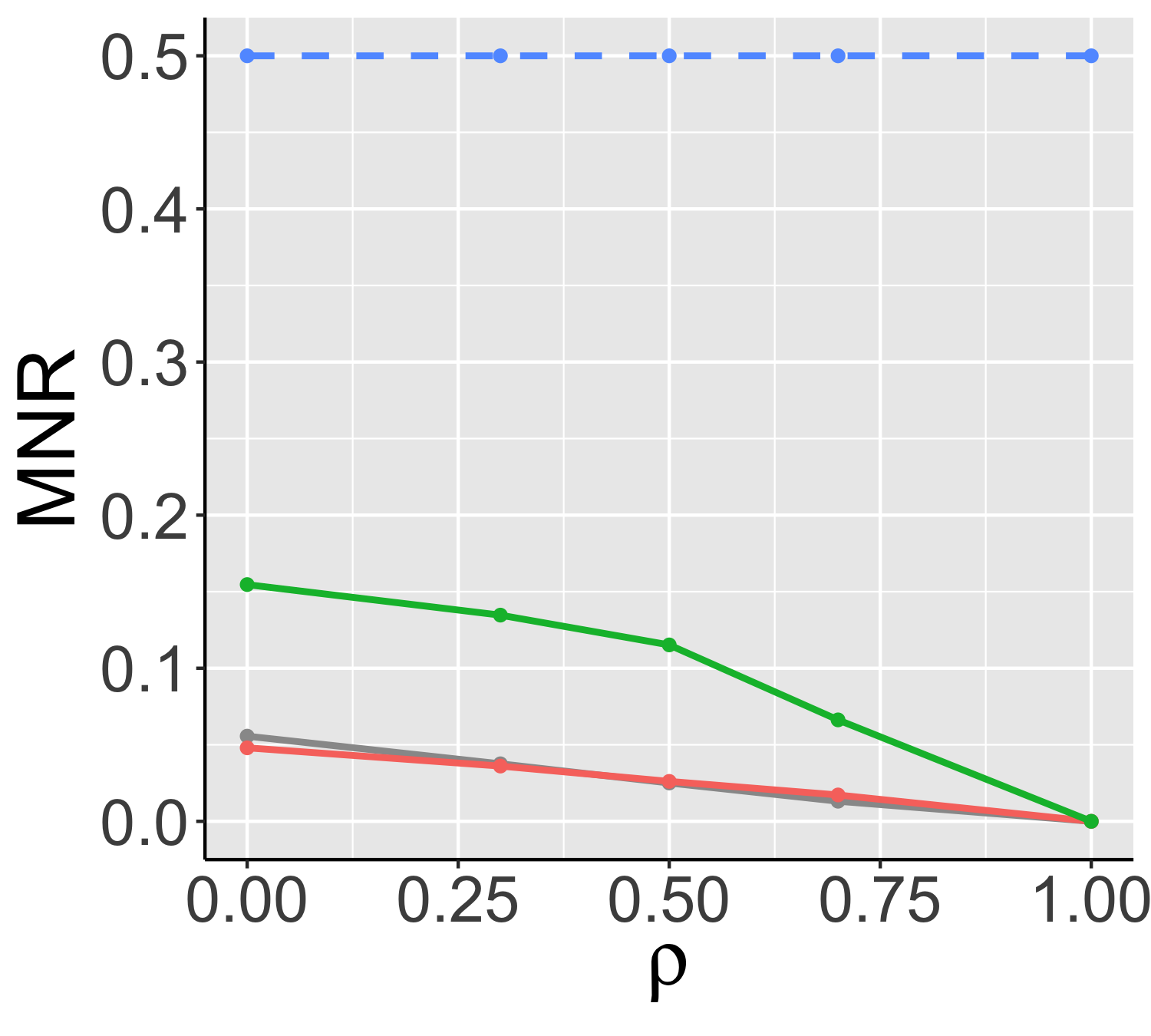}}
\caption{Performance of our algorithm for pairs of
  $\rho$-$\mathrm{RDPG}$ graphs on $n = 300$ vertices. The mean
  reciprocal rank (MRR) and mean normalized rank (MNR) are computed
  based on $500$ Monte Carlo replicates. The MRR and MNR are plotted
  for different values of the sparsity parameter $\rho$.
  The grey line corresponds to the method in \cite{agterberg2020vertex}.}
\label{rho-RDPG_compare}
\end{figure}

\begin{figure}[h!]
\centering
\subfigure{\includegraphics[width=0.35\textwidth]{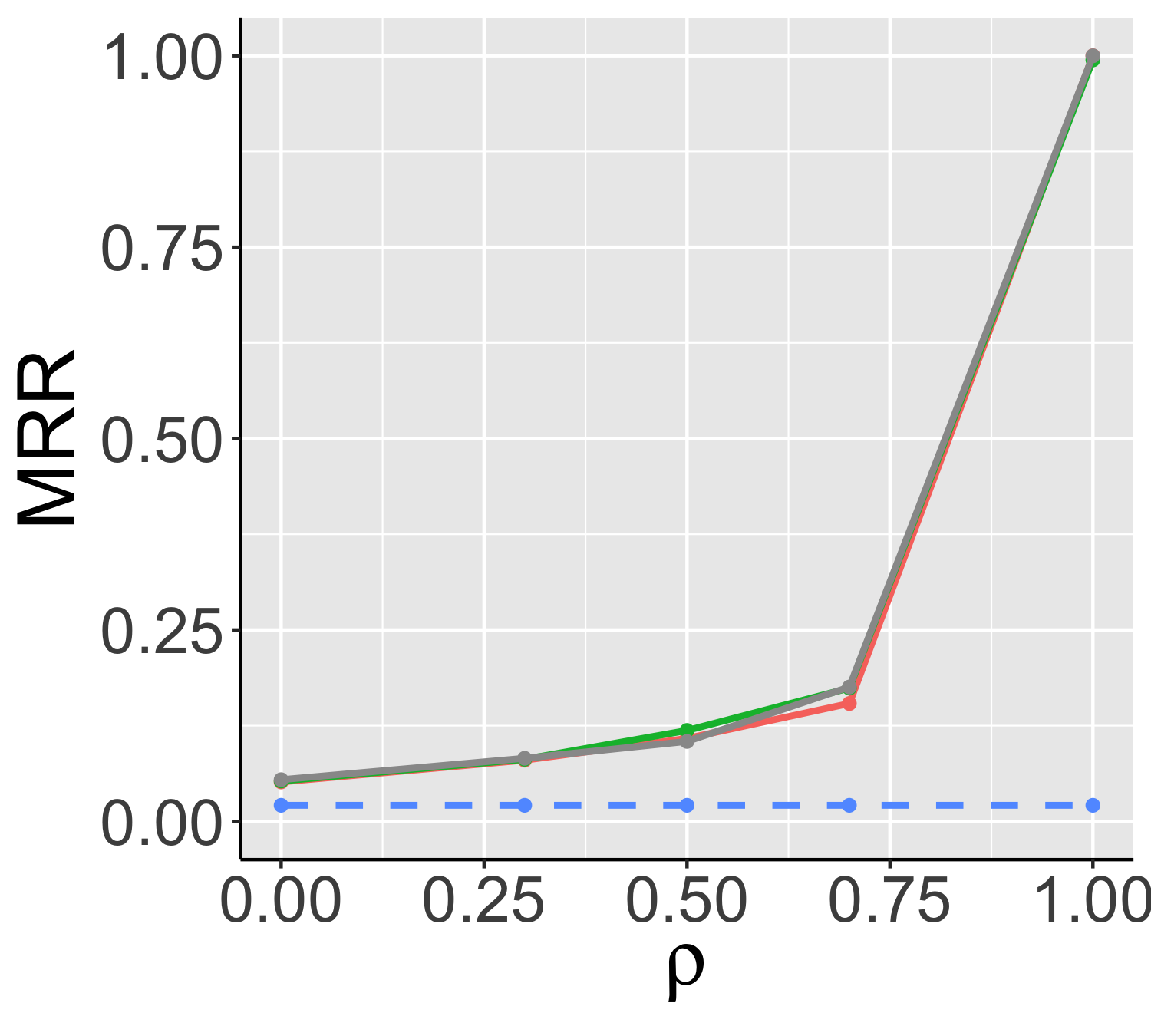}}
\subfigure{\includegraphics[width=0.35\textwidth]{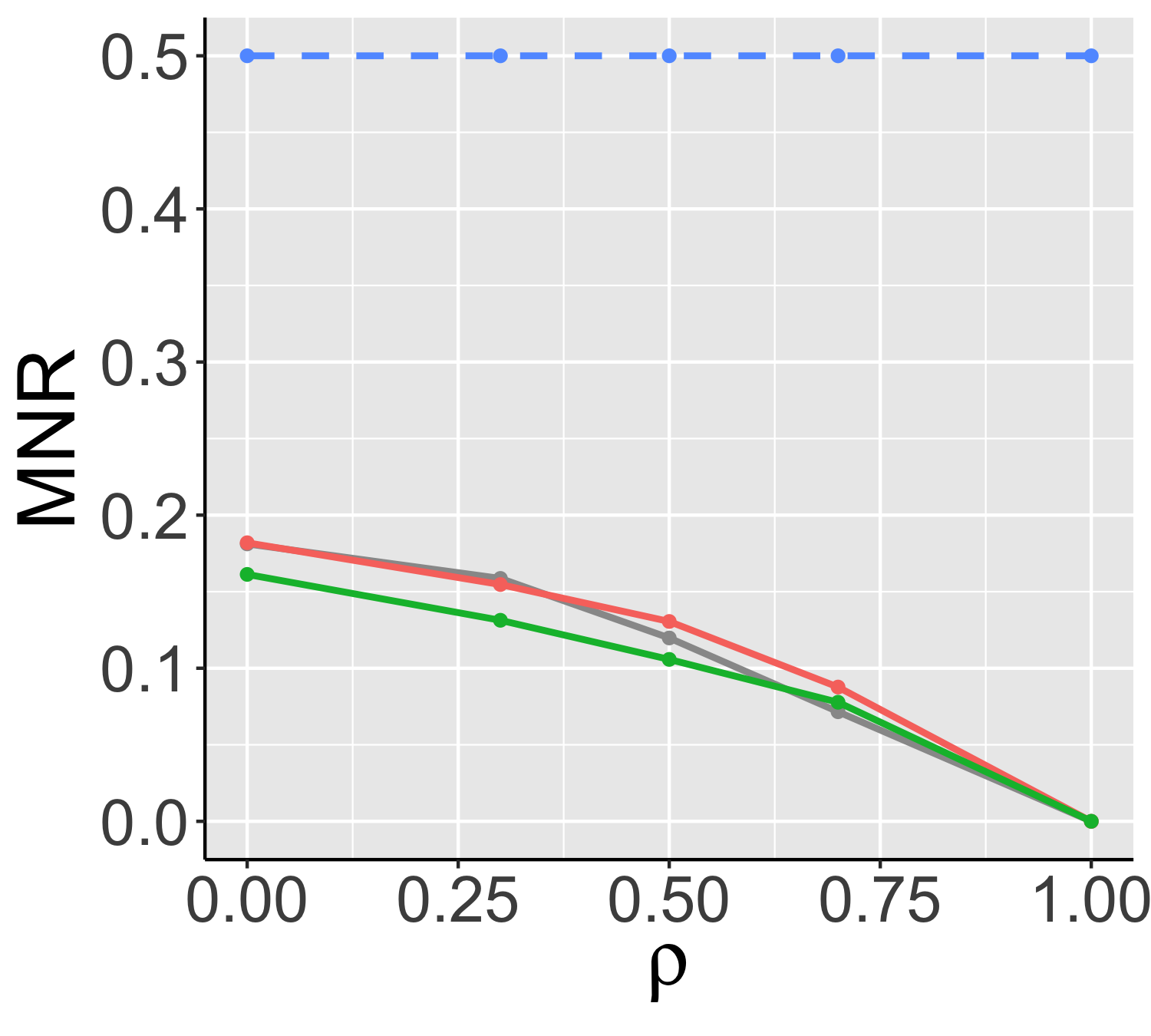}}
\caption{Performance of our algorithm for pairs of
  $\rho$-$\mathrm{SBM}$ graphs on $n = 300$ vertices. The mean
  reciprocal rank (MRR) and mean normalized rank (MNR) are computed
  based on $500$ Monte Carlo replicates. The MRR and MNR are plotted
  for different values of the sparsity parameter $\rho$.
  The grey line corresponds to the method in \cite{agterberg2020vertex}.}
\label{rho-SBM_compare}
\end{figure}

\begin{figure}[h!]
\centering
\subfigure[our method with orthogonal Procrustes]{\includegraphics[width=0.49\textwidth]{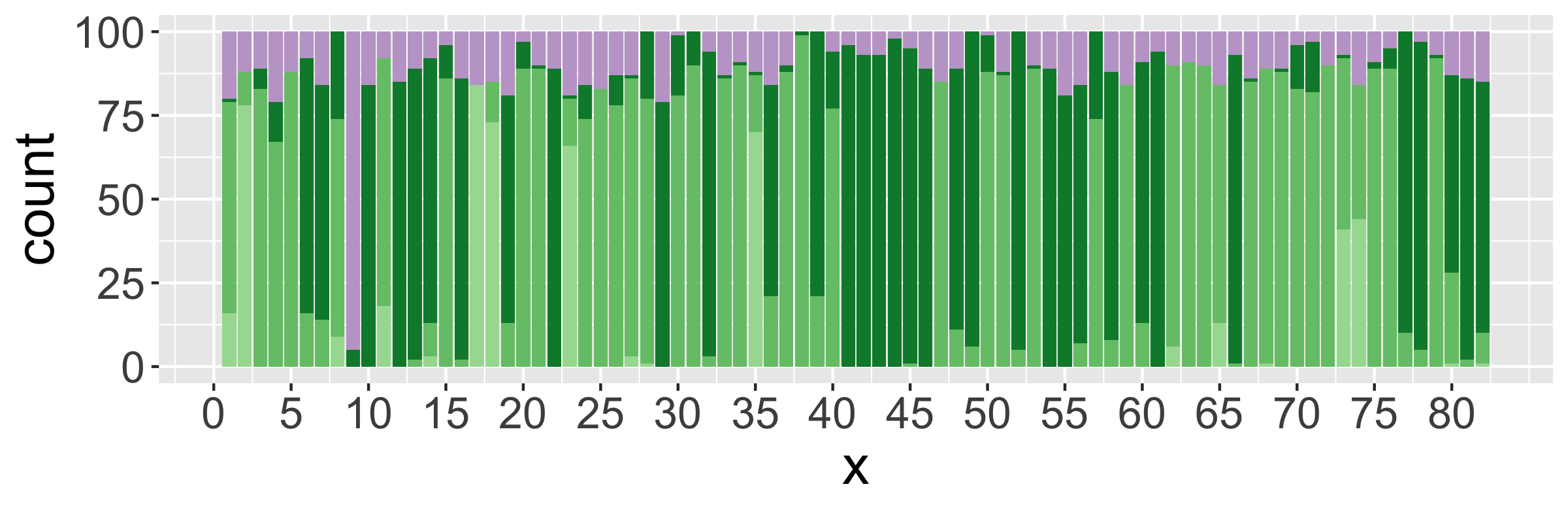}}
\subfigure[our method with point set registration]{\includegraphics[width=0.49\textwidth]{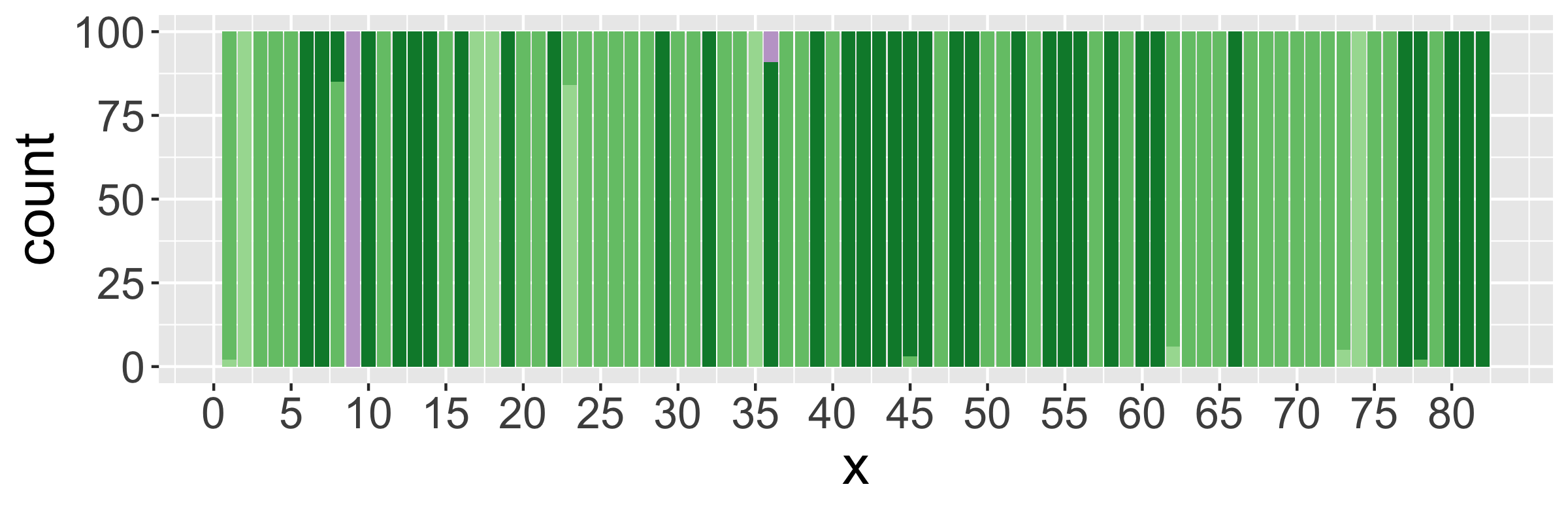}}

\subfigure[method in \cite{agterberg2020vertex}]{\includegraphics[width=0.49\textwidth]{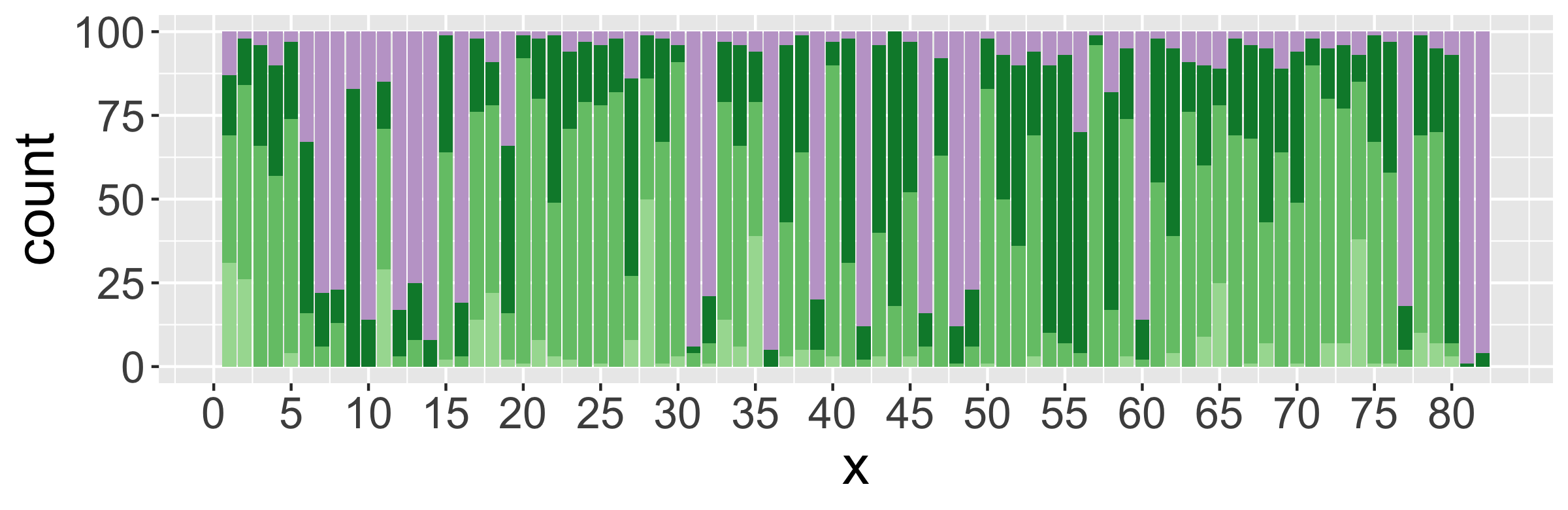}}
\caption{Performance of our algorithm and the method in \cite{agterberg2020vertex} for vertex nomination between the two high-school networks. Here we consider only the subgraphs induced by the $82$ shared vertices. For our methods, $2$ random seeds are also used for the quadratic programming step.}

\label{highschool_shared_Algo1_count_compare}
\end{figure}

\begin{table}[h!]
\centering
\begin{tabular}{c|cccccccc}
 & 1\% & 5\%& 10\%&25\%& 50\%& 75\%&95\%& 99\% \\
 \hline
Procrustes (10 seeds) &  0.003 & 0.013 & 0.030 & 0.074 & 0.196 & 0.387 & 0.750 & 0.870   \\
set registration (no seeds) & 0.002 & 0.013 & 0.025 & 0.073 & 0.196 & 0.386 & 0.757 & 0.876\\
\cite{agterberg2020vertex} (10 seeds) &  0.003 & 0.009 &0.024 &0.076 &0.214 &0.458 &0.744& 0.843\\  
\end{tabular}
\caption{Quantile levels of normalized rank (NR) values for vertex nomination
  with the Bing entity networks on $n = 1000$ vertices}
\label{tab:Bing_compare}
\end{table}

\clearpage
\subsection{Results of larger $n$}

\label{sec:n=1000}

\begin{figure}[h!]
\centering
\subfigure{\includegraphics[width=0.35\textwidth]{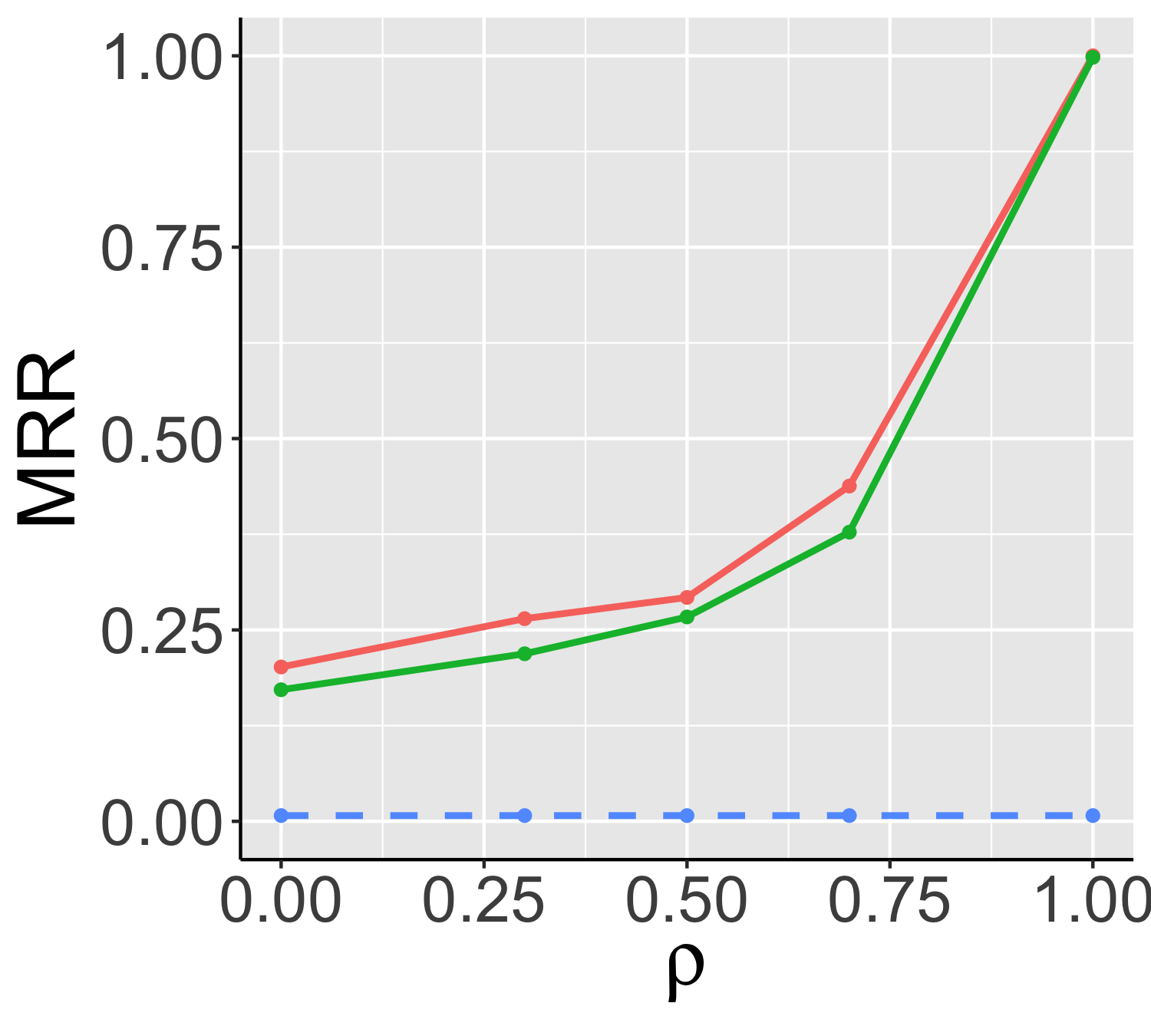}}
\subfigure{\includegraphics[width=0.35\textwidth]{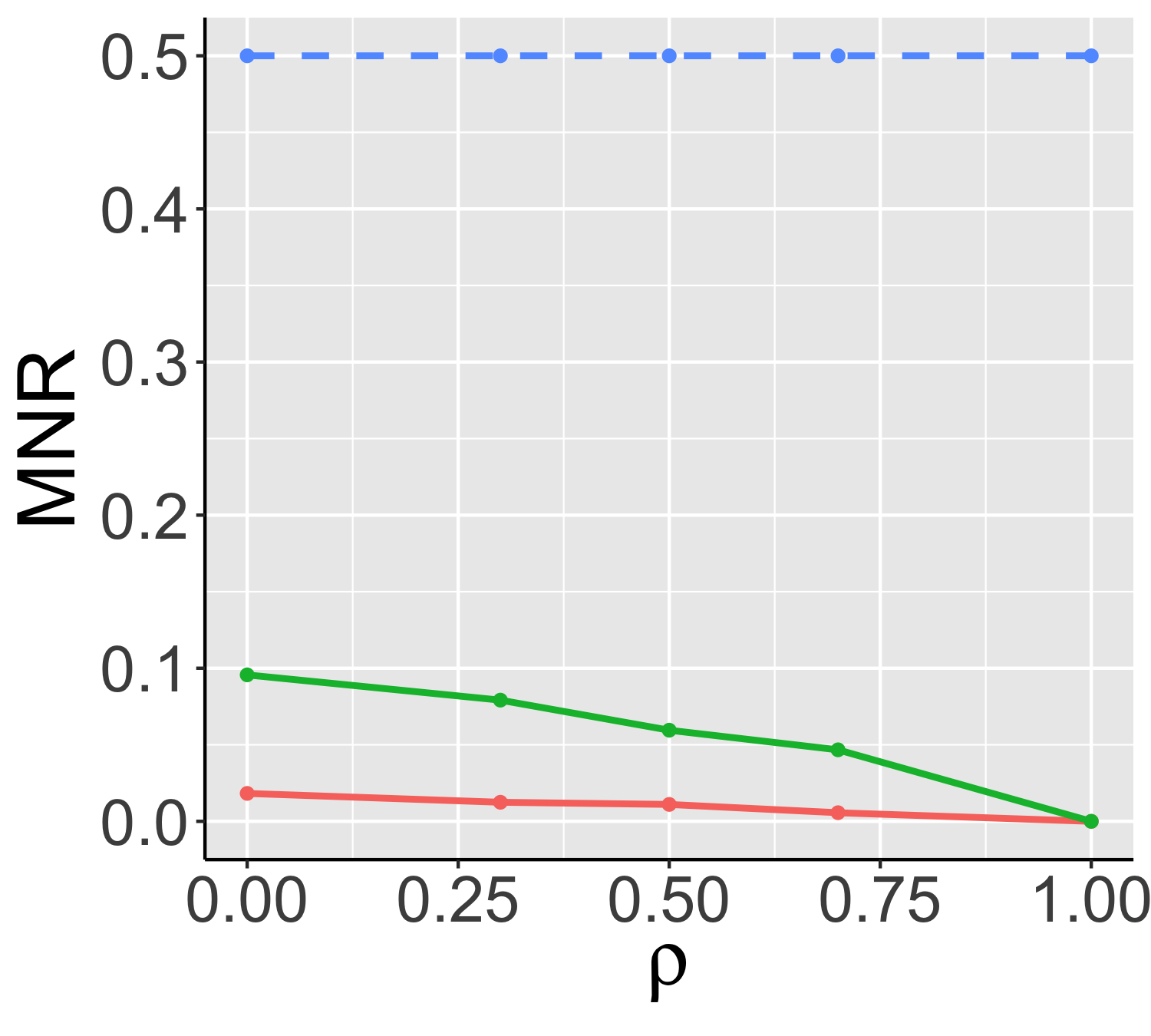}}
\caption{Performance of our algorithm for pairs of
  $\rho$-$\mathrm{RDPG}$ graphs on $n = 1000$ vertices. The mean
  reciprocal rank (MRR) and mean normalized rank (MNR) are computed
  based on $500$ Monte Carlo replicates. The MRR and MNR are plotted
  for different values of the correlation parameter $\rho$.}
\label{rho-RDPG_n=1000}
\end{figure}

\begin{figure}[h!]
\centering
\subfigure{\includegraphics[width=0.35\textwidth]{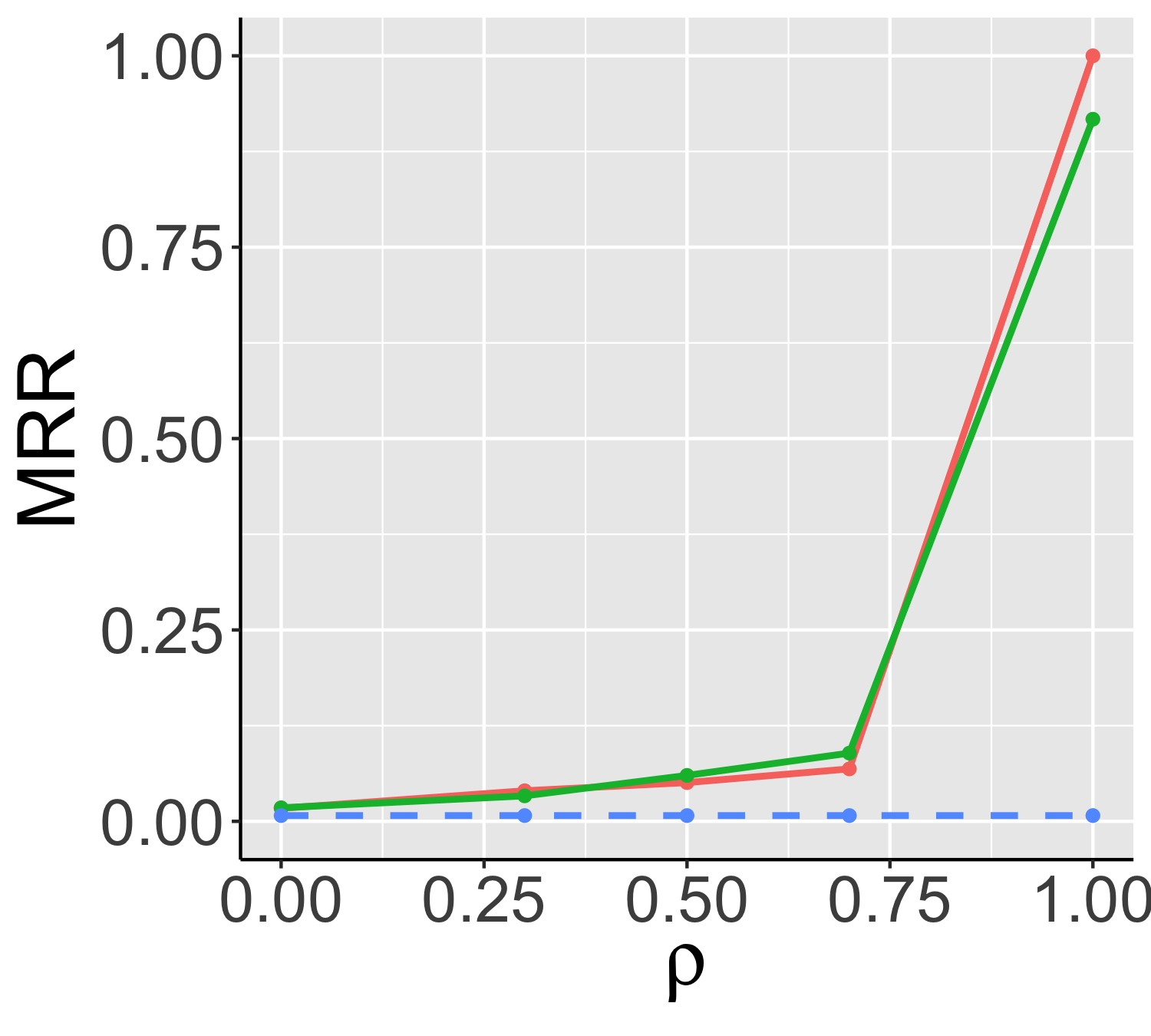}}
\subfigure{\includegraphics[width=0.35\textwidth]{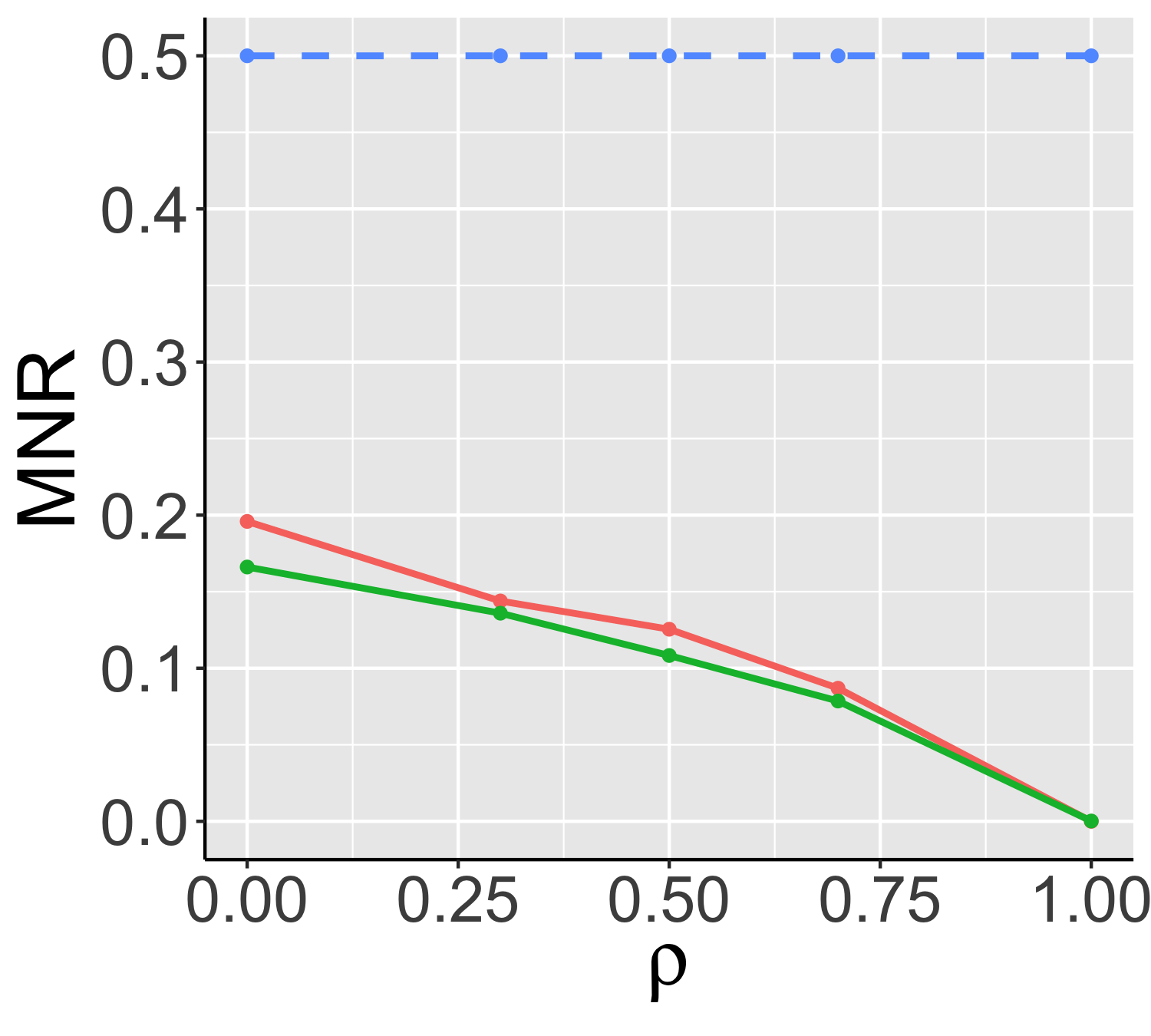}}
\caption{Performance of our algorithm for pairs of
  $\rho$-$\mathrm{SBM}$ graphs on $n = 1000$ vertices. The mean
  reciprocal rank (MRR) and mean normalized rank (MNR) are computed
  based on $500$ Monte Carlo replicates. The MRR and MNR are plotted
  for different values of the correlation parameter $\rho$.}
\label{rho-SBM_n=1000}
\end{figure}

\end{document}